\PassOptionsToPackage{unicode}{hyperref}
\PassOptionsToPackage{hyphens}{url}
\PassOptionsToPackage{dvipsnames,svgnames,x11names}{xcolor}
\documentclass[
  12pt]{article}

\usepackage{amsthm}
\usepackage{amsmath,amssymb}
\usepackage{multirow}
\usepackage[ruled,noend]{algorithm2e}
\usepackage{algpseudocode}
\usepackage[caption=false,position=top]{subfig}
\usepackage{iftex}
\ifPDFTeX
  \usepackage[T1]{fontenc}
  \usepackage[utf8]{inputenc}
  \usepackage{textcomp} 
\else 
  \usepackage{unicode-math}
  \defaultfontfeatures{Scale=MatchLowercase}
  \defaultfontfeatures[\rmfamily]{Ligatures=TeX,Scale=1}
\fi
\usepackage{lmodern}
\ifPDFTeX\else  
\fi
\IfFileExists{upquote.sty}{\usepackage{upquote}}{}
\IfFileExists{microtype.sty}{
  \usepackage[]{microtype}
  \UseMicrotypeSet[protrusion]{basicmath} 
}{}
\makeatletter
\@ifundefined{KOMAClassName}{
  \IfFileExists{parskip.sty}{%
    \usepackage{parskip}
  }{
    \setlength{\parindent}{0pt}
    \setlength{\parskip}{6pt plus 2pt minus 1pt}}
}{
  \KOMAoptions{parskip=half}}
\makeatother
\usepackage{xcolor}
\makeatletter
\ifx\paragraph\undefined\else
  \let\oldparagraph\paragraph
  \renewcommand{\paragraph}{
    \@ifstar
      \xxxParagraphStar
      \xxxParagraphNoStar
  }
  \newcommand{\xxxParagraphStar}[1]{\oldparagraph*{#1}\mbox{}}
  \newcommand{\xxxParagraphNoStar}[1]{\oldparagraph{#1}\mbox{}}
\fi
\ifx\subparagraph\undefined\else
  \let\oldsubparagraph\subparagraph
  \renewcommand{\subparagraph}{
    \@ifstar
      \xxxSubParagraphStar
      \xxxSubParagraphNoStar
  }
  \newcommand{\xxxSubParagraphStar}[1]{\oldsubparagraph*{#1}\mbox{}}
  \newcommand{\xxxSubParagraphNoStar}[1]{\oldsubparagraph{#1}\mbox{}}
\fi
\makeatother

\usepackage{longtable,booktabs,array}
\usepackage{calc} 
\usepackage{etoolbox}
\makeatletter
\patchcmd\longtable{\par}{\if@noskipsec\mbox{}\fi\par}{}{}
\makeatother
\IfFileExists{footnotehyper.sty}{\usepackage{footnotehyper}}{\usepackage{footnote}}
\makesavenoteenv{longtable}
\usepackage{graphicx}
\makeatletter
\def\maxwidth{\ifdim\Gin@nat@width>\linewidth\linewidth\else\Gin@nat@width\fi}
\def\maxheight{\ifdim\Gin@nat@height>\textheight\textheight\else\Gin@nat@height\fi}
\makeatother
\setkeys{Gin}{width=\maxwidth,height=\maxheight,keepaspectratio}
\makeatletter
\def\fps@figure{htbp}
\makeatother

\makeatletter
\@ifpackageloaded{caption}{}{\usepackage{caption}}
\AtBeginDocument{%
\ifdefined\contentsname
  \renewcommand*\contentsname{Table of contents}
\else
  \newcommand\contentsname{Table of contents}
\fi
\ifdefined\listfigurename
  \renewcommand*\listfigurename{List of Figures}
\else
  \newcommand\listfigurename{List of Figures}
\fi
\ifdefined\listtablename
  \renewcommand*\listtablename{List of Tables}
\else
  \newcommand\listtablename{List of Tables}
\fi
\ifdefined\figurename
  \renewcommand*\figurename{Figure}
\else
  \newcommand\figurename{Figure}
\fi
\ifdefined\tablename
  \renewcommand*\tablename{Table}
\else
  \newcommand\tablename{Table}
\fi
}
\@ifpackageloaded{float}{}{\usepackage{float}}
\floatstyle{ruled}
\@ifundefined{c@chapter}{\newfloat{codelisting}{h}{lop}}{\newfloat{codelisting}{h}{lop}[chapter]}
\floatname{codelisting}{Listing}

\makeatother
\makeatletter
\@ifpackageloaded{caption}{}{\usepackage{caption}}
\@ifpackageloaded{subcaption}{}{\usepackage{subcaption}}
\makeatother

\ifLuaTeX
  \usepackage{selnolig}  
\fi
\usepackage[]{natbib}
\usepackage{bookmark}

\IfFileExists{xurl.sty}{\usepackage{xurl}}{} 
\hypersetup{
  pdftitle={Title},
  pdfauthor={Author 1; Author 2},
  pdfkeywords={3 to 6 keywords, that do not appear in the title},
  colorlinks=true,
  linkcolor={blue},
  filecolor={Maroon},
  citecolor={Blue},
  urlcolor={Blue},
  pdfcreator={LaTeX via pandoc}}

\newtheorem{assu}{Assumption}
\newtheorem{thm}{Theorem}

\newcommand{\anon}{1}

\begin{document}

\def\spacingset#1{\renewcommand{\baselinestretch}%
{#1}\small\normalsize} \spacingset{1}


\if1\anon
{
  \title{\bf Learn-As-you-GO (LAGO) Trials: Optimizing Trials for Effectiveness and Power to Prevent Failed Trials}
  \author{Ante Bing\thanks{
    All authors were supported by NIH Grant R01 HL167936.\;\;\;\;\;\;}\hspace{.2cm}\\
    Department of Mathematics and Statistics, Boston University\\
    Donna Spiegelman\\
    Department of Biostatistics, Yale University\\ 
    Judith J. Lok\\
    Department of Mathematics and Statistics, Boston University
    }
  \maketitle
} \fi

\if0\anon
{
  \bigskip
  \bigskip
  \bigskip
  \begin{center}
    {\LARGE\bf Title}
\end{center}
  \medskip
} \fi

\bigskip
\begin{abstract}
The Learn-As-you-GO (LAGO) design provides a rigorous framework for adapting the intervention package based on accumulating data while the trial is ongoing. This article improves the flexibility of the LAGO design by incorporating statistical power as an optimization criterion (power goal) in LAGO optimizations. We propose the unconditional and conditional power approaches to add a power goal. Both approaches estimate the power at the end of the LAGO trial using data from prior stages, and increase the power at the end of the LAGO trial when the original trial was underpowered. Including a power goal maintains the asymptotic properties of the estimators of the treatment effect while preserving the asymptotic level of the statistical test at the end of the trial. We illustrate the benefits of our methods through a retrospective application to the BetterBirth Study, a large-scale study of maternal-newborn care that failed to show a significant effect on its primary outcome. This analysis demonstrates how our methods could have led to more intensive interventions and potentially significant results. The LAGO design with power goal optimizations provides investigators with a powerful tool to reduce the risk of failed trials due to insufficient power. 
\end{abstract}

\noindent%
\vfill

\newpage
\spacingset{1.8}

\section{Introduction}\label{sec-intro}

Adaptive trial designs offer flexibility by allowing pre-planned modifications based on accumulating data \citep{Huskins2018Adaptive}. Among adaptive trial designs, the Learn-As-you-GO (LAGO) design stands out for its capability to optimize an intervention package based on accumulating data while the trial is ongoing \citep{nevo2021analysis, bing2023learnasyougo}. LAGO trials are typically multi-stage, multi-component intervention studies. After each stage, the results collected up to that stage are analyzed, the intervention package is reassessed, and a revised intervention package is implemented in the subsequent stage. By iteratively optimizing the intervention package (LAGO optimizations), the LAGO design aims to prevent trial failure while reducing costs.

Current adaptive trial designs allow changes to randomization and dropping treatment arms \citep{FDAadaptiveNew}, but cannot modify intervention components based on interim results. This limitation may explain why some large trials have not shown significant effects \citep{stensland2014adult, semrau2017outcomes, fogel2018factors}. In practice, investigators frequently adapt evidence-based interventions to better suit local contexts and populations \citep{escoffery2018systematic, movsisyan2019adapting}. Such adaptations are often made ad hoc, without a systematic approach for determining when and how interventions are modified \citep{aarons2012dynamic}. The LAGO design formalizes the conditions under which interventions are adapted during an ongoing trial, allowing investigators to systematically optimize interventions to prevent trial failure while maintaining the integrity of the trial, and preserving the Type I error rate of the primary hypothesis test \citep{nevo2021analysis, bing2023learnasyougo}. The LAGO design also aligns with recent calls in implementation science for developing ongoing learning systems for intervention adaptation \citep{chambers2023advancing}.

The LAGO design has been compared to other designs such as the Multiphase Optimization Strategy (MOST) \citep{collins2007multiphase} and Sequential Multiple Assignment Randomized Trials (SMART) \citep{murphy2005experimental}. MOST and SMART differ significantly from the LAGO design. MOST is a three-phase design: a formative phase that does not involve quantitative data collection, an optimization phase that uses a factorial experiment to select intervention package components at fixed doses in relation to short-term surrogate outcomes, and a third phase consisting of a full-scale trial of the intervention package determined in phase 2, with no further adaptations. MOST uses data from only phase 3 for estimation and inference. MOST phase 2 can be paired with LAGO in phase 3 to foster continued learning after the trial launches, provided that the trial in phase 3 has multiple stages. SMART is a study design methodology for identifying dynamic treatment regimens in which fixed rules determine the next step in a patient's treatment depending on that particular patient's own outcomes to date. In contrast, LAGO uses adaptive methods to identify a complex, multi-component, multi-level intervention package that does not vary over time.

Initially developed for binary outcomes under a logistic regression model \citep{nevo2021analysis}, the LAGO design has since been extended to accommodate continuous outcomes under a flexible conditional mean model \citep{bing2023learnasyougo}. 
However, neither \citet{nevo2021analysis} nor \citet{bing2023learnasyougo} has directly incorporated statistical power into the LAGO design.

LAGO trials consist of $K>1$ stages. To distinguish between the recommended interventions and the optimal interventions, note that at the end of each stage $k$, where $k=1,\ldots,K-1$, the recommended intervention for stage $k+1$ is calculated using data from stages $1$ to $k$. At the end of the LAGO trial, the optimal intervention is estimated based on the data from all $K$ stages. LAGO optimization refers to the process of calculating the recommended interventions after each stage and estimating the optimal intervention in the final analysis. LAGO optimizations allow for various optimization criteria. For example, \citet{nevo2021analysis} considered an optimization criterion based on an outcome goal in LAGO trials with binary outcomes, where the intervention component composition was estimated to achieve a pre-specified success probability while minimizing cost. \citet{bing2023learnasyougo} considered another optimization criterion based on an outcome goal in LAGO trials with continuous outcomes, where the intervention component composition was estimated to achieve a pre-specified mean while minimizing cost.
LAGO optimizations based on an outcome goal, as described in previous work on LAGO, help to prevent failed trials by increasing the success probability or the outcome mean, but do not directly incorporate power.

This article considers LAGO optimizations with an additional power goal. The recommended interventions are calculated using an optimization criterion based on both an outcome goal and a power goal, while the optimal intervention at the end of the trial is estimated using an optimization criterion based only on an outcome goal. Incorporating a power goal into the LAGO optimizations presents challenges because power depends on the pre-specified statistical test at the end of the trial, and only partially observed data for the selected test are available during the LAGO optimizations. Thus, new methods are required to adapt the intervention package composition using partially observed data. We explore both unconditional and conditional power approaches to incorporate a power goal into the LAGO optimizations, with further details provided in Section~\ref{power goal section}. Section~\ref{level section} shows that regardless of the power goal approach, incorporating an additional power goal directly into the LAGO design preserves the asymptotic level of the pre-specified test at the end of the trial.

The LAGO design, incorporating both an outcome goal and a power goal, is illustrated through the BetterBirth Study, a costly failed trial of maternal and newborn care, involving approximately 160,000 mothers and children in Uttar Pradesh, India \citep{hirschhorn2015learning, semrau2017outcomes}. 
Despite its scale, the BetterBirth Study did not find a significant reduction in its primary composite outcome, which consisted of stillbirth, early neonatal death, maternal death, and self-reported maternal severe complications within 7 days after delivery \citep{semrau2017outcomes}. 
The BetterBirth Study consisted of three stages, with stage 1 having insufficient sample size for a LAGO optimization.
Section \ref{application} retrospectively applies the LAGO design with an outcome goal and a power goal to calculate the stage 3 recommended intervention that aims to increase the outcome to a pre-specified threshold and might have prevented a failed trial while minimizing cost. While the BetterBirth Study serves as a hypothetical example, the LAGO design is not merely a theoretical concept. Real-world studies are currently being optimized using the LAGO design, including PULESA-Uganda \citep{PULESA}, TASKPEN-Zambia \citep{herce2024evaluating}, MAP-IT-Nigeria \citep{aifah2023study}, and HPTN 096 \citep{HPTN096}.

This article is organized as follows.
Section~\ref{setting} provides an overview of LAGO trials. 
Section~\ref{power goal section} discusses the constraints imposed by the two versions of the power goal and how the new constraints affect the recommended interventions.
Section~\ref{asymptotic properties section} describes asymptotic properties of the estimator for the treatment effect in the final analysis based on data from all stages.
Section~\ref{level section} shows that incorporating a power goal in LAGO optimizations preserves the asymptotic level of the pre-specified test for two types of tests at the end of the trial.
Section~\ref{simulation} presents simulations evaluating the performance of the LAGO design with both an outcome goal and a power goal. 
Section~\ref{application} retrospectively applies the LAGO design with both an outcome goal and a power goal to the BetterBirth Study. 
Section~\ref{Discussion} addresses implications, limitations, and future directions for the LAGO design.

\section{Overview of LAGO Trials}\label{setting}
LAGO trials are typically multi-stage, multi-component intervention studies. A pre-specified number of centers, denoted by $J^{(k)}$, are enrolled in each stage $k$, where $k=1,\ldots,K$. Centers are randomized to either the intervention or control group, with each center participating in only one stage.  
For simplicity,  the main text focuses on the case where $K=2$, with details for $K>2$ provided in Appendix Section G. 

In LAGO trials, the intervention package, denoted by $\boldsymbol{x}$ (or $\boldsymbol{x}_j$ for center-specific intervention packages), consists of $P$ unique components. 
The cost of implementing $\boldsymbol{x}$ is calculated using a pre-specified cost function $C(\boldsymbol{x})$, where $C(\boldsymbol{x})$ is determined by subject-matter considerations, and typically represents the total monetary cost of the intervention package.
For example, a linear cost function with $P=2$ could be $C(\boldsymbol{x}) = 4x_1 + x_2$.
Let $\boldsymbol{\mathcal{U}}=(\mathcal{U}_1, ..., \mathcal{U}_P)$ and $\boldsymbol{\mathcal{L}}=(\mathcal{L}_1, ..., \mathcal{L}_P)$ denote the upper and lower bounds of the $P$ components of the intervention package $\boldsymbol{x}$.
Let $n_j^{(k)}$ denote the number of participants in center $j$ of stage $k$, where $j=1,\dots,J^{(k)}$.
Let $n^{(k)}$ denote the total number of participants in stage $k$, and let $n$ denote the total number of participants across all stages.
Let $Y_{ij}^{(k)}$ denote the outcome for participant $i$ in center $j$ of stage~$k$.

In a two-stage LAGO trial, stage 1 could implement either a standard design, such as a fractional factorial design, or a pre-specified intervention package $\boldsymbol{x}^{(1)}$, based on either the research team's educated guesses or calculations from pre-trial information. Appendix Section E describes the calculation of an optimal $\boldsymbol{x}^{(1)}$ from pre-trial information.
Let $\boldsymbol{a}_j^{(1)}$ be the actual intervention package implemented in center $j$ of stage 1. 
To properly estimate the treatment effect parameters, sufficient variation in the $\boldsymbol{a}_j^{(1)}$ values is required. Such variation could come from the study design, unplanned variation in the interventions, or both.
In settings without unplanned variation in the interventions, the actual intervention is the same as the recommended intervention, or $\boldsymbol{a}^{(1)} = \boldsymbol{x}^{(1)}$.
In large-scale intervention studies, centers may not adhere strictly to the recommended intervention. In such cases, it is assumed that $\boldsymbol{a}_j^{(1)} = h_j^{(1)} ( \boldsymbol{x}^{(1)} )$, where $h_j^{(1)}$ is a center-specific continuous deterministic function for each center $j$ in stage 1. The functions $h_j^{(1)}$ are typically unknown, and only the $\boldsymbol{a}_j^{(1)}$ are observed. 
Assuming sufficient variation in the interventions in stage 1,
either a logistic regression model or a generalized linear model (GLM) is fitted for the outcome of interest at the end of stage 1, depending on the outcome type. The outcome models are defined in Assumptions \ref{logistic assumption} and \ref{glm assumption} in Section \ref{modeling the outcome in LAGO}, respectively. 
Let $\hat{\boldsymbol{x}}^{(2, n^{(1)})}$ be the stage 2 recommended intervention. The superscript in $\hat{\boldsymbol{x}}^{(2, n^{(1)})}$ stresses that $\hat{\boldsymbol{x}}^{(2, n^{(1)})}$ is calculated using data from all $n^{(1)}$ participants in stage 1. Section~\ref{stage 1 of LAGO Trials} describes the details of calculating the stage 2 recommended intervention $\hat{\boldsymbol{x}}^{(2, n^{(1)})}$ based on the stage 1 data.

Once $\hat{\boldsymbol{x}}^{(2, n^{(1)})}$ is calculated, the LAGO trial continues to stage 2.
Let $\boldsymbol{A}_{j}^{(2, n^{(1)})}$ be the actual intervention package implemented at center $j$ in stage 2.
As in stage 1, when there is no unplanned variation in the interventions, $\boldsymbol{A}^{(2, n^{(1)})} = \hat{\boldsymbol{x}}^{(2, n^{(1)})}$. With unplanned variation in the interventions,
$\boldsymbol{A}_{j}^{(2, n^{(1)})} =  h_j^{(2)} (\hat{\boldsymbol{x}}^{(2, n^{(1)})})$, where $h_j^{(2)}$ is a center-specific continuous deterministic function for each center $j$ in stage 2.
Following stage 2, the final analysis estimates the optimal intervention, $\hat{\boldsymbol{x}}^{opt}$, based on data from both stage 1 and stage 2, where optimality is defined as achieving a pre-specified success probability/outcome mean while minimizing cost.

\subsection{Modeling the Outcome in LAGO Trials}\label{modeling the outcome in LAGO}
In LAGO trials with binary outcomes, consistency and asymptotic normality of the final estimator for the treatment effect, based on data from all stages combined, have been proven when using a logistic regression model \citep{nevo2021analysis}, and for LAGO trials with continuous outcomes, using a more flexible GLM \citep{bing2023learnasyougo}. 
Let $Y_{ij}$ denote the outcome for participant $i$ in center $j$.
Binary and continuous outcomes are modeled differently in this article.
\begin{assu}\label{logistic assumption}
In LAGO trials with binary outcomes, the probability of success for an individual $i$ in center $j$, given the recommended intervention package $\boldsymbol{X}=\boldsymbol{x}$ and the actual intervention packages $\boldsymbol{A}_j=\boldsymbol{a}_j$, follows a logistic regression model: 
$$
\operatorname{logit}\left(\operatorname{pr}\left(Y_{i j}=1 \mid \boldsymbol{A}_j=\right. \left.\boldsymbol{a}_j, \boldsymbol{X}=\boldsymbol{x}; \boldsymbol{\beta}\right)\right)
= 
\operatorname{logit}(p_{\boldsymbol{a}_j}\left( \boldsymbol{\beta}\right))
= 
\beta_0+\boldsymbol{\beta}_1^T \boldsymbol{a}_j.
$$
$\boldsymbol{\beta}^T=({\beta}_{0}, \boldsymbol{\beta}_{1}^T )$ are unknown parameters to be estimated. 
Let $\boldsymbol{\beta}^*$ denote the true parameter vector $\boldsymbol{\beta}$.
This logistic regression model assumes that the success probability only depends on the recommended intervention $\boldsymbol{x}$ through the actual interventions $\boldsymbol{a}_j$. 
\end{assu}

\begin{assu}\label{glm assumption}
In LAGO trials with continuous outcomes, 
the expected outcome for an individual $i$ in center $j$, given the recommended intervention package $\boldsymbol{X}=\boldsymbol{x}$ and the actual intervention packages $\boldsymbol{A}_j=\boldsymbol{a}_j$, follows a GLM:
$$
g\left(E\left(Y_{ij} | \boldsymbol{A}_j=\boldsymbol{a}_j, \boldsymbol{X}=\boldsymbol{x}; \boldsymbol{\beta} \right)\right) = \beta_{0}+\boldsymbol{\beta}_{1}^{T} \boldsymbol{a}_j,
$$
where $g()$ is a twice continuously differentiable link function. This GLM assumes that the expected outcome only depends on the recommended intervention $\boldsymbol{x}$ through the actual interventions $\boldsymbol{a}_j$. 
\end{assu}

\subsection{LAGO Optimizations Based on Outcome and Power Goals}\label{stage 1 of LAGO Trials}
Section~\ref{stage 1 of LAGO Trials} first provides an overview of LAGO optimizations based on an outcome goal, which was previously considered in \citet{nevo2021analysis} and \citet{bing2023learnasyougo}, and then explains how to modify the LAGO optimizations to add an additional power goal.

In LAGO trials, the pre-specified outcome goal is a minimal success probability ($\Tilde{p}$) for binary outcomes, or a minimal acceptable mean ($\Tilde{\mu}$) for continuous outcomes.
Assuming that a positive effect of the overall intervention package on the outcome is expected, the optimal interventions satisfy the following:
{\small
\begin{equation}\label{optimization no hat}
\text{Min}_{\boldsymbol{x}} C\left(\boldsymbol{x}\right)
\; \text {subject to the constraint that} \;
\left\{
\begin{array}{ll}
{p}_{\boldsymbol{x}}\left( {\boldsymbol{\beta}}\right) \geq \tilde{p} & \text{for binary outcomes,} \\
{E}\left(Y_{ij} | \boldsymbol{x}; {\boldsymbol{\beta}} \right) \geq \tilde{\mu} &\text{for continuous outcomes.}
\end{array}
\right. 
\;
\end{equation} 
}
Let $\hat{\boldsymbol{\beta}}^{(1)}$ denote the stage 1-based estimate of $\boldsymbol{\beta}$.
In practice, the stage 2 recommended intervention, $\hat{\boldsymbol{x}}^{(2, n^{(1)})}$, is obtained by solving
{\small
\begin{equation}\label{optimization}
\text{Min}_{\boldsymbol{x}} C\left(\boldsymbol{x}\right)
\; \text {subject to the constraint that} \;
\left\{
\begin{array}{ll}
{p}_{\boldsymbol{x}}\bigl( \hat{\boldsymbol{\beta}}^{(1)}\bigr) \geq \tilde{p} \;\; \text{for binary outcomes,} \\
{E}\bigl(Y_{ij} | \boldsymbol{x}; \hat{\boldsymbol{\beta}}^{(1)} \bigr) \geq \tilde{\mu} \;\text{for continuous outcomes.}
\end{array}
\right. 
\;
\end{equation} 
}

To incorporate a power goal, the LAGO optimization additionally considers the form of the pre-specified statistical test used at the end of the LAGO trial. 
Let $\Pi$ be a pre-specified power goal. 
Assuming that the LAGO optimizations consider both an outcome goal and a power goal, 
$\hat{\boldsymbol{x}}^{(2, n^{(1)})}$ is obtained by solving
\vspace{-0.2cm}
{\small
\begin{equation}\label{optimization two constraints}
\text{Min}_{\boldsymbol{x}} C\left(\boldsymbol{x}\right)
\; \text {subject to the constraint that} 
\begin{cases}
&{p}_{\boldsymbol{x}}\bigl( \hat{\boldsymbol{\beta}}^{(1)} \bigr) \geq \tilde{p} \;\; 
\text{or} \;\; {E}\bigl(Y_{ij} | \boldsymbol{x}; \hat{\boldsymbol{\beta}}^{(1)} \bigr) \geq \tilde{\mu} \;\;
, \\
&{Power}\bigl(\boldsymbol{x}; \hat{\boldsymbol{\beta}}^{(1)} \bigr) \geq \Pi.
\end{cases}
\end{equation}
}
\vspace{-0.1cm}
\noindent The power in equation (\ref{optimization two constraints}) refers to the power of the pre-specified test at the end of the trial that assesses the null hypothesis that the intervention package has no effect on the outcome of interest, regardless of variations in implementation. Most of this article focuses on equation (\ref{optimization two constraints}). Calculating $\hat{\boldsymbol{x}}^{(2, n^{(1)})}$ based on the power goal from equation (\ref{optimization two constraints}) is non-trivial, since the power depends on the test used, and only partial data for such test is observed at the end of stage 1. Section~\ref{power goal section} considers the power goal of the two-sample z-test for the difference between two proportions as a concrete example, and describes in detail both unconditional and conditional power approaches for calculating $\hat{\boldsymbol{x}}^{(2, n^{(1)})}$.
Appendix Section D presents examples of power goals of other common tests, including the two-sample t-test for the difference between two means.

Other LAGO optimizations can be considered. For example, in applications where the budget of the intervention package is fixed, the recommended intervention could aim to maximize power while requiring that the intervention package leads to a specified outcome goal, all within a fixed budget $c$. That is,
\vspace{-0.2cm}
{\small
\begin{equation*}
\text{Max}_{\boldsymbol{x}} \;{Power}\bigl(\boldsymbol{x}; \hat{\boldsymbol{\beta}}^{(1)} \bigr)
\; \text {subject to the constraint that} 
\begin{cases}
&{p}_{\boldsymbol{x}}\bigl( \hat{\boldsymbol{\beta}}^{(1)} \bigr) \geq \tilde{p} \;\; 
\text{or} \;\; {E}\bigl(Y_{ij} | \boldsymbol{x}; \hat{\boldsymbol{\beta}}^{(1)} \bigr) \geq \tilde{\mu} \;\;
, \\
&C\left(\boldsymbol{x}\right) \leq c.
\end{cases}
\end{equation*}
}

\vspace{-0.7cm}
Let $\hat{\boldsymbol{\beta}}^{(2)}$ be the estimate of $\boldsymbol{\beta}$ based on all data from stage 1 and stage 2 combined.
At the end of the LAGO trial, the optimal intervention, $\hat{\boldsymbol{x}}^{opt}$, is estimated based only on an outcome goal and solves
\small{
\begin{equation}\label{optimization opt int}
\text{Min}_{\boldsymbol{x}} C\left(\boldsymbol{x}\right)
\; \text {subject to the constraint that} \;
\left\{
\begin{array}{ll}
{p}_{\boldsymbol{x}}\bigl( \hat{\boldsymbol{\beta}}^{(2)}\bigr) \geq \tilde{p} \;\; \text{for binary outcomes,} \\
{E}\bigl(Y_{ij} | \boldsymbol{x}; \hat{\boldsymbol{\beta}}^{(2)} \bigr) \geq \tilde{\mu} \;\text{for continuous outcomes.}
\end{array}
\right. 
\;
\end{equation} 
}

\normalsize 
\vspace{-0.5cm}
\subsection{Hypothesis Testing in LAGO Trials}\label{hypothesis testing section}
A main objective of LAGO trials is to test the null hypothesis that the intervention package has no effect on the outcome of interest, regardless of variations in implementation. 
There are two ways to test this null hypothesis. One approach is to carry out a 1 degree-of-freedom (1-df) test, such as a t-test or z-test, comparing the difference in means or proportions between the intervention and control groups. The other approach is to carry out a $P$ degrees-of-freedom ($P$-df) test, such as a Wald test or likelihood ratio test, for the subvector of $\boldsymbol{\beta}$ characterizing the effect of the intervention.

The main text presents the details of conducting the 1-df two-sample z-test for the difference between two proportions in LAGO trials with binary outcomes. 
Let $S_1$ and $S_0$ be the sums of the outcomes in the intervention and control groups, respectively, in the final combined data from both stages.
Let $N_1$ and $N_0$ be the total sample sizes in the intervention and control groups in the final combined data from both stages, such that $\hat{p}_1 = S_1 / N_1$ and $\hat{p}_0 = S_0 / N_0$ are the estimated success probabilities in the intervention and control groups in the final combined data from both stages.
Let $S_1^{(k)}$ and $S_0^{(k)}$, and $n_1^{(k)}$ and $n_0^{(k)}$ be the sums of the outcomes and the total number of participants, respectively, in the intervention and control groups in stage $k$ $(k=1$ or $2)$.
Let $\alpha_{1}^{(k)} = \lim_{n \rightarrow \infty} n_1^{(k)}/N_1 > 0$ and $\alpha_{0}^{(k)} = \lim_{n \rightarrow \infty} n_0^{(k)}/N_0 > 0$, which assume that the ratios between the total sample sizes in the intervention and control groups in stage $k$ and the total sample sizes in the intervention and control groups in the final combined data from both stages converge to non-zero constants as $n$ approaches infinity. This leads to reasonable approximations if each stage contributes a substantial proportion of the total sample size.

\vspace{-0.2cm}
At the end of stage 2, the test statistic $Z$ of the two-sample z-test for the difference between two proportions with unpooled variance is
\begin{equation}\label{test statistic}
    Z = \frac{\hat{p}_1 - \hat{p}_0}{\sqrt{\frac{\hat{p}_1(1-\hat{p}_1)}{N_1} + \frac{\hat{p}_0(1-\hat{p}_0)}{N_0}}},
    \;\text{where} \;\;\hat{p}_1 = \frac{S_1^{(1)}+S_1^{(2)}}{N_1},
    \;\; \hat{p}_0 = \frac{S_0^{(1)}+S_0^{(2)}}{N_0}.
\end{equation}
Let $z_{\alpha}$ be the critical value corresponding to the upper $\alpha$ tail probability of a standard normal distribution. $H_0: p_1=p_0$ is rejected when $|Z|>z_{\alpha/2}$.
At $\alpha=0.05$, $z_{\alpha/2}\approx1.96$.

\subsection{Additional Notation and Assumptions}\label{notation and assumptions}
Recall that $\boldsymbol{a}_j^{(1)}$ is a $P\times1$ vector that represents the actual intervention package implemented in center $j$ of stage 1.
Let $\overline{\boldsymbol{a}}^{(1)}=\bigl(\boldsymbol{a}_{1}^{(1)}, \ldots, \boldsymbol{a}_{J^{(1)}}^{(1)}\bigr)$ be the actual interventions for all stage 1 centers.
Let $\boldsymbol{Y}_j^{(1)} = \bigl(Y_{1 j}^{(1)}, \ldots, Y_{n_{j}^{(1)} j}^{(1)} \bigr)$ be the observed outcomes for participants $1, \ldots, n_j^{(1)}$ in center $j$ of stage 1, and $\overline{\boldsymbol{Y}}^{(1)}=\bigl(\boldsymbol{Y}_{1}^{(1)}, \ldots, \boldsymbol{Y}_{J^{(1)}}^{(1)}\bigr)$ be the outcomes for all stage 1 centers.
The set of all stage 2 recommended interventions is denoted by
$\overline{\hat{\boldsymbol{x}}}^{\left(2, n^{(1)}\right)}=\bigl(\hat{\boldsymbol{x}}_{1}^{\left(2, n^{(1)}\right)}, \ldots, \hat{\boldsymbol{x}}_{J^{(2)}}^{\left(2, n^{(1)}\right)}\bigr)$. 
Let $\overline{\boldsymbol{A}}^{\left(2, n^{(1)}\right)}=\bigl(\boldsymbol{A}_{1}^{\left(2, n^{(1)}\right)}, \ldots, \boldsymbol{A}_{J^{(2)}}^{\left(2, n^{(1)}\right)}\bigr)$ be the actual interventions for all stage 2 centers. 
Let $Y_{ij}^{(2,n^{(1)})}$
be the observed outcome for participant $i$ in center $j$ of stage 2. 
Let $\boldsymbol{Y}_j^{(2,n^{(1)})} = \bigl(Y_{1 j}^{(2, n^{(1)})}, \ldots, Y_{n_{j}^{(2)} j}^{(2, n^{(1)})} \bigr)$ be the observed outcomes for participants $1, \ldots, n_j^{(2)}$ in center $j$ of stage 2, and $\overline{\boldsymbol{Y}}^{\left(2, n^{(1)}\right)}=\bigl(\boldsymbol{Y}_{1}^{\left(2, n^{(1)}\right)}, \ldots, \boldsymbol{Y}_{J^{(2)}}^{\left(2, n^{(1)}\right)}\bigr)$ be the outcomes for all stage 2 centers.

\begin{assu} \label{conditional indep assu}
Given $\overline{\hat{\boldsymbol{x}}}^{\left(2, n^{(1)}\right)}$, 
$\bigl(\overline{\boldsymbol{A}}^{\left(2, n^{(1)}\right)},\overline{\boldsymbol{Y}}^{\left(2, n^{(1)}\right)}\bigr)$ are independent of the stage 1 data $\bigl(\overline{\boldsymbol{a}}^{(1)},\overline{\boldsymbol{Y}}^{(1)} \bigr)$. 
\end{assu}  
Assumption \ref{conditional indep assu} states that learning from data collected so far takes place only through the calculation of the recommended intervention.

\begin{assu} \label{intervention_cvg_assumption}
For non center-specific stage 2 recommended intervention, $\hat{\boldsymbol{x}}^{(2,n^{(1)})}
\xrightarrow{P}
\boldsymbol{x}^{(2)}$ as $n^{(1)} \xrightarrow{} \infty$, where $\boldsymbol{x}^{(2)}$ is the limit of the recommended intervention.
For center-specific stage 2 recommended interventions, for each $j \in \{1,\ldots,J^{(2)}\}$, $\hat{\boldsymbol{x}}^{(2,n^{(1)})}_{j}
\xrightarrow{P}
\boldsymbol{x}_j^{(2)}$ as $n^{(1)} \xrightarrow{} \infty$, where $\boldsymbol{x}_j^{(2)}$ are center-specific limits of the recommended interventions.
\end{assu}
Section \ref{asymptotic properties section} discusses why Assumption \ref{intervention_cvg_assumption} holds when the stage 2 recommended intervention is obtained by solving an optimization problem with an outcome goal and a power goal.

As in \citet{nevo2021analysis} and \citet{bing2023learnasyougo}, under Assumption \ref{intervention_cvg_assumption}, the assumption on the $h_j$ that map recommended interventions to actual interventions, combined with the Continuous Mapping Theorem, imply that the $\boldsymbol{A}_j^{(2,n^{(1)})}$ converge in probability to some $\boldsymbol{a}_j^{(2)}$.

\begin{assu}\label{glm_unif_bound_assu}
The outcomes $Y_{ij}$ take values in a compact space, and the possible parameter vectors $\boldsymbol{\beta}$ are in a compact space.
\end{assu}

\begin{assu}\label{ind_error}
For LAGO trials with continuous outcomes,
let $\epsilon_{ij}$ be the error term for individual $i$ in center $j$ of the GLM from Assumption~\ref{glm assumption}. 
That is, $Y_{ij} = g^{-1}\bigl( \beta_{0}+\boldsymbol{\beta}_{1}^{T}\boldsymbol{a}_j \bigr) + \epsilon_{ij}$, with 
$E\left( \epsilon_{ij} | \boldsymbol{a}_j \right) = 0$. 
The $\epsilon_{ij}$ are independent, and their distributions are independent of the intervention package composition $\boldsymbol{a}_{j}$.
\end{assu}

The parameter space of $\boldsymbol{\beta}$ can be partitioned into two disjoint sets:  $\Lambda_{feasible}$ and $\Lambda_{infeasible}$, where $\Lambda_{feasible}$ consists of all $\boldsymbol{\beta}$ for which the constraints in (\ref{optimization no hat}) define a non-empty feasible region of $\boldsymbol{x}$, and $\Lambda_{infeasible}$ consists of all $\boldsymbol{\beta}$ for which the constraints yield an empty feasible region of $\boldsymbol{x}$.
\begin{assu}\label{feasible assu}
If $\boldsymbol{\beta}^* \in \Lambda_{feasible}$, equation (\ref{optimization no hat}) has a unique solution $\boldsymbol{x}$ in an open neighborhood of $\boldsymbol{\beta}^*$ that depends continuously on $\boldsymbol{\beta}$. 
\end{assu}

\citet{nevo2021analysis} argued that Assumption \ref{feasible assu} holds when, e.g., the cost function is linear with unit costs $c_1,...,c_P$, where the ratios $\beta_{1p}/c_p$ are all distinct for the $p=1,\ldots,P$ components with $\beta_{1p}>0$. For cubic cost functions or other forms of cost functions, Assumption \ref{feasible assu} needs to be proven or empirically verified for the specific problem. Appendix Section I provides a method for empirically verifying Assumption \ref{feasible assu}.

\begin{assu}\label{not reachable under the null assu}
    The outcome goal from equation (\ref{optimization no hat}), $\tilde{p}$ or $\tilde{\mu}$, cannot be reached without intervention. 
\end{assu}
\vspace{-0.2cm}
Assumption \ref{not reachable under the null assu} is needed in Section \ref{asymptotic properties section} to show that Assumption \ref{intervention_cvg_assumption} holds. It will usually hold in practice, where the study team intends to improve outcomes compared to standard of care.

\section{Including a Power Goal in LAGO Optimizations}\label{power goal section}
Including a power goal affects the calculation of the stage 2 recommended intervention, $\hat{\boldsymbol{x}}^{(2,n^{(1)})}$, where the exact formula depends on the pre-specified test at the end of the trial. 
As an example, Section~\ref{power goal section} considers a two-sample z-test for the difference between two proportions as the pre-specified test.

The outcome goal and the power goal from equation (\ref{optimization two constraints}) are discussed separately. In LAGO trials with binary outcomes, the outcome goal $\hat{p}_{\boldsymbol{x}}\bigl( \hat{\boldsymbol{\beta}}^{(1)} \bigr) \geq \tilde{p}$ implies that $\hat{\boldsymbol{x}}^{(2,n^{(1)})}$ needs to satisfy
$
    \operatorname{expit}\bigl(\hat{\beta}_0^{(1)}+\bigl(\hat{\boldsymbol{\beta}}_1^{(1)}\bigr)^T \hat{\boldsymbol{x}}^{(2,n^{(1)})}\bigr)
    \geq \tilde{p}.
$

Next, consider the power goal from equation (\ref{optimization two constraints}).
We propose two methods to incorporate the power goal into the LAGO design: the unconditional and conditional power approaches, both based on the stage 1 data. 
The unconditional power approach calculates $\hat{\boldsymbol{x}}^{(2, n^{(1)})}$ such that the test of equation (\ref{test statistic}) would achieve the pre-specified power goal $\Pi$ if the LAGO trial were to implement the same interventions in stage 1 as in reality, followed by the stage 2 recommended intervention, all planned pre-study.
The unconditional power approach thus does not condition on the observed stage 1 data and considers variation in both the stage 1 and stage 2 data. Section \ref{unconditional power approach} presents the unconditional power goal.

The conditional power approach conditions on the stage 1 data and estimates the conditional distribution of the test statistic, and only considers variation in the stage 2 data; this approach is related to conditional power calculation for futility analyses \citep{proschan1995designed}. Section \ref{conditional power approach} presents the conditional power goal.

Regardless of the approach, considering a power goal does not alter the formula of the test statistic at the end of stage 2. This is because the power goal only affects the calculation of the stage 2 recommended intervention, not the formula of the test at the end of the trial.

\subsection{Unconditional Power Based on Previous Stages}\label{unconditional power approach}

The unconditional power approach considers the actual stage 1 interventions and the stage 2 recommended interventions as fixed and calculates the power that this fixed design would have, based on the parameters estimated from the stage 1 data. For this fixed design, standard asymptotic theory on the two-sample z-test for the difference between two proportions implies that $Z^2$ from equation (\ref{test statistic}) approximately follows a non-central $\chi^2$ distribution with 1 degree of freedom and some non-centrality parameter $\lambda$.
\begin{thm}[Unconditional Power Approach]\label{unconditional power theorem}
    \textcolor{white}{xxx}\\
    Let $\chi^2_{\alpha, 1}$ be the upper $\alpha$ quantile of the central $\chi^2$ distribution with 1 degree of freedom. 
    For $\alpha=0.05$, $\chi^2_{\alpha, 1}=3.84$. 
    Let ${\lambda}_{min}$ be the minimum value of the non-centrality parameter for the non-central $\chi^2$ distribution with 1 degree of freedom, so that for a variable $T$ from a non-central $\chi^2$ distribution with non-centrality parameter ${\lambda}_{min}$, $P(T >\chi^2_{\alpha, 1})=\Pi$.
    Let
    \begin{equation*}
        \hat{S}_1^{(2)}\bigl(\hat{\boldsymbol{x}}^{(2,n^{(1)})}, \hat{\boldsymbol{\beta}}^{(1)}\bigr) = n_1^{(2)} \operatorname{expit}\bigl(\hat{\beta}_0^{(1)}+\bigl(\hat{\boldsymbol{\beta}}_1^{(1)}\bigr)^T \hat{\boldsymbol{x}}^{(2,n^{(1)})}\bigr),  
        \;\hat{S}_0^{(2)}\bigl(\hat{\boldsymbol{\beta}}^{(1)} \bigr) = n_0^{(2)} \operatorname{expit}\bigl( \hat{\beta}_0^{(1)}\bigr),
    \end{equation*}
    {\small
    \begin{equation*}
    \begin{aligned}
        &{\lambda}\bigl(\hat{\boldsymbol{x}}^{(2,n^{(1)})}; \hat{\boldsymbol{\beta}}^{(1)}\bigr) = \\ 
        &\left(
        \frac{\left.\Bigl({S_1^{(1)}+\hat{S}_1^{(2)}\bigl(\hat{\boldsymbol{x}}^{(2,n^{(1)})}, \hat{\boldsymbol{\beta}}^{(1)}\bigr) }\Bigr)\right/{N_1} - \left.\Bigl({S_0^{(1)}+\hat{S}_0^{(2)}\bigl(\hat{\boldsymbol{\beta}}^{(1)} \bigr)}\Bigr)\right/{N_0} }{\sqrt{{ \frac{S_1^{(1)}+\hat{S}_1^{(2)}\bigl(\hat{\boldsymbol{x}}^{(2,n^{(1)})}, \hat{\boldsymbol{\beta}}^{(1)}\bigr)}{N_1} \Bigl(1-\frac{S_1^{(1)}+\hat{S}_1^{(2)}\bigl(\hat{\boldsymbol{x}}^{(2,n^{(1)})}, \hat{\boldsymbol{\beta}}^{(1)}\bigr)}{N_1} \Bigr)} \frac{1}{N_1} + { \frac{S_0^{(1)}+\hat{S}_0^{(2)}\bigl(\hat{\boldsymbol{\beta}}^{(1)} \bigr)}{N_0} \Bigl(1-\frac{S_0^{(1)}+\hat{S}_0^{(2)}\bigl(\hat{\boldsymbol{\beta}}^{(1)} \bigr)}{N_0}\Bigr)} \frac{1}{N_0}}}
        \right)^2. 
    \label{hat z}
    \end{aligned}
    \end{equation*}
    }
    Under Assumptions \ref{logistic assumption}, \ref{conditional indep assu}, \ref{intervention_cvg_assumption}, \ref{glm_unif_bound_assu}, \ref{feasible assu}, and \ref{not reachable under the null assu}, the stage 2 recommended intervention $\hat{\boldsymbol{x}}^{(2,n^{(1)})}$
    solves the following optimization problem: 
    \begin{equation}\label{unconditonal power approach optimization problem}
        \text{Min}_{\boldsymbol{x}} C\left(\boldsymbol{x}\right)
        \; \text {subject to} \; 
        \operatorname{expit}\bigl(\hat{\beta}_{0}^{(1)}+(\hat{\boldsymbol{\beta}}^{(1)}_1)^T \hat{\boldsymbol{x}}^{(2,n^{(1)})} \bigr)
        \geq \tilde{p}, \;\text{and}\;\; {\lambda}\bigl(\hat{\boldsymbol{x}}^{(2,n^{(1)})}; \hat{\boldsymbol{\beta}}^{(1)}\bigr) \geq \lambda_{min} .
    \end{equation}
\end{thm}
\vspace{-0.5cm}
Appendix Section A provides the proof of Theorem \ref{unconditional power theorem}. The steps for calculating $\hat{\boldsymbol{x}}^{(2,n^{(1)})}$ from equation (\ref{unconditonal power approach optimization problem}) are:
\begin{enumerate}
    \item Given $\alpha$, find the critical value $\chi^2_{\alpha, 1}$ from the (central) $\chi^2$ distribution with 1 degree of freedom. For $\alpha=0.05$, $\chi^2_{\alpha, 1}=3.84$. 
    \item Estimate the minimum value of the non-centrality parameter, ${\lambda}_{min}$ (defined in Theorem~\ref{unconditional power theorem}), needed to achieve the pre-specified power goal $\Pi$, $P(T>\chi^2_{\alpha, 1})\geq \Pi$, under the alternative hypothesis.
    For $\alpha=0.05$ and $\Pi=0.8$, since $\chi^2_{\alpha, 1}=3.84$, ${\lambda}_{min}=7.85$.
    \vspace{-0.2cm}
    \item Solve for the value of $\tilde{p}_{pow\_u}$ so that for $\hat{\boldsymbol{x}}^{(2,n^{(1)})}$ that satisfies $\operatorname{expit}\bigl(\hat{\beta}_{0}^{(1)}+(\hat{\boldsymbol{\beta}}^{(1)}_1)^T \hat{\boldsymbol{x}}^{(2,n^{(1)})} \bigr) \geq \tilde{p}_{pow\_u}$, the constraint ${\lambda}\bigl(\hat{\boldsymbol{x}}^{(2,n^{(1)})}; \hat{\boldsymbol{\beta}}^{(1)}\bigr) \geq \lambda_{min}$ is satisfied.
    \vspace{-0.3cm}
    \item Calculate the maximum achievable success probability based on $\hat{\boldsymbol{\beta}}^{(1)}$, within the pre-specified lower and upper bounds of the $p$'th intervention components $\mathcal{L}_p$ and $\mathcal{U}_p$, denoted as $\hat{p}_{max}(\hat{\boldsymbol{\beta}}^{(1)}) = 
    \hat{\beta}_0^{(1)} + 
    \sum_{p=1}^P
    I\{\hat{\beta}_{1p}^{(1)} > 0\}\hat{\beta}_{1p}^{(1)}\mathcal{U}_p + I\{\hat{\beta}_{1p}^{(1)} < 0\}\hat{\beta}_{1p}^{(1)}\mathcal{L}_p.$
    \item 
    Calculate $\hat{\boldsymbol{x}}^{(2,n^{(1)})}$ using Algorithm \ref{alg:rec_int}, which ensures that $\hat{\boldsymbol{x}}^{(2,n^{(1)})}$ depends continuously on $\hat{\boldsymbol{\beta}}^{(1)}$.
\end{enumerate}
\begin{algorithm}[H]
        \caption{Calculating $\hat{\boldsymbol{x}}^{(2, n^{(1)})}$ based on an outcome goal and a power goal (unconditional power approach)}\label{alg:rec_int}
        \If{$\hat{p}_{max}\bigl(\hat{\boldsymbol{\beta}}^{(1)}\bigl) \geq \max\left( \tilde{p},\; \tilde{p}_{pow\_u} \right)$}{
            $\hat{\boldsymbol{x}}^{(2, n^{(1)})}$ is obtained by solving equation (\ref{optimization}) using $\max\left( \tilde{p},\; \tilde{p}_{pow\_u}\right)$ in place of $\tilde{p}$.
        }
        \ElseIf{$\hat{p}_{max}\bigl(\hat{\boldsymbol{\beta}}^{(1)}\bigl) \geq \tilde{p} $}{
            $\hat{\boldsymbol{x}}^{(2, n^{(1)})}$ is obtained by solving equation (\ref{optimization}) using $\hat{p}_{max}\bigl(\hat{\boldsymbol{\beta}}^{(1)}\bigl) $ in place of $\tilde{p}$.
        }
        \Else{
               $\hat{\boldsymbol{x}}^{(2, n^{(1)})}$ is obtained by using the shrinking method from 
               \citet{nevo2021analysis}, 
            which shrinks $\hat{\boldsymbol{x}}^{(2, n^{(1)})}$ towards the stage 1 recommended intervention using a continuous function of $\hat{\boldsymbol{\beta}}^{(1)}$. Appendix Section B discusses the details of the shrinking method.
        }
    \end{algorithm}

\subsection{Conditional Power Based on Previous Stages}\label{conditional power approach}

The conditional power approach estimates the stage 2 recommended intervention, $\hat{\boldsymbol{x}}^{(2,n^{(1)})}$, so that the test of equation (\ref{test statistic}) would achieve the power goal $\Pi$ conditional on the stage 1 data \citep{proschan1995designed}. 
\begin{thm}[Conditional Power Approach]\label{conditional power goal theorem}
    Let
    \begin{equation*}
    \hat{\Delta}\bigl(\hat{\boldsymbol{x}}^{(2,n^{(1)})}, \hat{\boldsymbol{\beta}}^{(1)}\bigr) =
    \left. \hat{S}_1^{(2)}\bigl(\hat{\boldsymbol{x}}^{(2,n^{(1)})}, \hat{\boldsymbol{\beta}}^{(1)}\bigr) \right/ N_1 - \left. \hat{S}_0^{(2)}\bigl(\hat{\boldsymbol{\beta}}^{(1)} \bigr) \right/ N_0 ,
    \end{equation*}
    where $ \hat{S}_1^{(2)}\bigl(\hat{\boldsymbol{x}}^{(2,n^{(1)})}, \hat{\boldsymbol{\beta}}^{(1)}\bigr)$ and $\hat{S}_0^{(2)}\bigl(\hat{\boldsymbol{\beta}}^{(1)} \bigr)$ are defined in Theorem \ref{unconditional power theorem}. Let
    \begin{equation*}
    \begin{aligned}
    &\hat{\sigma}^2\bigl(\hat{\boldsymbol{x}}^{(2,n^{(1)})}, \hat{\boldsymbol{\beta}}^{(1)}\bigr) = 
    \frac{n_0^{(2)}}{N_0^2} \operatorname{expit}\left( \hat{\beta}_0^{(1)}\right)\left(1-\operatorname{expit}\left( \hat{\beta}_0^{(1)}\right)\right) \\
    & \hspace{1cm} +\frac{n_1^{(2)}}{N_1^2} \operatorname{expit}\bigl(\hat{\beta}_0^{(1)}+\bigl(\hat{\boldsymbol{\beta}}_1^{(1)}\bigr)^T \hat{\boldsymbol{x}}^{(2,n^{(1)})}\bigr)\bigl(1-\operatorname{expit}\bigl(\hat{\beta}_0^{(1)}+\bigl(\hat{\boldsymbol{\beta}}_1^{(1)}\bigr)^T \hat{\boldsymbol{x}}^{(2,n^{(1)})}\bigr)\bigr).
    \end{aligned}
    \end{equation*}
    Under Assumptions \ref{logistic assumption}, \ref{conditional indep assu}, \ref{intervention_cvg_assumption}, \ref{glm_unif_bound_assu}, \ref{feasible assu}, and \ref{not reachable under the null assu}, the stage 2 recommended intervention $\hat{\boldsymbol{x}}^{(2,n^{(1)})}$ 
    solves the following optimization problem: 
    \begin{equation*}
        \text{Min}_{\boldsymbol{x}} C\left(\boldsymbol{x}\right)
        \; \text {subject to} \; 
        \operatorname{expit}\left(\hat{\beta}_{0}^{(1)}+(\hat{\boldsymbol{\beta}}^{(1)}_1)^T \hat{\boldsymbol{x}}^{(2,n^{(1)})} \right)
        \geq \tilde{p}, \;\text{and}
    \end{equation*}
    {\footnotesize
    \begin{equation*}
    \scalebox{0.9}{$
    \begin{aligned}
        & z_{\alpha/2}
        \sqrt{{ \frac{S_1^{(1)}+\hat{S}_1^{(2)}(\hat{\boldsymbol{x}}^{(2,n^{(1)})}, \hat{\boldsymbol{\beta}}^{(1)})}{N_1} \left(1-\frac{S_1^{(1)}+\hat{S}_1^{(2)}(\hat{\boldsymbol{x}}^{(2,n^{(1)})}, \hat{\boldsymbol{\beta}}^{(1)})}{N_1} \right)} \frac{1}{N_1} + { \frac{S_0^{(1)}+\hat{S}_0^{(2)}(\hat{\boldsymbol{\beta}}^{(1)} )}{N_0} \left(1-\frac{S_0^{(1)}+\hat{S}_0^{(2)}(\hat{\boldsymbol{\beta}}^{(1)} )}{N_0}\right)} \frac{1}{N_0}}
        \\
        & \hspace{4cm} - \frac{S_1^{(1)}}{N_1} + \frac{S_0^{(1)}}{N_0} - \hat{\Delta}\bigl(\hat{\boldsymbol{x}}^{(2,n^{(1)})}, \hat{\boldsymbol{\beta}}^{(1)}\bigr)
         - z_{\Pi} \; \hat{\sigma}^2\bigl(\hat{\boldsymbol{x}}^{(2,n^{(1)})}, \hat{\boldsymbol{\beta}}^{(1)}\bigr)
         \leq 0. \\ 
    \end{aligned}
    $}
    \end{equation*}
    }
\end{thm}
Appendix Section A provides the proof of Theorem \ref{conditional power goal theorem}, which differs mathematically from the proof of Theorem \ref{unconditional power theorem}, and provides detailed steps for calculating $\hat{\boldsymbol{x}}^{(2,n^{(1)})}$ based on an outcome goal and a power goal, using the conditional power approach.

\subsection{Comparing unconditional and conditional approaches}\label{compare two power approaches}

The unconditional power approach does not condition on the observed stage 1 data and takes into account the variance in both $\bigl(S_0^{(1)}, S_1^{(1)}\bigr)$ and $\bigl(S_0^{(2)},S_1^{(2)}\bigr)$.
In contrast, the conditional power approach conditions on the stage 1 data and only accounts for the variance in $\bigl(S_0^{(2)},S_1^{(2)}\bigr)$. 
Due to the larger variance considered in the distribution of the test statistic, the unconditional power approach requires a higher estimated mean of the test statistic to achieve the pre-specified power goal.
Thus, if a positive effect of the intervention on the outcome is expected, the unconditional power approach leads to a higher constraint value on $\operatorname{expit}\bigl(\hat{\beta}_{0}^{(1)}+(\hat{\boldsymbol{\beta}}^{(1)}_1)^T \hat{\boldsymbol{x}}^{(2,n^{(1)})} \bigl)$. In scenarios where the components of $\hat{\boldsymbol{\beta}}^{(1)}$ are positive, the unconditional power approach leads to larger recommended values for $\hat{\boldsymbol{x}}^{(2,n^{(1)})}$.

\section{Asymptotic Properties of \texorpdfstring{$\hat{\boldsymbol{\beta}}^{(2)}$}{beta}}\label{asymptotic properties section}
Appendix Section C shows that the estimator $\hat{\boldsymbol{\beta}}^{(2)}$ from the final analysis, which uses data from both stage 1 and stage 2, is consistent and asymptotically normal. The proof involves showing that Assumption \ref{intervention_cvg_assumption} holds when the stage 2 recommended intervention is obtained by solving an optimization problem based on both an outcome goal and a power goal. The proof considers two cases: case (1) the maximum achievable success probability within the pre-specified bounds of the intervention components, based on the true parameter $\boldsymbol{\beta}^*$, exceeds the outcome goal $\tilde{p}$, and case (2) it is less than $\tilde{p}$. For both cases, we show that Assumption \ref{intervention_cvg_assumption} holds, so that the theoretical framework from \citet{nevo2021analysis} and  \citet{bing2023learnasyougo} applies, thus establishing consistency and asymptotic normality of $\hat{\boldsymbol{\beta}}^{(2)}$.

\section{Level of the Test at the End of the LAGO Trial}\label{level section}

As discussed in Section \ref{hypothesis testing section}, this article considers two ways to test the null hypothesis that the intervention package has no effect on the outcome of interest, regardless of variations in implementation: the 1-df and $P$-df tests. We discuss these two types of tests separately and show why the level of the pre-specified test at the end of the LAGO trial retains the nominal value when incorporating an additional power goal in LAGO optimizations.

For the 1-df tests, such as a z-test or t-test, as previously noted in \citet{nevo2021analysis} and \citet{bing2023learnasyougo}, adapting the intervention package in LAGO trials does not affect the level of the 1-df test that will be carried out at the end of the trial, provided that the test is valid without the LAGO optimizations and the null hypothesis is that the intervention package has no effect on the outcome of interest, regardless of variations in implementation. Under this null hypothesis, regardless of how the intervention is adapted, the distribution of the outcome remains unchanged. Therefore, if the 1-df test is valid without the LAGO optimizations, it is also valid with any LAGO optimizations. 

For the $P$-df tests, specifically a Wald test, asymptotic properties of the estimator for $\boldsymbol{\beta}$ are needed to establish that the Wald test asymptotically has the correct type-1 error rate.
Section \ref{asymptotic properties section} showed that the estimator $\hat{\boldsymbol{\beta}}^{(2)}$ from the final analysis based on data from both stage 1 and stage 2 is consistent and asymptotically normal. Asymptotic normality of the estimator guarantees the asymptotic validity of the $P$-df Wald test.

\section{Simulations}\label{simulation}
Simulation studies were conducted to evaluate the methods presented in this article.
All scenarios considered two-stage LAGO trials with binary outcomes, with 2000 datasets simulated for each scenario. 
Scenario 1 compared LAGO designs with both an outcome goal and a power goal to LAGO designs with only an outcome goal. 
Scenario 2 compared LAGO designs with only a power goal to the non-LAGO fractional factorial designs, focusing specifically on power.

Scenario 1 considered 2 cases, 1a and 1b, and compared two optimization criteria for calculating the stage 2 recommended interventions: based on only an outcome goal, and based on both an outcome goal and a power goal of the two-sample z-test for the difference between two proportions. 
In both cases, the total number of centers in both the intervention group and the control group was $J = 4$ for each stage.
The per-center sample sizes considered were $n_j^{(1)}=n_j^{(2)}=40, 50, 60, \text{and}\; 100$.
Stage 1 of the two-stage LAGO trial used a fractional factorial design, setting the two-component intervention package $x_1$ and $x_2$ at $(0,0)$ in the control group, and $(1,0)$, $(0,4)$, and $(1,4)$ in the intervention group.
The minimum and maximum values of $x_1$ and $x_2$ were $[\mathcal{L}_1, \mathcal{U}_1] = [0,2]$ and $[\mathcal{L}_2, \mathcal{U}_2] = [0,8]$, respectively. 
The true coefficients were $(\beta^*_0, \beta^*_{11}, \beta^*_{12}) = (0.1, 0.3, 0.15)$.
The model for the binary outcomes was the same as described in Assumption \ref{logistic assumption}.
The power goals considered were $\Pi = 0.80, 0.85, 0.90 \;\text{and}\; 0.95$.

In Scenario 1a, a cubic cost function for the intervention package was considered:
$C(\boldsymbol{x})=2x_1^3-1.19x_1^2+10x_1+10 + 0.1x_2^3-0.2x_2^2+2x_2$.
The cubic cost function was chosen for its economy of scale, and 
Appendix Section F.1 provides additional details about the cubic cost function.  
In Scenario 1a, the outcome goal was $\tilde{p} = 0.7$, so that the true optimal intervention based on the outcome goal was $(0.5, 4.0)$ within two decimal places.
In Scenario 1b, a linear cost function for the intervention package was considered: $C(\boldsymbol{x}) = x_1 + 4x_2$.
The outcome goal was $\tilde{p}=0.7455$, so that the true optimal intervention based on the outcome goal was $(3.25, 0)$ within three decimal places.

Tables \ref{simulation 1b1 table 1} and \ref{simulation 1b1 table 2} present selected results for Scenario 1a. Appendix Sections {F.2} and {F.3} provide the full results for Scenario 1. 
Table \ref{simulation 1b1 table 1} reports bias, the ratios between the average estimated standard error and the empirical standard error of $\hat{\beta}_{11}$ and $\hat{\beta}_{12}$, coverage rates of the 95\% confidence intervals of the coefficients, and the power of the two-sample z-test for the difference between two proportions at the end of the LAGO trial. The first entry for each set of per-center sample sizes shows results based on the LAGO design with only an outcome goal.

Table \ref{simulation 1b1 table 1} shows that the LAGO design with an outcome goal and a power goal increased the power. Higher power goals corresponded to higher resulting power at the end of the LAGO trial. Under the LAGO design with both an outcome goal and a power goal, the relative bias of the coefficients, the ratio of the average estimated standard error to the empirical standard error of the coefficients, and the empirical coverage rate of the 95\% confidence intervals were all comparable to those calculated using the LAGO design with only an outcome goal. For any per-center sample size, both $\hat{\beta}_{11}$ and $\hat{\beta}_{12}$ had minimal relative bias, and the empirical coverage rates of the 95\% confidence intervals were close to 95\%.

As expected (Section~\ref{compare two power approaches}), with the same power goal, the unconditional power approach led to higher power compared to the conditional power approach.
With $n_j^{(1)}=n_j^{(2)}=40$, and higher power goals, e.g., when the power goal was 0.95, the power at the end of the LAGO trial did not reach the power goal for either the unconditional or conditional power approach. This was due to higher uncertainty around the coefficient estimates based on the stage 1 data with small sample sizes, and the inherent limitations of what could be achieved given the pre-specified bounds on the intervention components. When neither approach could reach the desired power at the end of the LAGO trial, as expected, the unconditional power approach led to power closer to the desired power goal.
With $n_j^{(1)}=n_j^{(2)}=50$ or $n_j^{(1)}=n_j^{(2)}=60$, and lower power goals, e.g., when the power goal was 0.80, both the unconditional and conditional power approaches reached the desired power. The conditional power approach exceeded the desired power goal less than the unconditional power approach. 
With $n_j^{(1)}=n_j^{(2)}=100$, adding a power goal to the LAGO design had minimal impact since the power at the end of the LAGO trial was already 98.6\% under the LAGO design with only an outcome goal.

Table \ref{simulation 1b1 table 2} reports the relative bias and quantiles of the true mean under the estimated optimal interventions, both calculated based on the stage 1 data and based on all data. The first entry for each set of per-center sample sizes shows results under the LAGO design with only an outcome goal.
As expected, the estimated optimal interventions had smaller bias and narrower quantile ranges of the mean outcomes when calculated based on all data, compared to those calculated based only on the stage 1 data.
With both an outcome goal and a power goal, the relative bias of the estimated optimal intervention components calculated based on all data was comparable to that calculated under the LAGO design with only an outcome goal.
The ranges of the 2.5\% and 97.5\% quantiles for the true mean under the final estimated optimal interventions were narrower compared to those calculated under the LAGO design with only an outcome goal.

Scenario 2 compared the performance of the LAGO design with a power goal to the non-LAGO, fractional factorial design, with a focus on power. Scenario 2 considered both cubic and linear cost functions. Results show that the LAGO design achieved substantially higher power. Higher power goals corresponded to higher actual power at the end of the trial, with the unconditional power approach typically yielding higher power than the conditional power approach. Detailed setup and findings for Scenario 2 are presented in Appendix Sections {F.4} and {F.5}.

\begin{table}[htbp]
\label{simulation 1a results table}
\renewcommand{\arraystretch}{0.5}
\caption{Selected simulation results for scenario 1a}
\hrule
    \centering{
    \subfloat[{Simulation study results for individual package component effects and power with a cubic cost function}{\label{simulation 1b1 table 1}}] {
        \scriptsize 
        \begin{tabular}{llllllllllllll}
         &  &  &  &  & \multicolumn{3}{c}{$\hat{\beta}_{11}$} &  & \multicolumn{3}{c}{$\hat{\beta}_{12}$} &  &  \\ \cline{6-8} \cline{10-12}
        $n_j^{(1)}$ & $n_j^{(2)}$ & \begin{tabular}[c]{@{}l@{}}\%\\ Power\\ Goal\end{tabular} & \begin{tabular}[c]{@{}l@{}}Power\\ Goal\\ Approach\\ (U/C)\end{tabular} &  & \begin{tabular}[c]{@{}l@{}}\%Rel\\ Bias\end{tabular} & \begin{tabular}[c]{@{}l@{}}SE/\\ EMP.SD\\ ($\times 100$)\end{tabular} & CP95 &  & \begin{tabular}[c]{@{}l@{}}\%Rel\\ Bias\end{tabular} & \begin{tabular}[c]{@{}l@{}}SE/\\ EMP.SD\\ ($\times 100$)\end{tabular} & CP95 &  & \begin{tabular}[c]{@{}l@{}}\%\\ Power\end{tabular} \\ \hline
        40 & 40 & — & --- &  & 3.23 & 101.3 & 95.4 &  & 1.20 & 96.4 & 95.0 &  & 68.8* \\
         &  & 80 & C &  & 4.08 & 101.8 & 95.4 &  & 2.63 & 94.8 & 94.9 &  & 81.2 \\
         &  &  & U &  & 4.23 & 102.3 & 95.5 &  & 2.89 & 95.2 & 95.0 &  & 83.2 \\
         &  & 90 & C &  & 4.19 & 102.6 & 95.6 &  & 3.06 & 95.5 & 95.0 &  & 83.9 \\
         &  &  & U &  & 3.68 & 103.2 & 95.6 &  & 3.59 & 96.6 & 95.0 &  & 85.6 \\ \hline
        50 & 50 & — & --- &  & 1.09 & 99.9 & 94.5 &  & 0.94 & 97.8 & 94.9 &  & 79.3* \\
         &  & 80 & C &  & 1.14 & 100.0 & 94.5 &  & 1.54 & 97.0 & 95.1 &  & 87.4 \\
         &  &  & U &  & 1.59 & 100.0 & 94.5 &  & 1.75 & 97.1 & 94.9 &  & 88.7 \\
         &  & 90 & C &  & 1.74 & 100.1 & 94.5 &  & 1.80 & 97.3 & 94.9 &  & 89.5 \\
         &  &  & U &  & 2.36 & 100.3 & 94.6 &  & 2.28 & 98.1 & 95.1 &  & 91.3 \\ \hline
        60 & 60 & — & --- &  & 4.42 & 103.3 & 95.6 &  & 2.47 & 103.6 & 96.1 &  & 87.5* \\
         &  & 80 & C &  & 5.02 & 102.7 & 95.5 &  & 2.44 & 103.6 & 96.0 &  & 92.5 \\
         &  &  & U &  & 5.15 & 103.0 & 95.6 &  & 2.43 & 103.8 & 96.0 &  & 93.2 \\
         &  & 90 & C &  & 5.25 & 102.9 & 95.5 &  & 2.47 & 104.1 & 96.0 &  & 93.7 \\
         &  &  & U &  & 5.19 & 102.7 & 95.2 &  & 2.56 & 104.7 & 96.2 &  & 95.1 \\ \hline
        100 & 100 & --- & --- &  & 5.18 & 99.6 & 94.3 &  & 3.42 & 98.7 & 95.7 &  & 98.2*
        \end{tabular}
    }
    }
    \hrule
    \centering{
    \subfloat[{Simulation study results for the estimated optimal intervention with a cubic cost function}\label{simulation 1b1 table 2}]{
        \scriptsize 
        \begin{tabular}{llllllllllll}
         &  &  &  &  & \multicolumn{3}{c}{Stage 1} &  & \multicolumn{3}{c}{Stage 2 / LAGO Optimized} \\ \cline{6-8} \cline{10-12} 
        $n_j^{(1)}$ & $n_j^{(2)}$ & \begin{tabular}[c]{@{}l@{}}\%\\ Power\\ Goal\end{tabular} & \begin{tabular}[c]{@{}l@{}}Power\\ Goal\\ Approach\\ (U/C)\end{tabular} &  & \begin{tabular}[c]{@{}l@{}}\%Rel\\ Bias\\ $\hat{x}_1^{opt}$\end{tabular} & \begin{tabular}[c]{@{}l@{}}\%Rel\\ Bias\\ $\hat{x}_2^{opt}$\end{tabular} & \begin{tabular}[c]{@{}l@{}}PrOpt\\ (Q2.5,Q97.5)\end{tabular} &  & \begin{tabular}[c]{@{}l@{}}\%Rel\\ Bias\\ $\hat{x}_1^{opt}$\end{tabular} & \begin{tabular}[c]{@{}l@{}}\%Rel\\ Bias\\ $\hat{x}_2^{opt}$\end{tabular} & \begin{tabular}[c]{@{}l@{}}PrOpt\\ (Q2.5,Q97.5)\end{tabular} \\ \hline
        40 & 40 & --- & --- &  & 3.5 & 23.8 & (0.570, 0.779) &  & 12.9 & 13.3 & (0.606, 0.788) \\
         &  & 80 & C &  &  &  &  &  & 11.5 & 14.1 & (0.605, 0.786) \\
         &  &  & U &  &  &  &  &  & 11.1 & 14.2 & (0.605, 0.786) \\
         &  & 90 & C &  &  &  &  &  & 11.0 & 14.3 & (0.605, 0.786) \\
         &  &  & U &  &  &  &  &  & 9.6 & 14.2 & (0.604, 0.786) \\ \hline
        50 & 50 & --- & --- &  & 9.5 & 19.7 & (0.581, 0.767) &  & 8.8 & 11.0 & (0.612, 0.786) \\
         &  & 80 & C &  &  &  &  &  & 7.3 & 11.1 & (0.616, 0.775) \\
         &  &  & U &  &  &  &  &  & 7.1 & 11.3 & (0.615, 0.777) \\
         &  & 90 & C &  &  &  &  &  & 6.9 & 11.4 & (0.615, 0.775) \\
         &  &  & U &  &  &  &  &  & 5.9 & 11.5 & (0.612, 0.775) \\ \hline
        60 & 60 & --- & --- &  & 4.1 & 15.9 & (0.593, 0.772) &  & 7.6 & 10.4 & (0.629, 0.759) \\
         &  & 80 & C &  &  &  &  &  & 6.8 & 11.0 & (0.628, 0.757) \\
         &  &  & U &  &  &  &  &  & 6.4 & 10.8 & (0.628, 0.757) \\
         &  & 90 & C &  &  &  &  &  & 5.9 & 10.7 & (0.629, 0.757) \\
         &  &  & U &  &  &  &  &  & 5.3 & 10.6 & (0.629, 0.757) \\ \hline
        100 & 100 & --- & --- &  & 5.7 & 10.3 & (0.614, 0.761) &  & 6.0 & 7.7 & (0.645, 0.744)
        \end{tabular}
    }
    }
    \hrule
    \vspace{0.1cm}
    \scriptsize{
    {
    \setlength{\baselineskip}{0.5\baselineskip}
    \raggedright{
        2000 simulated datasets.  
        Number of centers for each stage: $J$ = 4. 
        True coefficients: $(\beta^*_0, \beta^*_{11}, \beta^*_{12}) = (0.1, 0.3, 0.15)$. \\
        A fractional factorial design was used in stage 1 with intervention package components set to: (0,0), (1,0), (0,4), and (1,4), and true success probabilities = 0.525, 0.599, 0.668, and 0.731. Outcome goal: $\Tilde{p}$ = 0.7. \\
        Cost function (cubic): $C(\boldsymbol{x})=2x_1^3-1.19x_1^2+10x_1+10 + 0.1x_2^3-0.2x_2^2+2x_2$. 
        Ranges: for $x_1$, $[0,2]$, for $x_2$, $[0,8]$. \\ 
        No unplanned variation in the interventions.\\
        \smallskip 
        $^*$: percent power calculated using the LAGO design with only an outcome goal.\\
        $n_j^{(1)}$: number of participants in each center $j$ at stage 1. $n_j^{(2)}$: number of participants in each center $j$ at stage 2.\\
        \%Power Goal: percent power goal of the two-sample z-test for the difference between two proportions. For scenarios with only an outcome goal, the \%Power Goal is set to ---.  \\
        Power Goal Approach (U/C): U -- the unconditional power approach, C -- the conditional power approach. \\ 
        \%RelBias: percent relative bias $|100(\hat{\beta}-\beta^\star)/\beta^\star|$.
        SE: mean estimated standard error, 
        EMP.SD: empirical standard deviation.\\
        CP95: empirical coverage rate of the 95\% confidence intervals.\\
        \%Power: percent power of the two-sample z-test for the difference between two proportions at the end of the LAGO trial.\\
        \%Rel Bias $\hat{x}_1^{opt}$: relative bias of the first component of the estimated optimal intervention $100(\hat{x}_1^{opt} - {x}_1^{opt})/{x}_1^{opt}$. \\
        \%Rel Bias $\hat{x}_2^{opt}$: relative bias of the second component of the estimated optimal intervention $100(\hat{x}_2^{opt} - {x}_2^{opt})/{x}_2^{opt}$.\\
        PrOpt: true success probability under the recommended intervention, calculated using true coefficients. \\
        Q2.5, Q97.5: 2.5\% and 97.5\% quantiles.\\
    }}
}
\end{table}

\section{Illustrative Example: The BetterBirth Study}\label{application}
The BetterBirth Study evaluated the effectiveness of the World Health Organization (WHO)'s Safe Childbirth Checklist (SCC) for reducing maternal and neonatal morbidity and mortality in Uttar Pradesh, India \citep{hirschhorn2015learning}. The SCC aims to address major causes of childbirth complications by promoting healthcare workers' adherence to essential birth practices (EBP) \citep{molina2022safe}. While the BetterBirth Study showed a significant increase in EBP adherence, it failed to show a significant reduction in its primary composite outcome, which consisted of stillbirth, early neonatal death, maternal death, and self-reported maternal severe complications within 7 days after delivery \citep{semrau2017outcomes}.

The BetterBirth Study involved 157,689 mothers and newborns across 3 stages. Stage 1 was a pilot with 2 centers that developed the intervention package. Stage 2 included 4 centers to refine the intervention package composition. Both stages 1 and 2 collected outcome data before and after the intervention. Stage 3 was a cluster randomized controlled trial with 30 centers, with 15 centers in the intervention arm and 15 centers in the control arm. The intervention package compositions implemented varied between the stages, and the 3 stages were analyzed separately in the actual BetterBirth Study \citep{hirschhorn2015learning}.

Our analysis focused on the binary outcome of neonatal apnea, that is, temporary cessation of breathing in newborns. Observed within one hour after birth, neonatal apnea serves as an indicator of several underlying health issues \citep{martin2014fanaroff}. Neonatal apnea can result from various factors, including asphyxia, hypothermia, stress, infections, and physical injuries \citep{Eichenwald1997Apnea}. The SCC could potentially lower apnea incidence by promoting a cleaner delivery environment \citep{world2012guidelines}, ensuring high-quality intrapartum care through close labor monitoring (Chapter 30 of \citet{devakumar2019oxford}), and mandating immediate drying and wrapping of newborns \citep{world1997thermal}.

The observed proportions of neonatal apnea, based on data from all 3 stages, were 14.8\% and 12.3\% for the control and intervention arms, respectively. 
The LAGO design with an outcome goal and a power goal was retrospectively applied to the BetterBirth Study to demonstrate how it could have optimized the intervention package component composition, and might have led to a statistically significant reduction in neonatal apnea in the BetterBirth Study.

Based only on the stage 3 data, a two-sample z-test for the difference between two proportions (equation ({\ref{test statistic}})) yielded a non-significant p-value of 0.154. We have shown that data from all stages could have been used to conduct the test in the final analysis, provided that the null hypothesis is that the intervention package has no effect on the outcome of interest, regardless of variations in implementation (Section~{\ref{hypothesis testing section}}). Based on data from all stages, a two-sample z-test for the difference between two proportions (equation ({\ref{test statistic}})) yielded a marginally significant p-value of 0.064.
Our analysis focused on two intervention components: the number of coaching visits after the initial launch (coaching visits), with a minimum of 1 and a maximum of 40 visits considered, and the number of days of the initial intervention launch event (launch duration), with a minimum of 1 and a maximum of 5 days considered (as in \citet{nevo2021analysis}). 
The binary outcome, neonatal apnea, was modeled using a logistic regression model (Assumption \ref{logistic assumption}). 
Table \ref{estimated effects apnea} shows the estimated effects of the intervention components after each stage, using all data available up to that stage. Based on the results from the second and third columns of Table \ref{estimated effects apnea}, coaching visits were associated with a reduction in the odds of neonatal apnea, though not significantly so.

To illustrate how the LAGO design could have modified the intervention package components and achieved a statistically significant reduction in the rate of neonatal apnea, our analysis calculated the stage 3 recommended interventions. Simulations were then used to estimate the effect of implementing these recommended interventions in stage 3 of the BetterBirth Study, and the statistical power that would have been achieved at the end of the BetterBirth Study had it implemented LAGO optimization after its second stage.
In stages 1 and 2 of the BetterBirth Study, the number of coaching visits varied among patients, and the launch duration was 3 days in stage 1 and 2 days in stage 2.
Considering the small sample size in stage 1 ($n^{(1)}=72$), we combined stages 1 and 2 as LAGO-stage 1 and calculated recommended interventions for stage 3 (LAGO-stage 2). 
The stage 3 recommended interventions were calculated 
considering: (1) an outcome goal, and (2) an outcome goal and a power goal. 
Let $x_1$ and $x_2$ be the two intervention package components mentioned above: coaching visits and launch duration (in days), respectively. A cubic cost function $C_1(\boldsymbol{x}) = 380x_1 - 24x_1^2 + 0.6x_1^3 + 1700x_2 - 950x_2^2 + 220x_2^3$ was considered as in \citet{bing2023learnasyougo}.
This LAGO optimization aimed to reduce the neonatal apnea rate. Since the overall rate of neonatal apnea in the combined stages 1 and 2 data was 0.14, the outcome goals considered were $\tilde{p}=0.1$ and $0.075$. The power goals considered were $\Pi=0.80$ and $0.90$, and both the unconditional and conditional power approaches were considered.
For each of the stage 3 recommended interventions, the procedure was repeated 2000 times to obtain its corresponding estimated power. The coefficient estimates from the third column of Table \ref{estimated effects apnea} were considered as the true values and used to simulate the data. 
The simulated stage 3 data were then combined with the actual data from LAGO-stage 1. A two-sample z-test for the difference between two proportions was performed on the final combined data from all 3 stages. The proportion of iterations resulting in a statistically significant test result provided the estimated power.

Table \ref{apnea recommended interventions table} presents the stage 3 recommended interventions calculated using different LAGO designs, along with their corresponding estimated power at the end of the BetterBirth Study.  
Since the coefficient of launch duration from the second column of Table \ref{estimated effects apnea} is greater than 1, and the goal was to reduce the neonatal apnea rate, as expected, all stage 3 recommended interventions minimized the number of days in launch duration by setting it to its lower bound.
For the stage 3 recommended interventions calculated using the LAGO design with an outcome goal, as expected, since the goal was to reduce the neonatal apnea rate, lower outcome goals corresponded to greater numbers of coaching visits. 
With $\tilde{p}=0.1$, the stage 3 recommended interventions calculated using the LAGO design with both an outcome goal and a power goal had greater numbers of coaching visits compared to those calculated using the LAGO design with only an outcome goal. 
As expected, higher power goals corresponded to greater numbers of coaching visits, which corresponded to higher estimated power at the end of the BetterBirth Study. 
Regardless of the power goal approach, all the estimated power values reached the pre-specified power goals (Table \ref{apnea recommended interventions table}). 
Regardless of the value of the outcome goal, as expected, when the power goal was the determining factor, the unconditional power approach led to a greater number of coaching visits compared to those recommended using the conditional power approach, resulting in higher estimated power at the end of the BetterBirth Study, at increased cost.

Simulation results from Section \ref{simulation} showed that when the sample sizes are sufficiently large, the conditional power approach exceeded the pre-specified power goal less than the unconditional power approach. 
Given the large sample size in the BetterBirth Study, if the LAGO design with both an outcome goal and a power goal were to be used in practice, the conditional power approach should be considered. 
For calculating the stage 3 recommended interventions in practice, a reasonable approach could involve setting the outcome goal to 0.1 (aiming for an approximately 4\% decrease in apnea rate) and the power goal to 0.8. Based on Table \ref{apnea recommended interventions table}, this would result in a recommended intervention of 27 coaching visits, rounded up from 26.65, and 1 launch day, with an estimated power of 88.6\% and a cost of \$5544 per center.

\begin{table}[t]
    \caption{
    Package component effect estimates for neonatal apnea occurrences, and estimated recommended interventions using different LAGO designs.}
    \renewcommand{\arraystretch}{0.5}
    \centering
    \subfloat[Package component effect estimates and confidence intervals, calculated after each stage.]{
        \scriptsize{
        \begin{tabular}{llll}
         & \begin{tabular}[c]{@{}l@{}}Stage 1\\ $n^{(1)}$ = 72\\ OR (CI-OR)\end{tabular} & \begin{tabular}[c]{@{}l@{}}Stages 1-2\\ $n^{(1)}+n^{(2)}$ = 1779 \\ OR (CI-OR)\end{tabular} & \begin{tabular}[c]{@{}l@{}}Stages 1-3 (All Data)\\ $n^{(1)}+n^{(2)}+n^{(3)}$ = 2628 \\ OR (CI-OR)\end{tabular} \\
        Intercept & 0.190 (0.056, 0.500) & 0.160 (0.127, 0.200) & 0.174 (0.150, 0.200) \\
        Coaching Visits (per 5 visits) & 4.561 (0.719, 44.858) & 0.888 (0.754, 1.069) & 0.881 (0.749, 1.058) \\
        Launch Duration (days) & 0.681 (0.278, 1.396) & 1.144 (0.833, 1.508) & 1.116 (0.816, 1.462)
        \end{tabular}\label{estimated effects apnea}
        }
    }
    \vspace{0.4cm}
    \hrule
    \centering
    \subfloat[Estimated stage 3 recommended interventions and their estimated power based on different LAGO designs.]{
        \scriptsize{ 
            \begin{tabular}{lllllll}
            \multirow{2}{*}{\begin{tabular}[c]{@{}l@{}}LAGO \\ Design\end{tabular}} & \multirow{2}{*}{\begin{tabular}[c]{@{}l@{}}\%\\ Power\\ Goal\end{tabular}} & \multirow{2}{*}{\begin{tabular}[c]{@{}l@{}}Power\\ Goal\\ Approach\\ (U/C)\end{tabular}} & \multicolumn{2}{l}{Outcome Goal = 0.1} & \multicolumn{2}{l}{Outcome Goal = 0.075} \\ \cline{4-7} 
             &  &  & $\boldsymbol{x}^{(3,n^{(1)},n^{(2)})}$ & \begin{tabular}[c]{@{}l@{}}Sim\\ \%\\ Power\end{tabular} & $\boldsymbol{x}^{(3,n^{(1)},n^{(2)})}$ & \begin{tabular}[c]{@{}l@{}}Sim\\ \%\\ Power\end{tabular} \\ \hline
            \begin{tabular}[c]{@{}l@{}}Outcome \\ Goal\end{tabular} & --- & --- & (21.18, 1.00) & 72.6 & (34.49, 1.00) & 97.8 \\ \hline
            \multirow{2}{*}{\begin{tabular}[c]{@{}l@{}}Outcome \\ Goal +\\ Power \\ Goal\end{tabular}} & 80 & C & (26.65, 1.00) & 88.6 & (34.49, 1.00) & 97.8 \\
             &  & U & (29.79, 1.00) & 94.0 & (34.49, 1.00) & 97.8 \\
             & 90 & C & (30.81, 1.00) & 95.0 & (34.49, 1.00) & 97.8 \\
             &  & U & (36.42, 1.00) & 98.5 & (36.42, 1.00) & 98.5
            \end{tabular}
            \label{apnea recommended interventions table}
        }
    }
    \vspace{0.1cm}
    \scriptsize \\
    {
    \setlength{\baselineskip}{0.5\baselineskip}
        \raggedright{
            OR, estimated odds ratio exp($\hat{\beta}$). CI-OR, 95\% confidence interval for the odds ratio. \\
            LAGO Design: Outcome Goal -- the LAGO design with only an outcome goal, Outcome Goal + Power Goal -- the LAGO design with both an outcome goal and a power goal. \\
            \%Power Goal: percent power goal of the two-sample z-test for the difference between two proportions.  \\
            Power Goal Approach (U/C): U -- the unconditional power approach, C -- the conditional power approach. \\ 
            ${\boldsymbol{x}}^{(3,(n^{(1)},n^{(2)}))}$: stage 3 recommended interventions. The first component is the number of coaching visits and the second component is the launch duration in days. The second component in all of the stage 3 recommended interventions is always 1 because the coefficient estimate for launch duration calculated at the end of stage 2 was not negatively associated with neonatal apnea. \\
            \begin{minipage}{\textwidth}
            \setlength{\baselineskip}{0.5\baselineskip}
            Sim \% Power: Simulated percent power at the end of the BetterBirth Study calculated using the corresponding stage 3 recommended interventions.
            \end{minipage}
        }}
        \label{BB table}
\end{table}

\vspace{-0.75cm}
\section{Discussion}\label{Discussion}
\vspace{-0.25cm}
The LAGO design incorporating both an outcome goal and a power goal provides investigators with a useful tool to prevent trial failure. Including a power goal improves LAGO's ability to detect treatment effects and further enhances the flexibility of the LAGO design.

Adding a power goal to the LAGO design increases the power at the end of the LAGO trial if the original LAGO trial was underpowered (Section~\ref{simulation}). The bias and standard errors of the coefficient estimates from the outcome model are similar to those calculated using the LAGO design with only an outcome goal. In some cases, the LAGO design with both an outcome goal and a power goal improves estimation of the optimal intervention calculated based on an outcome goal at the end of the LAGO trial.

Including a power goal in addition to an outcome goal in LAGO optimizations leads to larger increases in the recommended intervention package component levels, which can help prevent failed trials in real-world studies (Section~\ref{application}). If applied to the BetterBirth Study, a large and expensive trial that failed to find a significant effect on the primary outcome \citep{hirschhorn2015learning}, the LAGO design with both an outcome goal and a power goal would have led to a greater number of coaching visits and might have prevented trial failure. 

This article discusses two approaches for including a power goal in LAGO optimizations: the unconditional power approach (Section~\ref{unconditional power approach}) and the conditional power approach (Section~\ref{conditional power approach}). The unconditional power approach does not condition on the stage 1 data, while the conditional power approach does. 
When choosing between the unconditional and conditional power approaches, our findings suggest that for small sample sizes and high power goals, the unconditional power approach is preferable. For large sample sizes and low power goals, the conditional power approach is preferable. As argued in Section~\ref{simulation}, the unconditional power approach leads to higher power at increased cost. Therefore, when investigators are unsure whether their sample sizes are adequate, the unconditional power approach should be used to be conservative.  
In addition, it could be expected that for LAGO optimizations in later stages of the LAGO trial, the differences between the unconditional and conditional power approaches will be larger.

The LAGO design with both an outcome goal and a power goal also facilitates the inclusion of futility stops in LAGO trials. Notably, stopping the trial for futility preserves the type I error \citep{snapinn2006assessment}.
At the end of each stage $k$, to incorporate a power goal in LAGO optimizations, we estimate the power of the pre-specified test at the end of the trial if stage $k+1$ onwards were to implement the calculated recommended intervention. If none of the intervention package compositions being considered for stage $k+1$ onwards leads to satisfactory estimated power, investigators may choose to stop the trial for futility.

Appendix Section H provides a detailed examination of the relationship between the outcome goal and the power goal.
The power goal is particularly useful with small sample sizes and when the anticipated pre-specified outcome goal is close to the baseline value without interventions, that is when the intervention effect is small under the outcome goal.
In settings with large sample sizes, as discussed in Section \ref{asymptotic properties section}, the outcome goal alone will be sufficient to achieve satisfactory asymptotic power at the end, provided that the outcome goal cannot be reached in the control group. 
However, since trials operate with finite sample sizes, the power goal serves as an additional safeguard to help prevent trial failure.

Section \ref{asymptotic properties section} showed that Assumption \ref{intervention_cvg_assumption} holds when the stage 2 recommended intervention is calculated based on both an outcome goal and a power goal. The theory from \citet{nevo2021analysis} and  \citet{bing2023learnasyougo} applies, and it follows that the estimator $\hat{\boldsymbol{\beta}}^{(2)}$ from the final analysis, which uses data from all stages, is consistent and asymptotically normal. For LAGO trials with more than two stages, Appendix Section G provides the details, and the proof from Section \ref{asymptotic properties section} also applies.

As discussed in Section \ref{level section}, incorporating an additional power goal directly into the LAGO design preserves the asymptotic level of the statistical test at the end of the LAGO trial, as long as the null hypothesis states that the intervention package has no effect on the outcome of interest, regardless of variations in implementation. 
For 1-df tests (e.g., z-test or t-test), under the null hypothesis, the distribution of the outcome remains unchanged regardless of how the intervention was adapted.  
For $P$-df tests (specifically the Wald test), asymptotic normality of the estimator of the treatment effect in the final analysis based on combined data from all stages guarantees the asymptotic validity of the tests.

Implementing the LAGO design requires pre-trial collaboration between investigators and statisticians. The collaboration phase determines which aspects of the data will be used for calculating LAGO recommended interventions, which intervention package components to optimize, the lower and upper bounds of the selected intervention components, an appropriate outcome goal and/or power goal, and the cost function for the intervention package.
In addition, investigators need to determine whether to change the intervention component composition if the pre-specified outcome goal is already reached with the previous implementation of the intervention package during LAGO optimizations.

Future research will further extend the LAGO design, accommodating repeated outcomes and cluster/center effects. One extension is to address confounding by indication in the LAGO design, where center characteristics may affect how centers implement the recommended intervention. Another area for future research is the development of open-source software to calculate recommended and optimal interventions based on various optimization criteria, with the goal of making LAGO designs more accessible.

The LAGO design incorporating both an outcome goal and a power goal directly takes statistical power into account, helps prevent failed trials, and could have led to significant findings in the BetterBirth Study (Section~\ref{application}), while maintaining consistency and asymptotic normality of the estimator of the treatment effect in the final analysis. 
Current studies implementing the LAGO design include PULESA-Uganda \citep{PULESA}, TASKPEN-Zambia \citep{herce2024evaluating}, MAP-IT-Nigeria \citep{aifah2023study}, and HPTN 096 \citep{HPTN096}.
The LAGO design is a valuable methodological tool for conducting trials in various areas of implementation research.

\vspace{-0.5cm}
\bibliography{bibliography.bib}

\bigskip 
\begin{center}
{\Large\bf Appendix}
\end{center}
\appendix
The Appendix is organized as follows.
Section \ref{proof of theorems 1 and 2} provides the proofs of Theorems 1 and 2 from the main text.
Section \ref{shrinking method section} describes the details of the shrinking method used in Algorithm 1 from Section 3.1 of the main text.
Section \ref{asymp properties section} provides the proof of the asymptotic properties of the estimator $\hat{\boldsymbol{\beta}}^{(2)}$.
Section \ref{additional tests} presents details of LAGO trials with power goals for other common tests that are not discussed in detail in the main text.
Section \ref{both lago stages} shows the calculation of the stage 1 recommended interventions using pre-trial information.
Section \ref{additional simulation tables} presents additional simulation results.
Section \ref{K>2} extends the proposed LAGO methods to LAGO trials with more than two stages.
Section \ref{power goal and outcome goal appendix} examines the relationship between the outcome goal and the power goal.

\setcounter{thm}{0}
\section{Proofs of Theorems 1 and 2}\label{proof of theorems 1 and 2}
\subsection{Proof of Theorem 1}
For convenience, we restate Theorem 1 from the main text.
\begin{thm}[Unconditional Power Approach]\label{unconditional power theorem appendix}
    \textcolor{white}{xxx}\\
    Let $\chi^2_{\alpha, 1}$ be the upper $\alpha$ quantile of the central $\chi^2$ distribution with 1 degree of freedom. 
    For $\alpha=0.05$, $\chi^2_{\alpha, 1}=3.84$. 
    Let ${\lambda}_{min}$ be the minimum value of the non-centrality parameter for the non-central $\chi^2$ distribution with 1 degree of freedom, so that for a variable $T$ from a non-central $\chi^2$ distribution with non-centrality parameter ${\lambda}_{min}$, the probability of $T$ exceeding $\chi^2_{\alpha, 1}$ equals $\Pi$. 
    Let
    \begin{equation}
        \hat{S}_1^{(2)}\bigl(\hat{\boldsymbol{x}}^{(2,n^{(1)})}, \hat{\boldsymbol{\beta}}^{(1)}\bigr) = n_1^{(2)} \operatorname{expit}\bigl(\hat{\beta}_0^{(1)}+\bigl(\hat{\boldsymbol{\beta}}_1^{(1)}\bigr)^T \hat{\boldsymbol{x}}^{(2,n^{(1)})}\bigr),  
        \label{mu_1_2}
    \end{equation}
    \begin{equation}
        \hat{S}_0^{(2)}\bigl(\hat{\boldsymbol{\beta}}^{(1)} \bigr) = n_0^{(2)} \operatorname{expit}\bigl( \hat{\beta}_0^{(1)}\bigr),
        \label{mu_0_2}
    \end{equation}
    {\small
    \begin{equation}
    \begin{aligned}
        &{\lambda}\left(\hat{\boldsymbol{x}}^{(2,n^{(1)})}; \hat{\boldsymbol{\beta}}^{(1)}\right) = \\ 
        &\left(
        \frac{\left.\bigl({S_1^{(1)}+\hat{S}_1^{(2)}\bigl(\hat{\boldsymbol{x}}^{(2,n^{(1)})}, \hat{\boldsymbol{\beta}}^{(1)}\bigr) }\bigr)\right/{N_1} - \left.\bigl({S_0^{(1)}+\hat{S}_0^{(2)}\bigl(\hat{\boldsymbol{\beta}}^{(1)} \bigr)}\bigr)\right/{N_0} }{\sqrt{{ \frac{S_1^{(1)}+\hat{S}_1^{(2)}\bigl(\hat{\boldsymbol{x}}^{(2,n^{(1)})}, \hat{\boldsymbol{\beta}}^{(1)}\bigr)}{N_1} \Bigl(1-\frac{S_1^{(1)}+\hat{S}_1^{(2)}\bigl(\hat{\boldsymbol{x}}^{(2,n^{(1)})}, \hat{\boldsymbol{\beta}}^{(1)}\bigr)}{N_1} \Bigr)} \frac{1}{N_1} + { \frac{S_0^{(1)}+\hat{S}_0^{(2)}\bigl(\hat{\boldsymbol{\beta}}^{(1)} \bigr)}{N_0} \Bigl(1-\frac{S_0^{(1)}+\hat{S}_0^{(2)}\bigl(\hat{\boldsymbol{\beta}}^{(1)} \bigr)}{N_0}\Bigr)} \frac{1}{N_0}}}
        \right)^2. 
    \label{theorem 1 lambda}
    \end{aligned}
    \end{equation}
    }
    Under Assumptions 1, 3, 4, 5, 7, 8 from the main text, the stage 2 recommended intervention $\hat{\boldsymbol{x}}^{(2,n^{(1)})}$, subject to both an outcome goal and an unconditional power goal, solves the following optimization problem: 
    \begin{equation*}
        \text{Min}_{\boldsymbol{x}} C\left(\boldsymbol{x}\right)
        \; \text {subject to} \; 
        \operatorname{expit}\bigl(\hat{\beta}_{0}^{(1)}+(\hat{\boldsymbol{\beta}}^{(1)}_1)^T \hat{\boldsymbol{x}}^{(2,n^{(1)})} \bigr)
        \geq \tilde{p}, \;\text{and}\;\; {\lambda}\bigl(\hat{\boldsymbol{x}}^{(2,n^{(1)})}; \hat{\boldsymbol{\beta}}^{(1)}\bigr) \geq \lambda_{min} .
    \end{equation*}
\end{thm}
\begin{proof}
\textcolor{white}{xxx}\\
At the end of stage 1, not all data are available to calculate the test statistic $Z$ (equation (5) from the main text). For the stage 1-related terms in $Z$, $S_1^{(1)}$ and $S_0^{(1)}$, their means are estimated using the observed sums of the outcomes in the intervention and control groups from stage 1. For the stage 2-related terms in $Z$, $S_1^{(2)}$ and $S_0^{(2)}$, their means are estimated under $\hat{\boldsymbol{x}}^{(2, n^{(1)})}$ and $\hat{\boldsymbol{\beta}}^{(1)}$. Specifically, $\hat{S}_1^{(2)}\left(\hat{\boldsymbol{x}}^{(2,n^{(1)})}, \hat{\boldsymbol{\beta}}^{(1)}\right)$ from equation (\ref{mu_1_2}) and $\hat{S}_0^{(2)}\left(\hat{\boldsymbol{\beta}}^{(1)} \right)$ from equation (\ref{mu_0_2}) are used to estimate the means of ${S}_1^{(2)}$ and ${S}_0^{(2)}$, respectively.

Under the alternative hypothesis, $Z^2$ approximately follows a non-central $\chi^2$ distribution with 1 degree of freedom. 
The non-centrality parameter $\lambda$ can be estimated under $\hat{\boldsymbol{x}}^{(2, n^{(1)})}$ and $\hat{\boldsymbol{\beta}}^{(1)}$, by ${\lambda}\left(\hat{\boldsymbol{x}}^{(2,n^{(1)})}; \hat{\boldsymbol{\beta}}^{(1)}\right)$ from equation (\ref{theorem 1 lambda}).
To satisfy the unconditional power goal, the condition $P(Z^2>\chi^2_{\alpha, 1})\geq \Pi$ under the alternative hypothesis is ensured by achieving ${\lambda}\left(\hat{\boldsymbol{x}}^{(2,n^{(1)})}; \hat{\boldsymbol{\beta}}^{(1)}\right) \geq {\lambda}_{min}$. Combining equation (\ref{theorem 1 lambda}) with $
    \operatorname{expit}\bigl(\hat{\beta}_0^{(1)}+\bigl(\hat{\boldsymbol{\beta}}_1^{(1)}\bigr)^T \hat{\boldsymbol{x}}^{(2,n^{(1)})}\bigr)
    \geq \tilde{p}.
$ implies Theorem~\ref{unconditional power theorem appendix}.
\end{proof}

\subsection{Proof of Theorem 2}\label{proof of theorem 2 section}
For convenience, we restate Theorem 2 from the main text.

Recall that $z_{\Pi}$ is the critical value corresponding to the upper $\Pi$ tail probability of a standard normal distribution.
\begin{thm}[Conditional Power Approach]\label{conditional power goal theorem appendix}
    \begin{equation}
    \text{Let}\;\;
    \Delta\bigl(\hat{\boldsymbol{x}}^{(2,n^{(1)})}, \hat{\boldsymbol{\beta}}^{(1)}\bigr) =
    \left. \hat{S}_1^{(2)}\bigl(\hat{\boldsymbol{x}}^{(2,n^{(1)})}, \hat{\boldsymbol{\beta}}^{(1)}\bigr) \right/ N_1 - \left. \hat{S}_0^{(2)}\bigl(\hat{\boldsymbol{\beta}}^{(1)} \bigr) \right/ N_0 ,
    \label{long mu theorem 2}
    \end{equation}
    where $ \hat{S}_1^{(2)}\bigl(\hat{\boldsymbol{x}}^{(2,n^{(1)})}, \hat{\boldsymbol{\beta}}^{(1)}\bigr)$ and $\hat{S}_0^{(2)}\bigl(\hat{\boldsymbol{\beta}}^{(1)} \bigr)$ are defined in Theorem \ref{unconditional power theorem appendix}.
    \begin{equation}
    \begin{aligned}
    \text{Let}\;\;
    &\hat{\sigma}^2\bigl(\hat{\boldsymbol{x}}^{(2,n^{(1)})}, \hat{\boldsymbol{\beta}}^{(1)}\bigr) = 
    \frac{n_0^{(2)}}{N_0^2} \operatorname{expit}\left( \hat{\beta}_0^{(1)}\right)\left(1-\operatorname{expit}\left( \hat{\beta}_0^{(1)}\right)\right) \\
    & \hspace{1cm} +\frac{n_1^{(2)}}{N_1^2} \operatorname{expit}\bigl(\hat{\beta}_0^{(1)}+\bigl(\hat{\boldsymbol{\beta}}_1^{(1)}\bigr)^T \hat{\boldsymbol{x}}^{(2,n^{(1)})}\bigr)\bigl(1-\operatorname{expit}\bigl(\hat{\beta}_0^{(1)}+\bigl(\hat{\boldsymbol{\beta}}_1^{(1)}\bigr)^T \hat{\boldsymbol{x}}^{(2,n^{(1)})}\bigr)\bigr).
    \end{aligned}
    \label{long sigma theorem 2}
    \end{equation}
    Under Assumptions 1, 3, 4, 5, 7, 8 from the main text, the stage 2 recommended intervention $\hat{\boldsymbol{x}}^{(2,n^{(1)})}$, subject to both an outcome goal and the conditional power goal, solves the following optimization problem: 
    \begin{equation*}
        \text{Min}_{\boldsymbol{x}} C\left(\boldsymbol{x}\right)
        \; \text {subject to} \; 
        \operatorname{expit}\left(\hat{\beta}_{0}^{(1)}+(\hat{\boldsymbol{\beta}}^{(1)}_1)^T \hat{\boldsymbol{x}}^{(2,n^{(1)})} \right)
        \geq \tilde{p}, \;\text{and}
    \end{equation*}
    {\footnotesize
    \begin{equation}\label{power goal z test for proportion theorem 2}
    \scalebox{0.9}{$
    \begin{aligned}
        & z_{\alpha/2}
        \sqrt{{ \frac{S_1^{(1)}+\hat{S}_1^{(2)}\left(\hat{\boldsymbol{x}}^{(2,n^{(1)})}, \hat{\boldsymbol{\beta}}^{(1)}\right)}{N_1} \left(1-\frac{S_1^{(1)}+\hat{S}_1^{(2)}\left(\hat{\boldsymbol{x}}^{(2,n^{(1)})}, \hat{\boldsymbol{\beta}}^{(1)}\right)}{N_1} \right)} \frac{1}{N_1} + { \frac{S_0^{(1)}+\hat{S}_0^{(2)}\left(\hat{\boldsymbol{\beta}}^{(1)} \right)}{N_0} \left(1-\frac{S_0^{(1)}+\hat{S}_0^{(2)}\left(\hat{\boldsymbol{\beta}}^{(1)} \right)}{N_0}\right)} \frac{1}{N_0}}
        \\
        & \hspace{4cm} - \frac{S_1^{(1)}}{N_1} + \frac{S_0^{(1)}}{N_0} - \Delta\left(\hat{\boldsymbol{x}}^{(2,n^{(1)})}, \hat{\boldsymbol{\beta}}^{(1)}\right)
         - z_{\Pi} \; \hat{\sigma}^2\left(\hat{\boldsymbol{x}}^{(2,n^{(1)})}, \hat{\boldsymbol{\beta}}^{(1)}\right)
         \leq 0. \\ 
    \end{aligned}
    $}
    \end{equation}
    }
\end{thm}

\begin{proof}
\textcolor{white}{xxx}\\
Consider the test statistic $Z$ from equation (5) of the main text. 
The $H_0: p_1=p_0$ is rejected at significance level $\alpha$ when $|Z|>z_{\alpha / 2}$. Assuming the intervention aims to increase the success probability, a positive effect of the intervention on the outcome would be expected. Thus, under the alternative hypothesis for which the power is calculated, the probability of observing $Z<-z_{\alpha / 2}$ is usually negligible. Recall that $\overline{\boldsymbol{a}}^{(1)}$ and $\overline{\boldsymbol{Y}}^{(1)}$ are the actual interventions and outcomes for all stage 1 centers; $\hat{p}_1$ and $\hat{p}_0$ are the estimated success probabilities in the intervention and control groups in the final combined data from both stages; $n$ is the total number of participants across all stages. 
To satisfy the conditional power goal, the condition $P\left({Z}>z_{\alpha / 2} \;\middle\vert\; \overline{\boldsymbol{a}}^{(1)}, \overline{\boldsymbol{Y}}^{(1)} \right) \geq \Pi$ needs to be satisfied under the alternative hypothesis, where 
\begin{equation*}
\begin{aligned}
    Z = \frac{\hat{p}_1 - \hat{p}_0}{\sqrt{\frac{\hat{p}_1(1-\hat{p}_1)}{N_1} + \frac{\hat{p}_0(1-\hat{p}_0)}{N_0}}}
    = \frac{ \sqrt{n} \left(\hat{p}_1 - \hat{p}_0\right)}{\sqrt{n} \sqrt{\frac{\hat{p}_1(1-\hat{p}_1)}{N_1} + \frac{\hat{p}_0(1-\hat{p}_0)}{N_0}}}.
\end{aligned}
\end{equation*}
By the central limit theorem, as $n \rightarrow \infty$, $\sqrt{n} \left(\hat{p}_1 - \hat{p}_0\right)$ converges in distribution to a normal distribution. 
Let $r_1 = \lim_{n\rightarrow\infty}N_1/n$ and $r_0 = \lim_{n\rightarrow\infty}N_0/n = 1-r_1$. Then, as $n\rightarrow\infty$, 
$$\sqrt{n}  {\sqrt{\frac{\hat{p}_1(1-\hat{p}_1)}{N_1} + \frac{\hat{p}_0(1-\hat{p}_0)}{N_0}}} \xrightarrow{P} \sqrt{ \frac{p_1(1-p_1)}{r_1} + \frac{p_0(1-p_0)}{r_0}}.$$ By Slutsky's Theorem, the variance of the limiting distribution of $Z$ is determined by the numerator $\left(\hat{p}_1 - \hat{p}_0\right)$.
At the end of stage 1, $\hat{p}_1$ and $\hat{p}_0$ are unknown. To satisfy the condition $P\left({Z}>z_{\alpha / 2} \;\middle\vert\; \overline{\boldsymbol{a}}^{(1)}, \overline{\boldsymbol{Y}}^{(1)} \right) \geq \Pi$ under the alternative, the conditional power approach estimates the distribution of $\sqrt{n}\left(\hat{p}_1 - \hat{p}_0\right)$ and the value of $\sqrt{n}\sqrt{{{p}_1(1-{p}_1)}/N_1 + {{p}_0(1-{p}_0)}/N_0}$ under $\hat{\boldsymbol{x}}^{(2,n^{(1)})}$ and $\hat{\boldsymbol{\beta}}^{(1)}$.

$P\left({Z}>z_{\alpha / 2} \;\middle\vert\; \overline{\boldsymbol{a}}^{(1)}, \overline{\boldsymbol{Y}}^{(1)} \right) \geq \Pi$ is equivalent to 
\begin{equation}\label{conditional approach first long eq}
\begin{aligned}
    &P\left(\left(\hat{p}_1 - \hat{p}_0\right) >z_{\alpha / 2} {\sqrt{\frac{\hat{p}_1(1-\hat{p}_1)}{N_1} + \frac{\hat{p}_0(1-\hat{p}_0)}{N_0}}} \;\middle\vert\; \overline{\boldsymbol{a}}^{(1)}, \overline{\boldsymbol{Y}}^{(1)}  \right) \geq \Pi.
\end{aligned}
\end{equation}
Let 
\begin{equation*}
    G_1 = \sqrt{ { \frac{S_1^{(1)}+{S}_1^{(2)}}{N_1} \left(1-\frac{S_1^{(1)}+{S}_1^{(2)}}{N_1} \right)} \frac{1}{N_1} + { \frac{S_0^{(1)}+{S}_0^{(2)}}{N_0} \left(1-\frac{S_0^{(1)}+{S}_0^{(2)}}{N_0}\right)} \frac{1}{N_0} }.
\end{equation*}
Equation (\ref{conditional approach first long eq}) is equivalent to
\begin{equation}
\begin{aligned}
    & P \left( \frac{{S}_1^{(2)}}{N_1} - \frac{{S}_0^{(2)}}{N_0}
    > z_{\alpha/2} \; G_1 - \frac{S_1^{(1)}}{N_1} + \frac{S_0^{(1)}}{N_0} \;\middle\vert\; \overline{\boldsymbol{a}}^{(1)}, \overline{\boldsymbol{Y}}^{(1)}  \right) \geq \Pi.
\end{aligned}\label{conditional approach satisfy this theorem 2}
\end{equation}
At the end of stage 1, ${S}_1^{(2)}$, ${S}_0^{(2)}$ and $G_1$ from equation (\ref{conditional approach satisfy this theorem 2}) are unknown. 
The conditional power approach estimates the mean and variance of the distribution of $ \bigl({{S}_1^{(2)}}/{N_1} - {{S}_0^{(2)}}/{N_0} \bigl) $ based on $\hat{\boldsymbol{x}}^{(2,n^{(1)})}$ and $\hat{\boldsymbol{\beta}}^{(1)}$, denoted by $\Delta\bigl(\hat{\boldsymbol{x}}^{(2,n^{(1)})}, \hat{\boldsymbol{\beta}}_1^{(1)}\bigl)$ and $\hat{\sigma}^2\bigl(\hat{\boldsymbol{x}}^{(2,n^{(1)})}, \hat{\boldsymbol{\beta}}_1^{(1)}\bigl)$, respectively, and estimate the value of $G_1$ by replacing ${S}_0^{(2)}$ and ${S}_1^{(2)}$ in $G_1$ with $\hat{S}_1^{(2)}\bigl(\hat{\boldsymbol{x}}^{(2,n^{(1)})}, \hat{\boldsymbol{\beta}}^{(1)}\bigl)$ and $\hat{S}_0^{(2)}\bigl(\hat{\boldsymbol{\beta}}^{(1)} \bigl)$ (see equations (\ref{mu_1_2}) and (\ref{mu_0_2})).

The estimated asymptotic distributions of ${S}_0^{(2)}$ and ${S}_1^{(2)}$ based on $\hat{\boldsymbol{x}}^{(2,n^{(1)})}$ and $\hat{\boldsymbol{\beta}}^{(1)}$ are $\operatorname{bin}\bigl(n_0^{(2)}, \operatorname{expit}\bigl(\hat{\beta}_0^{(1)}\bigl) \bigl)$ and $\operatorname{bin}\Bigl(n_1^{(2)}, \operatorname{expit}\bigl(\hat{\beta}_0^{(1)}+\bigl(\hat{\boldsymbol{\beta}}_1^{(1)}\bigl)^T \hat{\boldsymbol{x}}^{(2,n^{(1)})}\bigl) \Bigl)$, respectively. By normal approximation of these binomial distributions, the estimated distributions of ${{S}_0^{(2)}}/{N_0} $ and ${{S}_1^{(2)}}/{N_1} $ based on $\hat{\boldsymbol{x}}^{(2,n^{(1)})}$ and $\hat{\boldsymbol{\beta}}^{(1)}$ are
{\small
\begin{equation*}
\begin{aligned}
& \mathcal{N}\left(\frac{n_0^{(2)} \operatorname{expit}\left(\hat{\beta}_0^{(1)}\right)}{N_0}, \frac{n_0^{(2)} \operatorname{expit}\left(\hat{\beta}_0^{(1)}\right) \left(1-\operatorname{expit}\left(\hat{\beta}_0^{(1)}\right)\right)}{N_0^2}\right),\\
& 
\text{and} \;\;
\mathcal{N}\left(\frac{n_1^{(2)} \operatorname{expit}\left(\hat{\beta}_0^{(1)}+\left(\hat{\boldsymbol{\beta}}_1^{(1)}\right)^T \hat{\boldsymbol{x}}^{(2,n^{(1)})}\right)}{N_1}, \right.\\
& \hspace{1cm} \left.\frac{n_1^{(2)} \operatorname{expit}\left(\hat{\beta}_0^{(1)}+\left(\hat{\boldsymbol{\beta}}_1^{(1)}\right)^T \hat{\boldsymbol{x}}^{(2,n^{(1)})}\right) \left(1-\operatorname{expit}\left(\hat{\beta}_0^{(1)}+\left(\hat{\boldsymbol{\beta}}_1^{(1)}\right)^T \hat{\boldsymbol{x}}^{(2,n^{(1)})}\right)\right)}{N_1^2}\right),
\;\text{respectively}.
\end{aligned}
\end{equation*}
}
Thus,
$\Delta\bigl(\hat{\boldsymbol{x}}^{(2,n^{(1)})}, \hat{\boldsymbol{\beta}}^{(1)}\bigl)$ and $\hat{\sigma}^2\bigl(\hat{\boldsymbol{x}}^{(2,n^{(1)})}, \hat{\boldsymbol{\beta}}^{(1)}\bigl)$ from equations (\ref{long mu theorem 2}) and (\ref{long sigma theorem 2}) are the mean and variance of the estimated distribution of $ \bigl({{S}_1^{(2)}}/{N_1} - {{S}_0^{(2)}}/{N_0} \bigl) $ based on the observed stage 1 data. 
Let 
\begin{equation}\label{define G2 theorem 2}
G_2=\left(\left(\frac{{S}_1^{(2)}}{N_1}-\frac{{S}_0^{(2)}}{N_0}\right)-\Delta\left(\hat{\boldsymbol{x}}^{(2,n^{(1)})}, \hat{\boldsymbol{\beta}}^{(1)}\right)\right) \frac{1}{\hat{\sigma}^2\left(\hat{\boldsymbol{x}}^{(2,n^{(1)})}, \hat{\boldsymbol{\beta}}^{(1)}\right)}.
\end{equation}
Equation (\ref{conditional approach satisfy this theorem 2}) can be reformulated as
\begin{equation}\label{conditional approach std normal satisfy this theorem 2}
    \begin{aligned}
        &P \left( G_2
        > 
        \left(
        z_{\alpha/2} \; G_1 - \frac{S_1^{(1)}}{N_1} + \frac{S_0^{(1)}}{N_0} - \Delta\left(\hat{\boldsymbol{x}}^{(2,n^{(1)})}, \hat{\boldsymbol{\beta}}^{(1)}\right)  \right)  
        \frac{1}{\hat{\sigma}^2\left(\hat{\boldsymbol{x}}^{(2,n^{(1)})}, \hat{\boldsymbol{\beta}}^{(1)}\right)} \;\middle\vert\; \overline{\boldsymbol{a}}^{(1)}, \overline{\boldsymbol{Y}}^{(1)} 
        \right) \geq \Pi.
    \end{aligned}
\end{equation}
Let $\Phi(\cdot)$ be the cumulative distribution function of the standard normal distribution. 
Since interventions are not adapted in stage 1, standard asymptotic theory implies that the stage 1-based estimator $\hat{\boldsymbol{\beta}}^{(1)}$ converges in probability to the true parameter $\boldsymbol{\beta}^*$.
By the Central Limit Theorem, $G_2$ (equation (\ref{define G2 theorem 2})) is approximately normally distributed with mean 0 and variance 1. 
Using $\Phi(\cdot)$, equation (\ref{conditional approach std normal satisfy this theorem 2}) is equivalent to
\begin{equation}\label{cdf equation theorem 2}
    \scalebox{0.95}{$
    \begin{aligned}
        &1- \Phi\left( \left( z_{\alpha/2} \; G_1 - \frac{S_1^{(1)}}{N_1} + \frac{S_0^{(1)}}{N_0} - \Delta\left(\hat{\boldsymbol{x}}^{(2,n^{(1)})}, \hat{\boldsymbol{\beta}}^{(1)}\right) \right) \frac{1}{\hat{\sigma}^2\left(\hat{\boldsymbol{x}}^{(2,n^{(1)})}, \hat{\boldsymbol{\beta}}^{(1)}\right)} \;\middle\vert\; \overline{\boldsymbol{a}}^{(1)}, \overline{\boldsymbol{Y}}^{(1)} 
        \right) \geq 1-\Phi\left(z_{\Pi}\right),
    \end{aligned}
    $}
\end{equation}
because $\Pi=\Phi\left(-z_{\Pi}\right)=1-\Phi\left(z_{\Pi}\right)$. Equation (\ref{cdf equation theorem 2}) implies that $\hat{\boldsymbol{x}}^{(2,n^{(1)})}$ needs to satisfy
\footnotesize
\begin{equation}\label{conditional power goal final derivation}
    \scalebox{0.9}{$
    \begin{aligned}
        & \left(z_{\alpha/2} \; \sqrt{{ \frac{S_1^{(1)}+\hat{S}_1^{(2)}\left(\hat{\boldsymbol{x}}^{(2,n^{(1)})}, \hat{\boldsymbol{\beta}}^{(1)}\right) }{N_1} \left(1-\frac{S_1^{(1)}+\hat{S}_1^{(2)}\left(\hat{\boldsymbol{x}}^{(2,n^{(1)})}, \hat{\boldsymbol{\beta}}^{(1)}\right)}{N_1} \right)} \frac{1}{N_1} + { \frac{S_0^{(1)}+\hat{S}_0^{(2)}\left(\hat{\boldsymbol{\beta}}^{(1)} \right)}{N_0} \left(1-\frac{S_0^{(1)}+\hat{S}_0^{(2)}\left(\hat{\boldsymbol{\beta}}^{(1)} \right)}{N_0}\right)} \frac{1}{N_0}}\right.\\ 
        & \hspace{5cm} - \frac{S_1^{(1)}}{N_1} + \frac{S_0^{(1)}}{N_0} - \Delta\left(\hat{\boldsymbol{x}}^{(2,n^{(1)})}, \hat{\boldsymbol{\beta}}^{(1)}\right) \scalebox{1.2}{$\Bigg)$} \frac{1}{\hat{\sigma}^2\left(\hat{\boldsymbol{x}}^{(2,n^{(1)})}, \hat{\boldsymbol{\beta}}^{(1)}\right)}
         \leq z_{\Pi}.
    \end{aligned}
    $}
\end{equation}
\normalsize
Combining equation (\ref{conditional power goal final derivation}) with $
    \operatorname{expit}\bigl(\hat{\beta}_0^{(1)}+\bigl(\hat{\boldsymbol{\beta}}_1^{(1)}\bigr)^T \hat{\boldsymbol{x}}^{(2,n^{(1)})}\bigr)
    \geq \tilde{p}.
$ implies Theorem \ref{conditional power goal theorem appendix}.
\end{proof}

In summary, the method for calculating $\hat{\boldsymbol{x}}^{(2,n^{(1)})}$ based on an outcome goal and a power goal, using the conditional power approach is:
\begin{enumerate}
    \item Given $\alpha$, find the critical values $z_{\alpha/2}$ and $z_{\Pi}$ from the standard normal distribution.
    For $\alpha=0.05$ and $\Pi=0.8$, $z_{\alpha/2} = 1.96$ and $z_{\Pi} = -0.84$.

    \item Solve for the value of $\tilde{p}_{pow\_c}$ so that for $ \hat{\boldsymbol{x}}^{(2,n^{(1)})}$ that satisfies $\operatorname{expit}\left(\hat{\beta}_{0}^{(1)}+(\hat{\boldsymbol{\beta}}^{(1)}_1)^T \hat{\boldsymbol{x}}^{(2,n^{(1)})} \right) \geq \tilde{p}_{pow\_c}$, the constraint from equation (\ref{power goal z test for proportion theorem 2}) is satisfied.

    \item Calculate the maximum achievable success probability, $\hat{p}_{max}\bigl(\hat{\boldsymbol{\beta}}^{(1)}\bigl)$ (see equation (10) from the main text for details).

    \item Calculate the stage 2 recommended intervention $\hat{\boldsymbol{x}}^{(2,n^{(1)})}$ based on Algorithm 1 from Section 3.1 of the main text, using $\tilde{p}_{pow\_c}$ in place of $\tilde{p}_{pow\_u}$. We refer to this modified algorithm as Algorithm 2.
\end{enumerate}

\section{Details of the Shrinking Method}\label{shrinking method section}
As described in Section 3.1 of the main text, when $\hat{p}_{max}\bigl(\hat{\boldsymbol{\beta}}^{(1)}\bigl) < \tilde{p} $, Algorithm 1 uses the shrinking method from \citet{nevo2021analysis} to calculate the recommended intervention $\hat{\boldsymbol{x}}^{(2, n^{(1)})}$. The shrinking method was first introduced in Section 5.1 of the Appendixs of \citet{nevo2021analysis}; here we provide more details.

For ease of presentation, we consider the case where the intervention package has 2 components ($P=2$), and $\hat{\beta}_{11}^{(1)}$ and $\hat{\beta}_{12}^{(1)}$ are positive. Other cases follow in a similar manner. Based on Algorithm 1 from Section 3.1 of the main text, the shrinking method is used when $\hat{\beta}_{0}^{(1)} + \hat{\beta}_{11}^{(1)}\mathcal{U}_1+\hat{\beta}_{12}^{(1)}\mathcal{U}_2 < logit(\tilde{p})$, where $\mathcal{U}_p$ is the pre-specified upper bound of the $p$-th intervention component.

We show the shrinking method for the first intervention component. The shrinking method for the second intervention component follows in a similar manner. 
Let $\tilde{X}^{(2,n^{(1)})}_{11}$ be the solution of 
$$
\hat{\beta}_{0}^{(1)}+\hat{\beta}_{11}^{(1)} \tilde{X}_{11}^{(2,n^{(1)})}+I\left\{\hat{\beta}_{12}^{(1)}>0\right\} \hat{\beta}_{12}^{(1)} \mathcal{U}_2+I\left\{\hat{\beta}_{12}^{(1)}<0\right\} \hat{\beta}_{12}^{(1)} \mathcal{L}_2=\operatorname{logit}(\tilde{p}).
$$
That is, $\tilde{X}_{11}^{(2,n^{(1)})}$ is the value of the first intervention component such that the outcome goal $\tilde{p}$ is reached, regardless of the pre-specified upper and lower bounds of the first intervention component, and when the second intervention component is on the boundary that is most effective.

Let $\beta_{11}^{max}$ be the solution of 
$$
\beta_{11}^{max} \mathcal{U}_1+I\left\{\hat{\beta}_{12}^{(1)}>0\right\} \hat{\beta}_{12}^{(1)} \mathcal{U}_2+I\left\{\hat{\beta}_{12}^{(1)}<0\right\} \hat{\beta}_{12}^{(1)} \mathcal{L}_2=\operatorname{logit}(\tilde{p}).
$$
Let $\beta_{11}^{min} = \beta_{11}^{max}/2$.
As shown in Figure \ref{fig:shrinking}, the shrinking method uses a function $m(\hat{\boldsymbol{\beta}})$, that is continuous in $\boldsymbol{\beta}$, to calculate the recommended intervention $\hat{\boldsymbol{x}}^{(2, n^{(1)})}$.

\begin{figure}[h]
    \centering
    \includegraphics[width=0.8\linewidth]{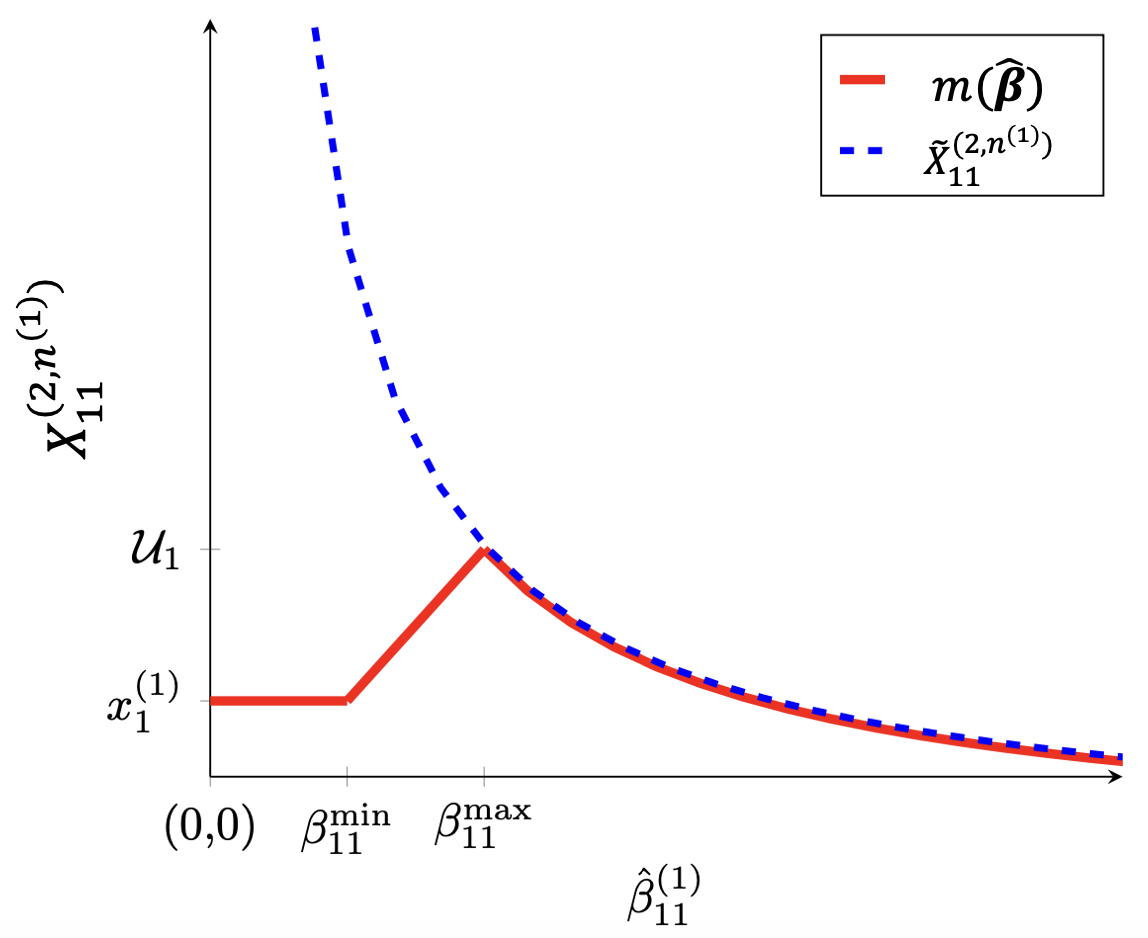}
    \caption{Illustration of the shrinking method for the first intervention component.}
    \label{fig:shrinking}
\end{figure}

The intuition behind the function $m(\hat{\boldsymbol{\beta}})$ is that when the value of $\hat{\beta}_{11}^{(1)}$ is close to 0, the stage 1 recommended intervention will be used in stage 2. When the value of $\hat{\beta}_{11}^{(1)}$ falls between $\beta_{11}^{min}$ and $\beta_{11}^{max}$, the value of the stage 2 recommended intervention is shrunk from the maximum $\mathcal{U}_1$ value toward the stage 1 recommended intervention. 

In summary, the shrinking method works as follows. For each $p= 1, \ldots, P$, 
let $\beta^{max}_{1p}$ be the solution of
\begin{equation*}\label{get beta max equation}
 \begin{aligned}
     \hat\beta_{0}^{(1)} + 
     \beta^{max}_{1p} \mathcal{U}_p + \sum_{\substack{q=1\\q \neq p}}^{P}  I\{ \hat\beta_{1q}^{(1)}>0\}\hat\beta_{1q}^{(1)}\mathcal{U}_q + 
     I\{ \hat\beta_{1q}^{(1)}<0\}\hat\beta_{1q}^{(1)}\mathcal{L}_q
     = logit(\tilde{p}).
 \end{aligned}
\end{equation*}
Let $\beta^{min}_{1p} = \beta^{max}_{1p} / 2$, and 
\begin{equation*}
    \begin{aligned}
        \hat{x}^{(2, n^{(1)})}_p = 
        \begin{cases}
            {x}^{(1)}_p &\text{if} \; \hat\beta_{1p}^{(1)} \leq \beta^{min}_{1p}, \\
            {x}^{(1)}_p + \frac{\mathcal{U}_p - {x}^{(1)}_p}{\beta^{max}_{1p} - \beta^{min}_{1p}} \bigl(\hat\beta_{1p}^{(1)} - \beta^{min}_{1p}\bigr) &\text{otherwise}.
        \end{cases}
    \end{aligned}
\end{equation*}

\section{Proof of Asymptotic Properties of \texorpdfstring{$\hat{\boldsymbol{\beta}}^{(2)}$}{beta}}\label{asymp properties section}

\begin{thm}\label{asymptotic properties of beta theorem}
    Under Assumptions 1-8 from the main text, $\hat{\boldsymbol{\beta}}^{(2)}$ is both consistent and asymptotically normal.
\end{thm}

\begin{proof}
\textcolor{white}{xxx}\\
Let ${p}_{max}(\boldsymbol{\beta}^*)$ denote the maximum achievable success probability within the pre-specified bounds $\mathcal{L}_p$ and $\mathcal{U}_p$ of the intervention components $p=1,\ldots, P$, based on the true parameter $\boldsymbol{\beta}^*$.
Let $p_{no}(\boldsymbol{\beta}^*)$ and $\hat{p}_{no}\bigl(\hat{\boldsymbol{\beta}}^{(1)}\bigl)$ denote the success probability without any intervention based on the true parameter $\boldsymbol{\beta}^*$ and its estimate from the stage 1 data $\hat{\boldsymbol{\beta}}^{(1)}$, respectively.

To see that Assumption 4 from the main text holds when $\hat{\boldsymbol{x}}^{(2, n^{(1)})}$ is calculated using Algorithm 1 from Section 3.1 of the main text, first note that the stage 1 interventions of a LAGO trial are fixed in advance. Standard asymptotic theory implies that as the sample size increases, the stage 1-based estimator $\hat{\boldsymbol{\beta}}^{(1)}$ converges in probability to the true parameter $\boldsymbol{\beta}^*$. This convergence holds, e.g., when $\hat{\boldsymbol{\beta}}^{(1)}$ is either 1) the maximum likelihood estimator based on the stage 1 data, or 2) a solution to generalized estimating equations \citep{liang1986longitudinal} based on the stage 1 data.
Therefore, by the Continuous Mapping Theorem, as $n \xrightarrow{} \infty$, $\hat{p}_{max}\bigl(\hat{\boldsymbol{\beta}}^{(1)}\bigl) \xrightarrow{} {p}_{max}\left(\boldsymbol{\beta}^*\right)$ and $\hat{p}_{no}\bigl(\hat{\boldsymbol{\beta}}^{(1)}\bigl) \xrightarrow{} p_{no}\left(\boldsymbol{\beta}^*\right)$.

\begin{enumerate}
    \item If ${p}_{max}\left(\boldsymbol{\beta}^*\right) \geq \tilde{p}$, then by Assumption 8 from the main text, ${p}_{max}\left(\boldsymbol{\beta}^*\right) \geq \tilde{p} > p_{no}\left(\boldsymbol{\beta}^*\right)$. 
    For the unconditional power approach, the power in equation (3) from the main text is calculated as the unconditional $P(Z^2>\chi^2_{\alpha, 1})$ under the alternative hypothesis, or
    \begin{equation}\label{power equation proof}
    \begin{aligned}
        P\left(\left( \frac{\hat{p}_1 - \hat{p}_0}{\sqrt{\frac{\hat{p}_1(1-\hat{p}_1)}{N_1} + \frac{\hat{p}_0(1-\hat{p}_0)}{N_0}}}\right)^2>\chi^2_{\alpha, 1}\right).
    \end{aligned}
    \end{equation}
    Recall from Section 2.3 of the main text that $\alpha_{1}^{(k)} = \lim_{n \rightarrow \infty} n_1^{(k)}/N_1 > 0$. 
    Let $p_1^{(1)}\left(\boldsymbol{\beta}^*\right)$ denote the success probability in stage 1 based on the true parameter $\boldsymbol{\beta}^*$.
    Asymptotically, when the outcome goal can be reached,
    $\hat{p}_1 \xrightarrow{} \alpha_{1}^{(1)} p_1^{(1)}\left(\boldsymbol{\beta}^*\right) + \alpha_{1}^{(2)} \tilde{p}$, and 
    $\hat{p}_0 \xrightarrow{} p_{no}\left(\boldsymbol{\beta}^*\right)$.
    Because $p_1^{(1)}\left(\boldsymbol{\beta}^*\right) \geq p_{no}\left(\boldsymbol{\beta}^*\right)$, 
    as $n \xrightarrow{} \infty$, since $\tilde{p} > p_{no}\left(\boldsymbol{\beta}^*\right)$, $\hat{p}_1-\hat{p}_0 > 0$ and the positive denominator in equation (\ref{power equation proof}) goes to 0, it follows that the probability in equation (\ref{power equation proof}) goes to 1. 
    Thus, asymptotically, when the outcome goal is satisfied, the power goal is also satisfied.
    That is, for any $\epsilon > 0$, $\exists$ a sample size $n$ such that for all sample sizes $\geq n$, $\hat{\boldsymbol{x}}^{(2, n^{(1)})}$ based on only the outcome goal will also satisfy the power goal in equation (3) from the main text with probability $>1-\epsilon$. Then by Assumption 7 from the main text, Assumption 4 from the main text holds.
    \item If ${p}_{max}\left({\boldsymbol{\beta}}^{*}\right) < \tilde{p} $, asymptotically, as the sample size in stage 1 increases, Algorithm 1 from Section 3.1 of the main text uses the shrinking method from \citet{nevo2021analysis} to ensure that $\hat{\boldsymbol{x}}^{(2, n^{(1)})}$ is a continuous function of ${\boldsymbol{\beta}}$, thus satisfying Assumption 4 from the main text. 
\end{enumerate}

The same arguments apply when $\hat{\boldsymbol{x}}^{(2, n^{(1)})}$ is calculated based on the conditional power approach, using Algorithm 2 described in Section \ref{proof of theorem 2 section}. Therefore, Assumption 4 from the main text also holds when Algorithm 2 is used. Based on the results from \citet{nevo2021analysis, bing2023learnasyougo}, the estimator $\hat{\boldsymbol{\beta}}^{(2)}$ is consistent and asymptotically normal. 
\end{proof}

\section{Examples of LAGO Trials with Power Goals of Alternative Tests} \label{additional tests}
Appendix \ref{additional tests} describes in detail how the power goal works with the following tests: 
the two-sample z-test for the difference between two proportions with pooled variance (Section \ref{z test for proportions}), 
the Wald-type $P$ degrees-of-freedom test for binary outcomes (Section \ref{wald binary outcomes section}), 
the two-sample t-test for the difference between two means with unpooled variance (Section \ref{t test unpooled section}), 
the two-sample t-test for the difference between two means with pooled variance (Section \ref{two-sample t test Appendix}), 
and the Wald-type $P$ degrees-of-freedom test for continuous outcomes (Section \ref{wald cts outcomes section}).
Other tests can also be considered.

\subsection{Power Goal of the Two-Sample z-test for the Difference between Two Proportions with Pooled Variance}\label{z test for proportions}

Let
\begin{equation}
    \hat{p}_{pool} = 
    \frac{S_1^{(1)}+S_1^{(2)} + S_0^{(1)}+S_0^{(2)}}{N_1+N_0}.
    \label{p_pool}
\end{equation}
The test statistic of the two-sample z-test for the difference between two proportions with pooled variance is
\begin{equation}\label{z pool definition}
    Z_{pool} = \frac{\hat{p}_1 - \hat{p}_0}{\sqrt{ \hat{p}_{pool} (1-\hat{p}_{pool}) \left( \frac{1}{N_1} + \frac{1}{N_0} \right) }}.
\end{equation}

The unconditional power and conditional power approaches are discussed separately. First, consider the unconditional power approach. Let 
    \begin{equation}\label{p hat pool estimated}
        \hat{p}_{pool}(\hat{\boldsymbol{x}}^{\left(2, n^{(1)}\right)}, \hat{\boldsymbol{\beta}}^{(1)}) = \frac{ S_{1}^{(1)} + \hat{S}_1^{(2)}(\hat{\boldsymbol{x}}^{(2,n^{(1)})}, \hat{\boldsymbol{\beta}}^{(1)}) +S_{0}^{(1)} + \hat{S}_0^{(2)}(\hat{\boldsymbol{\beta}}^{(1)}) }{N_1+N_0}.
    \end{equation}
\begin{thm}{Two-sample z-test with pooled variance, difference between two proportions, unconditional power goal.}\label{modified unconditional power theorem}
    \textcolor{white}{xxx}\\
    Let $F_{\chi^2}(k; \nu, \lambda) = P(Z^2 \leq k)$ be the CDF of the non-central $\chi^2$ distribution with $\nu$ degrees of freedom and non-centrality parameter $\lambda$.
    Under Assumptions 1, 3, 4, 5, 7 and 8 from the main text, the stage 2 recommended intervention $\hat{\boldsymbol{x}}^{(2,n^{(1)})}$, subject to both an outcome goal and an unconditional power goal, is obtained by solving the following optimization problem: 
    \begin{equation*}
        \text{Min}_{\boldsymbol{x}} C\left(\boldsymbol{x}\right)
        \; \text {subject to} \; 
        \operatorname{expit}(\hat{\beta}_{0}^{(1)}+(\hat{\boldsymbol{\beta}}^{(1)}_1)^T {\hat{\boldsymbol{x}}^{(2,n^{(1)})}} )
        \geq \tilde{p}, \;\text{and} 
    \end{equation*}
    \begin{equation}\label{pooled z test power condition long}
    \begin{aligned}
        &1 - \\
        &F_{\chi^2}\left( \chi^2_{\alpha, 1} \frac{{ \hat{p}_{pool}(\hat{\boldsymbol{x}}^{\left(2, n^{(1)}\right)}, \hat{\boldsymbol{\beta}}^{(1)})  (1-\hat{p}_{pool}(\hat{\boldsymbol{x}}^{\left(2, n^{(1)}\right)}, \hat{\boldsymbol{\beta}}^{(1)}) ) ( \frac{1}{N_1} + \frac{1}{N_0} ) }}{{{\frac{S_1^{(1)}+\hat{S}_1^{(2)}(\hat{\boldsymbol{x}}^{(2,n^{(1)})}, \hat{\boldsymbol{\beta}}^{(1)})}{N_1} \left(1-\frac{S_1^{(1)}+\hat{S}_1^{(2)}(\hat{\boldsymbol{x}}^{(2,n^{(1)})}, \hat{\boldsymbol{\beta}}^{(1)})}{N_1}\right)}\frac{1}{N_1} + {\frac{S_0^{(1)}+\hat{S}_0^{(2)}(\hat{\boldsymbol{\beta}}^{(1)})}{N_0}\left(1-\frac{S_0^{(1)}+\hat{S}_0^{(2)}(\hat{\boldsymbol{\beta}}^{(1)})}{N_0}\right)}\frac{1}{N_0}}} ; \right.\\
        &\qquad\qquad 1, {\lambda}\bigl(\hat{\boldsymbol{x}}^{(2,n^{(1)})}; \hat{\boldsymbol{\beta}}^{(1)}\bigl)  \scalebox{1.4}{$\Bigg)$} 
        \geq \Pi,
    \end{aligned}
    \end{equation}
    where $\hat{S}_1^{(2)}(\hat{\boldsymbol{x}}^{(2,n^{(1)})}, \hat{\boldsymbol{\beta}}^{(1)})$, $\hat{S}_0^{(2)}(\hat{\boldsymbol{\beta}}^{(1)} )$, and ${\lambda}\bigl(\hat{\boldsymbol{x}}^{(2,n^{(1)})}; \hat{\boldsymbol{\beta}}^{(1)}\bigl)$ are defined in Theorem 1 of the main text.
\end{thm}

\begin{proof}
    \textcolor{white}{xxx}\\
    To satisfy the unconditional power goal, the condition $P(Z_{pool}^2>\chi^2_{\alpha, 1})\geq \Pi$ needs to be satisfied under the alternative hypothesis.  

    As compared to Theorem 1 from Section 3 of the main text, $Z^2_{pool}$ does not approximately follow a non-central $\chi^2$ distribution with 1 degree of freedom under the alternative hypothesis, since the variance estimate for the numerator of $Z_{pool}$ is biased under the alternative hypothesis.
 
    Notice that with $Z$ as defined in Section 2.3 of the main text,
    \begin{equation*}
    \begin{aligned}
            Z_{pool} &= Z_{pool} \; \frac{\sqrt{ \hat{p}_{pool} (1-\hat{p}_{pool}) \left( \frac{1}{N_1} + \frac{1}{N_0} \right) }}{\sqrt{\frac{\hat{p}_1(1-\hat{p}_1)}{N_1} + \frac{\hat{p}_0(1-\hat{p}_0)}{N_0}}} \; \frac{\sqrt{\frac{\hat{p}_1(1-\hat{p}_1)}{N_1} + \frac{\hat{p}_0(1-\hat{p}_0)}{N_0}}}{\sqrt{ \hat{p}_{pool} (1-\hat{p}_{pool}) \left( \frac{1}{N_1} + \frac{1}{N_0} \right) }} \\ 
            &= Z \; \frac{\sqrt{\frac{\hat{p}_1(1-\hat{p}_1)}{N_1} + \frac{\hat{p}_0(1-\hat{p}_0)}{N_0}}}{\sqrt{ \hat{p}_{pool} (1-\hat{p}_{pool}) \left( \frac{1}{N_1} + \frac{1}{N_0} \right) }},
    \end{aligned}
    \end{equation*}
    or,
    \begin{equation*}
    \begin{aligned}
        Z_{pool}^2 = Z^2 \; \frac{{\frac{\hat{p}_1(1-\hat{p}_1)}{N_1} + \frac{\hat{p}_0(1-\hat{p}_0)}{N_0}}}{{ \hat{p}_{pool} (1-\hat{p}_{pool}) \left( \frac{1}{N_1} + \frac{1}{N_0} \right) }}.
    \end{aligned}
    \end{equation*}
    Therefore, the condition $P(Z^2_{pool}>\chi^2_{\alpha, 1})\geq \Pi$ is equivalent to
    \begin{equation*}
    \begin{aligned}
        P\left(Z^2 \; \frac{{\frac{\hat{p}_1(1-\hat{p}_1)}{N_1} + \frac{\hat{p}_0(1-\hat{p}_0)}{N_0}}}{{ \hat{p}_{pool} (1-\hat{p}_{pool}) \left( \frac{1}{N_1} + \frac{1}{N_0} \right) }}>\chi^2_{\alpha, 1}\right)\geq \Pi,
    \end{aligned}
    \end{equation*}
    or 
    \begin{equation}\label{cond new}
    \begin{aligned}
        P\left(Z^2 >\chi^2_{\alpha, 1} \frac{{ \hat{p}_{pool} (1-\hat{p}_{pool}) \left( \frac{1}{N_1} + \frac{1}{N_0} \right) }}{{\frac{\hat{p}_1(1-\hat{p}_1)}{N_1} + \frac{\hat{p}_0(1-\hat{p}_0)}{N_0}}}\right)\geq \Pi.
    \end{aligned}
    \end{equation}
    As stated in the proof for Theorem 1 of the main text, under the alternative hypothesis, $Z^2$ approximately follows a non-central $\chi^2$ distribution with 1 degree of freedom and non-centrality parameter $\lambda$ from Theorem 1 of the main text. 
    Equation (\ref{cond new}) is equivalent to 
    \begin{equation*}
    \begin{aligned}
        1 - F_{\chi^2}\left( \chi^2_{\alpha, 1} \frac{{ \hat{p}_{pool} (1-\hat{p}_{pool}) \left( \frac{1}{N_1} + \frac{1}{N_0} \right) }}{{\frac{\hat{p}_1(1-\hat{p}_1)}{N_1} + \frac{\hat{p}_0(1-\hat{p}_0)}{N_0}}} ; 1, \lambda \right)
        \geq \Pi,
    \end{aligned}
    \end{equation*}
    where $F_{\chi^2}(k; \nu, \lambda) = P(Z^2 \leq k)$ is the CDF of the non-central $\chi^2$ distribution with $\nu$ degrees of freedom and non-centrality parameter $\lambda$ from Theorem 1 of the main text. 

    At the end of stage 1, $\hat{p}_{pool}$, $\hat{p}_1$, $\hat{p}_0$, and $\lambda$ are all unknown. 
    The best guesses for $\hat{p}_{pool}$, $\hat{p}_1$, $\hat{p}_0$ under the stage 2 recommended intervention $\hat{\boldsymbol{x}}^{(2,n^{(1)})}$ are $\hat{p}_{pool}(\hat{\boldsymbol{x}}^{\left(2, n^{(1)}\right)}, \hat{\boldsymbol{\beta}}^{(1)})$ of equation (\ref{p hat pool estimated}),
    $\left.({S_1^{(1)}+\hat{S}_1^{(2)}(\hat{\boldsymbol{x}}^{(2,n^{(1)})}, \hat{\boldsymbol{\beta}}^{(1)})})\right/{N_1}$, 
    and $\left.({S_0^{(1)}+\hat{S}_0^{(2)}(\hat{\boldsymbol{\beta}}^{(1)})})\right/{N_0}$, respectively.
    Similar to the proof of Theorem 1 from the main text, the non-centrality parameter $\lambda$ of the non-central $\chi^2$ distribution of $Z^2$ can be estimated under $\hat{\boldsymbol{x}}^{(2, n^{(1)})}$, based on the stage 1 data, by ${\lambda}\bigl(\hat{\boldsymbol{x}}^{(2,n^{(1)})}; \hat{\boldsymbol{\beta}}^{(1)}\bigl)$ of equation (8) of the main text.
    
    Therefore, to satisfy the unconditional power goal, the condition $P(Z_{pool}^2>\chi^2_{\alpha, 1})\geq \Pi$ is estimated to be satisfied by ensuring that equation (\ref{pooled z test power condition long}) holds. Combining with the constraint imposed by the outcome goal implies Theorem \ref{modified unconditional power theorem}.
\end{proof}
    
\bigskip 
Next, consider the conditional power approach.
Recall that $z_{\alpha}$ is the critical value corresponding to the upper $\alpha$ tail probability of a standard normal distribution: for $Z \sim \mathcal{N}(0,1)$, $P(Z > z_{\alpha}) = \alpha$.
\begin{thm}{Two-sample z-test with pooled variance, difference between two proportions, conditional power goal.}\label{modified conditional power goal theorem}
    \textcolor{white}{xxx}\\
    Let 
    \begin{equation*}
    \Delta(\hat{\boldsymbol{x}}^{(2,n^{(1)})}, \hat{\boldsymbol{\beta}}^{(1)}) = 
    \frac{n_1^{(2)}}{N_1} \operatorname{expit}(\hat{\beta}_0^{(1)}+(\hat{\boldsymbol{\beta}}_1^{(1)})^T \hat{\boldsymbol{x}}^{(2,n^{(1)})}) -\frac{n_0^{(2)}}{N_0} \operatorname{expit}( \hat{\beta}_0^{(1)}), \\
    \end{equation*}
    \begin{equation*}
    \begin{aligned}
    \text{and}\;\; &\hat{\sigma}^2(\hat{\boldsymbol{x}}^{(2,n^{(1)})}, \hat{\boldsymbol{\beta}}^{(1)}) = 
    \frac{n_0^{(2)}}{N_0^2} \operatorname{expit}( \hat{\beta}_0^{(1)})(1-\operatorname{expit}( \hat{\beta}_0^{(1)})) \\
    & \hspace{2cm} +\frac{n_1^{(2)}}{N_1^2} \operatorname{expit}(\hat{\beta}_0^{(1)}+(\hat{\boldsymbol{\beta}}_1^{(1)})^T \hat{\boldsymbol{x}}^{(2,n^{(1)})})(1-\operatorname{expit}(\hat{\beta}_0^{(1)}+(\hat{\boldsymbol{\beta}}_1^{(1)})^T \hat{\boldsymbol{x}}^{(2,n^{(1)})})).
    \end{aligned}
    \end{equation*}
    Under Assumptions 1, 3, 4, 5, 7 and 8 from the main text, the stage 2 recommended intervention $\hat{\boldsymbol{x}}^{(2,n^{(1)})}$, subject to both an outcome goal and the conditional power goal from equation (3) of the main text, solves the following optimization problem: 
    \begin{equation*}
        \text{Min}_{\boldsymbol{x}} C\left(\boldsymbol{x}\right)
        \; \text {subject to} \; 
        \operatorname{expit}\left(\hat{\beta}_{0}^{(1)}+(\hat{\boldsymbol{\beta}}^{(1)}_1)^T {\hat{\boldsymbol{x}}^{(2,n^{(1)})}} \right)
        \geq \tilde{p}, \;\text{and}
    \end{equation*}
    \begin{equation*}
    \begin{aligned}
        & z_{\alpha/2}
        \sqrt{
        \left( \frac{1}{N_1} + \frac{1}{N_0} \right)\left(\hat{p}_{pool}(\hat{\boldsymbol{x}}^{\left(2, n^{(1)}\right)})\left(1-\hat{p}_{pool}(\hat{\boldsymbol{x}}^{\left(2, n^{(1)}\right)})\right)\right)
        } \\
        & \hspace{3cm} - \frac{S_1^{(1)}}{N_1} + \frac{S_0^{(1)}}{N_0} - \Delta\left(\hat{\boldsymbol{x}}^{(2,n^{(1)})}, \hat{\boldsymbol{\beta}}^{(1)}\right)
         - z_{\Pi} \; \hat{\sigma}\left(\hat{\boldsymbol{x}}^{(2,n^{(1)})}, \hat{\boldsymbol{\beta}}^{(1)}\right)
         \leq 0, \\ 
    \end{aligned}
    \end{equation*}
    where $\hat{p}_{pool}\bigl(\hat{\boldsymbol{x}}^{\left(2, n^{(1)}\right)}, \hat{\boldsymbol{\beta}}^{(1)}\bigl)$ is defined in equation (\ref{p hat pool estimated}).
\end{thm}
The proof of Theorem \ref{modified conditional power goal theorem} follows the same steps as that of Theorem 2 in Section 3.2 of the main text, with the only difference being that the pooled variance of equation (\ref{p_pool}) is used instead of the unpooled variance. Therefore, we omit the detailed proof of Theorem \ref{modified conditional power goal theorem} here.

\subsection{Power Goal of a \texorpdfstring{$P$}{p} Degree-of-Freedom Test for Binary Outcomes}\label{wald binary outcomes section}

As discussed in Section 2.3 of the main text, the $P$-df test refers to the Wald-type test that is used to assess the significance of the $P$-dimensional intervention package. This section recommends center-specific interventions $\boldsymbol{x}_j^{(2,n^{(1)})}$, where $j=1,...,J$, because the power of the $P$-df test is higher if one allows for center-specific recommended interventions.

Let $\hat{\boldsymbol{\beta}}^{(2)} = \bigl( \hat{\beta}_0^{(2)}, \bigl(\hat{\boldsymbol{\beta}}_1^{(2)}\bigr)^T \bigr)^T$ be the estimator of $\boldsymbol{\beta}$ from the final analysis based on data from both stages combined, and let
$\boldsymbol{\Sigma}$ be the asymptotic variance-covariance matrix of $\hat{\boldsymbol{\beta}}^{(2)}$. 
Let $\hat{\boldsymbol{\Sigma}}^{(2)}$ be the estimator for $\boldsymbol{\Sigma}$ based on data from both stages combined, and let $\hat{{\boldsymbol{\Sigma}}}^{(2)}_{\boldsymbol{\beta}_1}$ be the sub-matrix of $\hat{\boldsymbol{\Sigma}}^{(2)}$ relevant only to $\boldsymbol{\beta}_1$.  
The Wald-type test statistic is
\begin{equation}\label{w assum section}
{W} = n \left(\hat{\boldsymbol{\beta}}_1^{(2)}-\mathbf{0}\right)^T\left(\hat{\boldsymbol{\Sigma}}_{\boldsymbol{\beta}_1}^{(2)}\right)^{-1}\left(\hat{\boldsymbol{\beta}}_1^{(2)}-\mathbf{0}\right)
=
n \left(\hat{\boldsymbol{\beta}}_1^{(2)}\right)^T\left(\hat{\boldsymbol{\Sigma}}_{\boldsymbol{\beta}_1}^{(2)}\right)^{-1} \hat{\boldsymbol{\beta}}_1^{(2)}.
\end{equation}
By Theorem \ref{asymptotic properties of beta theorem}, 
under the null hypothesis that the intervention package has no effect on the outcome of interest, regardless of variations in implementation, $\hat{\boldsymbol{\beta}}^{(2)}$ is asymptotically normal under the conditions presented in Theorem \ref{asymptotic properties of beta theorem}. Then, under the null, $W$ asymptotically follows a $\chi^2$ distribution with $P$ degrees of freedom, and the null hypothesis is rejected at significance level $\alpha$ when $W > \chi_{\alpha, P}^2$, where $\chi_{\alpha, P}^2$ is the critical value corresponding to the upper $\alpha$ tail probability of a (central) $\chi^2$ distribution with $P$ degrees of freedom \citep{nevo2021analysis}. 

The asymptotic properties of the estimator $\hat{\boldsymbol{\beta}}^{(2)}$ proven in \citet{nevo2021analysis}, \citet{bing2023learnasyougo}, and Theorem \ref{asymptotic properties of beta theorem} imply that, under the alternative hypothesis for which the power is calculated,
$W$ approximately follows a non-central $\chi^2$ distribution with non-centrality parameter approximated using the stage 1 data by
$\lambda(\hat{\boldsymbol{\beta}}^{(2)}) = 
n \bigl(\hat{\boldsymbol{\beta}}_1^{(2)}\bigr)^T\bigl(\hat{\boldsymbol{\Sigma}}_{\boldsymbol{\beta}_1}^{(2)}\bigr)^{-1} \hat{\boldsymbol{\beta}}_1^{(2)}
$.

For the unconditional power approach, at the end of stage 1, the condition ${P}\left({W} > \chi_{\alpha, P}^2\right) \geq \Pi$ needs to be satisfied under the alternative hypothesis, unconditionally on the stage 1 data (although using the stage 1 data to calculate the recommended interventions $\hat{\boldsymbol{x}}_j^{(2,n^{(1)})}$, and to estimate the true underlying $\boldsymbol{\beta}^*$).

\begin{thm}{P-df test, binary outcomes, unconditional power goal.}\label{binary pdf thm} 
\textcolor{white}{xxx}\\  
Let $\chi^2_{\alpha, P}$ be the upper $\alpha$ quantile of the central $\chi^2$ distribution with $P$ degrees of freedom. For $\alpha=0.05$ and $P=1$, $\chi^2_{\alpha, P}=3.84$; for $\alpha=0.05$ and $P=2$, $\chi^2_{\alpha, P}=5.99$. 
Let ${\lambda}_{P,min}$ be the minimum value of the non-centrality parameter for the non-central $\chi^2$ distribution with $P$ degrees of freedom, so that for a variable $T$ from a non-central $\chi^2$ distribution with non-centrality parameter ${\lambda}_{P,min}$, the probability of $T$ exceeding $\chi^2_{\alpha, P}$ equals $\Pi$. 

Let
\begin{equation}\label{estimated I}
\begin{aligned}
&\hat{\boldsymbol{\Sigma}}(\hat{\boldsymbol{x}}_1^{(2,n^{(1)})},...,\hat{\boldsymbol{x}}_J^{(2,n^{(1)})};\hat{\boldsymbol{\beta}}^{(1)}) 
 = \\
&\left( \frac{1}{n} \sum_{\substack{j=1\\\text{\textcolor{white}{jjj}}}}^{J_1} {n_j^{(1)}} \left(\begin{array}{c}
1 \\
\boldsymbol{a}_j^{(1)} 
\end{array}\right)\left(\begin{array}{c}
1 \\
\boldsymbol{a}_j^{(1)}
\end{array}\right)^T\left[1-p_{\boldsymbol{a}_j^{(1)}}\left(\hat{\boldsymbol{\beta}}^{(1)} 
\right)\right] p_{\boldsymbol{a}_j^{(1)}}\left(\hat{\boldsymbol{\beta}}^{(1)} 
\right) \right. \\
&  +\frac{1}{n}\sum_{\substack{j=1\\\text{j intervention}}}^{J_2} {n_j^{(2)}} \left(\begin{array}{c}
1 \\
{\hat{\boldsymbol{x}}}^{(2,n^{(1)})}_j 
\end{array}\right)\left(\begin{array}{c}
1 \\
\hat{\boldsymbol{x}}^{(2,n^{(1)} )}_j 
\end{array}\right)^T \left[1-p_{\hat{\boldsymbol{x}}^{(2,n^{(1)})}_j}\left(\hat{\boldsymbol{\beta}}^{(1)} 
\right)\right] p_{\hat{\boldsymbol{x}}^{(2,n^{(1)})}_j}\left(\hat{\boldsymbol{\beta}}^{(1)} 
\right)  \\
& + \frac{1}{n}\left. \sum_{\substack{j=1\\\text{j control}}}^{J_2} {n_j^{(2)}} \left(\begin{array}{c}
1 \\
{\boldsymbol{0}} 
\end{array}\right)\left(\begin{array}{c}
1 \\
{\boldsymbol{0}} 
\end{array}\right)^T \left[1-p_{{\boldsymbol{0}}}\left(\hat{\boldsymbol{\beta}}^{(1)} 
\right)\right] p_{{\boldsymbol{0}}}\left(\hat{\boldsymbol{\beta}}^{(1)} 
\right) \right)^{-1},
\end{aligned}
\end{equation}
and let $\left(\hat{\boldsymbol{\Sigma}}(\hat{\boldsymbol{x}}_1^{(2,n^{(1)})},...,\hat{\boldsymbol{x}}_J^{(2,n^{(1)})};\hat{\boldsymbol{\beta}}^{(1)})\right)^{-1}_{\boldsymbol{\beta}_1}$ be the sub-matrix of $\left(\hat{\boldsymbol{\Sigma}}(\hat{\boldsymbol{x}}_1^{(2,n^{(1)})},...,\hat{\boldsymbol{x}}_J^{(2,n^{(1)})};\hat{\boldsymbol{\beta}}^{(1)})\right)^{-1} $ relevant only to $\boldsymbol{\beta}_1$. 
Let 
\begin{equation}\label{ncp new}
{\lambda}\bigl(\hat{\boldsymbol{x}}_1^{(2,n^{(1)})},...,\hat{\boldsymbol{x}}_J^{(2,n^{(1)})}; \hat{\boldsymbol{\beta}}^{(1)}\bigl) = n \; \left(\hat{\boldsymbol{\beta}}_1^{(1)}\right)^T \left(\hat{\boldsymbol{\Sigma}}(\hat{\boldsymbol{x}}_1^{(2,n^{(1)})},...,\hat{\boldsymbol{x}}_J^{(2,n^{(1)})};\hat{\boldsymbol{\beta}}^{(1)})\right)^{-1} _{\boldsymbol{\beta}_1}\hat{\boldsymbol{\beta}}_1^{(1)}.    
\end{equation}
Under Assumptions 1, 3, 4, 5, 7 and 8 from the main text, the stage 2 recommended interventions $\hat{\boldsymbol{x}}_1^{(2,n^{(1)})},...,\hat{\boldsymbol{x}}_J^{(2,n^{(1)})}$, subject to both an outcome goal and an unconditional power goal for the $P$-df test, are the solution to the following optimization problem: 
    \begin{equation}\label{Pdf test theorem binary}
    \begin{aligned}
    &\text{Min}_{\boldsymbol{x}_{1},...,\boldsymbol{x}_{J}} C\left(\boldsymbol{x}_{1},...,\boldsymbol{x}_{J}\right)
    \; \text {subject to} \; 
    \frac{1}{J}\sum_{j=1}^{J}
    \operatorname{expit}(\hat{\beta}_{0}^{(1)}+(\hat{\boldsymbol{\beta}}^{(1)}_1)^T \hat{\boldsymbol{x}}_j^{(2,n^{(1)})}) \geq \tilde{p}, \\
    & \;\;\text{and}\;\;
     {\lambda}\bigl(\hat{\boldsymbol{x}}_1^{(2,n^{(1)})},...,\hat{\boldsymbol{x}}_J^{(2,n^{(1)})}; \hat{\boldsymbol{\beta}}^{(1)}\bigl) \geq {\lambda}_{P,min}.
    \end{aligned}
    \end{equation}
\end{thm}
\begin{proof}
\textcolor{white}{xxx}\\
At the end of stage 1, not all data are available to calculate the test statistic $W$ (equation (\ref{w assum section})): $\hat{\boldsymbol{\beta}}_1^{(2)}$ and $\hat{\boldsymbol{\Sigma}}_{\boldsymbol{\beta}_1}^{(2)}$ are unknown. We estimate their means based on the stage 1 data and under the stage 2 recommended interventions $\hat{\boldsymbol{x}}_1^{(2,n^{(1)})},...,\hat{\boldsymbol{x}}_J^{(2,n^{(1)})}$. 

For $\hat{\boldsymbol{\beta}}_1^{(2)}$, we use the estimate $\hat{\boldsymbol{\beta}}_1^{(1)}$ from the fitted model based on the stage 1 data.

For $\hat{\boldsymbol{\Sigma}}_{\boldsymbol{\beta}_1}^{(2)}$,
\citet{nevo2021analysis} showed that it can be consistently estimated by the inverse of the observed Fisher information (see equation (A.13) of their Appendix):
\begin{equation}\label{new fisher}
\begin{aligned}
&  \left(\frac{1}{n} \sum_{j=1}^{J_1} {n_j^{(1)}} \left(\begin{array}{c}
1 \\
\boldsymbol{a}_j^{(1)} 
\end{array}\right)\left(\begin{array}{c}
1 \\
\boldsymbol{a}_j^{(1)} 
\end{array}\right)^T\left[1-p_{\boldsymbol{a}_j^{(1)}}\left(\hat{\boldsymbol{\beta}}^{(2)} %
\right)\right] p_{\boldsymbol{a}_j^{(1)}}\left(\hat{\boldsymbol{\beta}}^{(2)} %
\right) \right. \\
& \left.\qquad +\frac{1}{n}\sum_{j=1}^{J_2} {n_j^{(2)}} \left(\begin{array}{c}
1 \\
\boldsymbol{A}_j^{(2, n^{(1)})} 
\end{array}\right)\left(\begin{array}{c}
1 \\
\boldsymbol{A}_j^{(2, n^{(1)})} 
\end{array}\right)^T\left[1-p_{\boldsymbol{A}_j^{(2, n^{(1)})}}\left(\hat{\boldsymbol{\beta}}^{(2)}
\right)\right] p_{\boldsymbol{A}_j^{(2, n^{(1)})}}\left(\hat{\boldsymbol{\beta}}^{(2)} 
\right) \right)^{-1}.
\end{aligned}
\end{equation}
At the end of stage 1, $\hat{\boldsymbol{\beta}}^{(2)}$ and $\boldsymbol{A}_j^{(2, n^{(1)})}$ from equation (\ref{new fisher}) are unknown. To address this, we substitute $\hat{\boldsymbol{\beta}}^{(2)}$ and the $\boldsymbol{A}_j^{(2, n^{(1)})}$ by $\hat{\boldsymbol{\beta}}^{(1)}$ and the $\hat{\boldsymbol{x}}_j^{(2,n^{(1)})}$, to obtain $\hat{\boldsymbol{\Sigma}}(\hat{\boldsymbol{x}}_1^{(2,n^{(1)})},...,\hat{\boldsymbol{x}}_J^{(2,n^{(1)})};\hat{\boldsymbol{\beta}}^{(1)}) $ of equation \ref{estimated I}. 

The asymptotic properties of the estimator $\hat{\boldsymbol{\beta}}^{(2)}$ proven in \citet{nevo2021analysis} and Section 4 of the main text imply that, under the alternative hypothesis, $W$ approximately follows a non-central $\chi^2$ distribution with non-centrality parameter $\lambda(\hat{\boldsymbol{\beta}}^{(2)})$. 
Based on the stage 1 data, the non-centrality parameter $\lambda(\hat{\boldsymbol{\beta}}^{(2)})$ can be approximated under the $\hat{\boldsymbol{x}}_j^{(2,n^{(1)})}$ by ${\lambda}\bigl(\hat{\boldsymbol{x}}_1^{(2,n^{(1)})},...,\hat{\boldsymbol{x}}_J^{(2,n^{(1)})}; \hat{\boldsymbol{\beta}}^{(1)}\bigl)$ from equation (\ref{ncp new}).

It follows that to satisfy the unconditional power goal, the condition ${P}\left({W} > \chi_{\alpha, P}^2\right)> \Pi$ under the alternative hypothesis is approximately satisfied by aiming for ${\lambda}\bigl(\hat{\boldsymbol{x}}_1^{(2,n^{(1)})},...,\hat{\boldsymbol{x}}_J^{(2,n^{(1)})}; \hat{\boldsymbol{\beta}}^{(1)}\bigl) \geq {\lambda}_{P,min}$. Combining with equation (4) from the main text implies Theorem \ref{binary pdf thm}.
\end{proof}

\subsection{Power Goal of the Two-Sample t-Test for the Difference between Two Means with Unpooled Variance}\label{t test unpooled section}
The test statistic of the two-sample t-test for the difference between two means with unpooled variance in the final analysis based on data from both stages combined is
\begin{equation}\label{t unpool test statistic}
    T_{unpool}=\frac{\hat{\mu}_1-\hat{\mu}_0} {\sqrt{
    \hat{\sigma}_1^2/N_1 + \hat{\sigma}_0^2/N_0
    }},
\end{equation}
where
\begin{equation*}
\begin{aligned}
    &\hat{\mu}_1 = \frac{S_1^{(1)}+S_1^{(2)}}{N_1},
    \;\; \hat{\mu}_0 = \frac{S_0^{(1)}+S_0^{(2)}}{N_0}, \\
    &\hat{\sigma}_1^2 = \frac{1}{N_1-1} \sum_{\substack{i=1\\\text{i intervention}}}^{N_1} \left( Y_i - \hat{\mu}_1  \right)^2, 
    \;\; \hat{\sigma}_0^2 = \frac{1}{N_0-1} \sum_{\substack{i=1\\\text{i control}}}^{N_0} \left( Y_i - \hat{\mu}_0  \right)^2.
\end{aligned}
\end{equation*}

Let $\hat{\sigma}_0^{2{(1)}}$ and $\hat{\sigma}_1^{2{(1)}}$ be the stage 1-based estimated variances of the outcome in the control and intervention groups, respectively.
Let $\hat{\mu}_0^{(1)}$ and $\hat{\mu}_1^{(1)}$ be the stage 1-based estimated means in the control and intervention groups, respectively.

The unconditional power and conditional power approaches will be discussed separately. First consider the unconditional power approach. 
\begin{thm}{Two-sample t-test with unpooled variance, difference between two means, unconditional power goal.}\label{unconditional power theorem t test unpooled}
    \textcolor{white}{xxx}\\
    Let
    \begin{equation}\label{sigma tilde}
    \begin{aligned}
        \tilde{\sigma}^2_1\bigl(\hat{\boldsymbol{x}}^{(2,n^{(1)})};\hat{\boldsymbol{\beta}}^{(1)} \bigl) = \frac{1}{N_1-1} \left( (N_1-2)\hat{\sigma}_1^{2(1)} + \frac{n_1^{(1)}n_1^{(2)}}{N_1} \left( g^{-1}\left(\hat{\beta}_0^{(1)}+\hat{\boldsymbol{\beta}}_1^{(1)\;T} {\hat{\boldsymbol{x}}^{(2,n^{(1)})}}\right)  - \hat\mu_1^{(1)}\right)^2 \right),
    \end{aligned}
    \end{equation}
    and
    \begin{equation}
    {\lambda}\bigl(\hat{\boldsymbol{x}}^{(2,n^{(1)})}; \hat{\boldsymbol{\beta}}^{(1)} \bigl) =
    \frac{\left(\left.\bigl({S_1^{(1)}+ n_1^{(2)} \; g^{-1}\bigl(\hat{\beta}_0^{(1)}+\hat{\boldsymbol{\beta}}_1^{(1)\;T} {\hat{\boldsymbol{x}}^{(2,n^{(1)})}}\bigl) }\bigl) \right/{N_1} 
        - 
        \left.\bigl(S_0^{(1)}+ n_0^{(2)} \; g^{-1} \bigl(\hat{\beta}_0^{(1)}\bigl) \bigl)\right/{N_0} \right)^2}
        { \tilde{\sigma}^2_1\bigl(\hat{\boldsymbol{x}}^{(2,n^{(1)})};\hat{\boldsymbol{\beta}}^{(1)} \bigl)/N_1
        + \hat{\sigma}_0^{2(1)} / N_0
        } .
    \label{ncp t unpooled}
    \end{equation}
    Let ${\lambda}_{min}$ be the minimum value of the non-centrality parameter for the non-central $\chi^2$ distribution with 1 degree of freedom, so that for a variable $T$ from a non-central $\chi^2$ distribution with non-centrality parameter ${\lambda}_{min}$, the probability of $T$ exceeding $\chi^2_{\alpha, 1}$ equals $\Pi$. 
    
    Under Assumptions 2 --- 8 from the main text, the stage 2 recommended intervention ${\boldsymbol{x}}^{(2,n^{(1)})}$, subject to both an outcome goal and an unconditional power goal, solves the following optimization problem: 
    \begin{equation*}
        \text{Min}_{\boldsymbol{x}} C\left(\boldsymbol{x}\right)
        \; \text {subject to} \; 
        g^{-1}(\hat{\beta}_{0}^{(1)}+(\hat{\boldsymbol{\beta}}^{(1)}_1)^T {\hat{\boldsymbol{x}}^{(2,n^{(1)})}} )
        \geq \tilde{\theta}, \;\text{and} \; {\lambda}\bigl(\hat{\boldsymbol{x}}^{(2,n^{(1)})}; \hat{\boldsymbol{\beta}}^{(1)} \bigl)  \geq \lambda_{min}.
    \end{equation*}
\end{thm}

\begin{proof}
    \textcolor{white}{xxx}\\
    Under the alternative hypothesis, $T_{unpool}$ approximately follows a normal distribution with a non-zero mean and variance 1. Therefore, $T^2_{unpool}$ approximately follows a non-central $\chi^2$ distribution with 1 degree of freedom. The unconditional power approach estimates the non-centrality parameter of this non-central $\chi^2$ distribution based on the stage 1 data.
    
    At the end of stage 1, not all data are available to calculate the test statistic $T_{unpool}$ of equation (\ref{t unpool test statistic}). 
    The stage 2-related terms in $T_{unpool}$, $S_1^{(2)}$ and $S_0^{(2)}$, are estimated under $\hat{\boldsymbol{x}}^{(2, n^{(1)})}$ and $\hat{\boldsymbol{\beta}}^{(1)}$. 
    Specifically, $\left.\left({S_1^{(1)}+ n_1^{(2)} \; g^{-1}\left(\hat{\beta}_0^{(1)}+\hat{\boldsymbol{\beta}}_1^{(1)\;T} {\hat{\boldsymbol{x}}^{(2,n^{(1)})}}\right) }\right) \right/{N_1} $ and $\left.\left(S_0^{(1)}+ n_0^{(2)} \; g^{-1} \left(\hat{\beta}_0^{(1)}\right) \right)\right/{N_0}$ are used to estimate $\hat{\mu}_1$ and $\hat{\mu}_0$.
    
    Estimating the variance terms, $\hat{\sigma}_1^2$ and $\hat{\sigma}_0^2$, in $T_{unpool}$ can be done based on the stage 1 data as follows.
    Based on the stage 1 data, the best guess of $\hat{\sigma}_0^2$ at the end of the study is $\hat{\sigma}_0^{2{(1)}}$. 
    For $\hat{\sigma}_1^2$, notice that 
    \begin{equation}\label{plus minus original}
    \begin{aligned}
        & \hat{\sigma}_1^2  = 
        \frac{1}{N_1-1} \left( \sum_{\substack{i=1\\\text{i intervention}}}^{n_1^{(1)}} \left( Y_i^{(1)} - \hat{\mu}_1  \right)^2 + \sum_{\substack{i=1\\\text{i intervention}}}^{n_1^{(2)}} \left( Y_i^{(2, n^{(1)})} -  \hat{\mu}_1  \right)^2 \right)
        \\
        &= \frac{1}{N_1-1} \left( \sum_{\substack{i=1\\\text{i intervention}}}^{n_1^{(1)}} \left( Y_i^{(1)} - \hat{\mu}_1^{(1)} + \hat{\mu}_1^{(1)} - \hat{\mu}_1  \right)^2 + \right. \\
        &\left. \hspace{1cm} \sum_{\substack{i=1\\\text{i intervention}}}^{n_1^{(2)}} \left( Y_i^{(2, n^{(1)})} - g^{-1}\left(\hat{\beta}_0^{(1)}+\hat{\boldsymbol{\beta}}_1^{(1)\;T} {\hat{\boldsymbol{x}}^{(2,n^{(1)})}}\right) + g^{-1}\left(\hat{\beta}_0^{(1)}+\hat{\boldsymbol{\beta}}_1^{(1)\;T} {\hat{\boldsymbol{x}}^{(2,n^{(1)})}}\right)  -  \hat{\mu}_1  \right)^2 \right).
    \end{aligned}
    \end{equation}
    Consider the two summations from equation (\ref{plus minus original}) separately. First,
    \begin{equation}\label{pm expanded}
    \begin{aligned}
        &\sum_{\substack{i=1\\\text{i intervention}}}^{n_1^{(1)}} \left( Y_i^{(1)} - \hat{\mu}_1^{(1)} + \hat{\mu}_1^{(1)} - \hat{\mu}_1  \right)^2 \\ 
        &= 
        \sum_{\substack{i=1\\\text{i intervention}}}^{n_1^{(1)}} \left(Y_i^{(1)} - \hat{\mu}_1^{(1)} \right)^2
        + \sum_{\substack{i=1\\\text{i intervention}}}^{n_1^{(1)}} \left(\hat{\mu}_1^{(1)} - \hat{\mu}_1 \right)^2 + 2\left(\hat{\mu}_1^{(1)} - \hat{\mu}_1 \right)\sum_{\substack{i=1\\\text{i intervention}}}^{n_1^{(1)}} \left(Y_i^{(1)} - \hat{\mu}_1^{(1)} \right) \\
        &= \sum_{\substack{i=1\\\text{i intervention}}}^{n_1^{(1)}} \left(Y_i^{(1)} - \hat{\mu}_1^{(1)} \right)^2
        + n_1^{(1)} \left(\hat{\mu}_1^{(1)} - \hat{\mu}_1 \right)^2;
    \end{aligned}
    \end{equation}
    The cross-product term from equation (\ref{pm expanded}) equals zero because $$\hat{\mu}_1^{(1)} = \sum_{\substack{i=1\\\text{i intervention}}}^{n_1^{(1)}} Y_i^{(1)} / n_1^{(1)}.$$ Based on the stage 1 data, $\left.\left({n_1^{(1)} \hat{\mu}_1^{(1)}+ n_1^{(2)} \; g^{-1}\left(\hat{\beta}_0^{(1)}+\hat{\boldsymbol{\beta}}_1^{(1)\;T} {\hat{\boldsymbol{x}}^{(2,n^{(1)})}}\right) }\right) \right/{N_1} $ is the best estimate of $\hat{\mu}_1$ at the end of the study.
    Therefore, equation (\ref{pm expanded}) equals approximately 
    \begin{equation}\label{first sum}
    \begin{aligned}
        &\hat{\sigma}^{2(1)}_1 \left(n_1^{(1)} - 1 \right) + n_1^{(1)} \left( \hat{\mu}_1^{(1)} -  \frac{\left({n_1^{(1)} \hat{\mu}_1^{(1)}+ n_1^{(2)} \; g^{-1}\left(\hat{\beta}_0^{(1)}+\hat{\boldsymbol{\beta}}_1^{(1)\;T} {\hat{\boldsymbol{x}}^{(2,n^{(1)})}}\right) }\right)}{N_1} \right) ^2 \\
        &=
        \hat{\sigma}^{2(1)}_1 \left(n_1^{(1)} - 1 \right) + n_1^{(1)} \left( \frac{\hat{\mu}_1^{(1)}N_1}{N_1} - \frac{n_1^{(1)} \hat{\mu}_1^{(1)}}{N_1} - \frac{ n_1^{(2)} \; g^{-1}\left(\hat{\beta}_0^{(1)}+\hat{\boldsymbol{\beta}}_1^{(1)\;T} {\hat{\boldsymbol{x}}^{(2,n^{(1)})}}\right) }{N_1} \right) ^2
        \\
        &=
        \hat{\sigma}^{2(1)}_1 \left(n_1^{(1)} - 1 \right) + n_1^{(1)} \left( \frac{n_1^{(2)}\hat{\mu}_1^{(1)}}{N_1} - \frac{ n_1^{(2)} \; g^{-1}\left(\hat{\beta}_0^{(1)}+\hat{\boldsymbol{\beta}}_1^{(1)\;T} {\hat{\boldsymbol{x}}^{(2,n^{(1)})}}\right) }{N_1} \right) ^2
        \\
        &=
        \hat{\sigma}^{2(1)}_1 \left(n_1^{(1)} - 1 \right) + \frac{n_1^{(1)}(n_1^{(2)})^2}{N_1^2} \left( \hat{\mu}_1^{(1)} - g^{-1}\left(\hat{\beta}_0^{(1)}+\hat{\boldsymbol{\beta}}_1^{(1)\;T} {\hat{\boldsymbol{x}}^{(2,n^{(1)})}}\right) \right)^2.
    \end{aligned}
    \end{equation}

    Similarly, 
    \begin{equation}\label{second term}
    \begin{aligned}
        & \sum_{\substack{i=1\\\text{i intervention}}}^{n_1^{(2)}} \left( Y_i^{(2, n^{(1)})} - g^{-1}\left(\hat{\beta}_0^{(1)}+\hat{\boldsymbol{\beta}}_1^{(1)\;T} {\hat{\boldsymbol{x}}^{(2,n^{(1)})}}\right) + g^{-1}\left(\hat{\beta}_0^{(1)}+\hat{\boldsymbol{\beta}}_1^{(1)\;T} {\hat{\boldsymbol{x}}^{(2,n^{(1)})}}\right)  -  \hat{\mu}_1  \right)^2\\
        &= \sum_{\substack{i=1\\\text{i intervention}}}^{n_1^{(2)}}
        \left(Y_i^{(2, n^{(1)})} - g^{-1}\left(\hat{\beta}_0^{(1)}+\hat{\boldsymbol{\beta}}_1^{(1)\;T} {\hat{\boldsymbol{x}}^{(2,n^{(1)})}}\right)  \right)^2
        + \sum_{\substack{i=1\\\text{i intervention}}}^{n_1^{(2)}}
        \left( g^{-1}\left(\hat{\beta}_0^{(1)}+\hat{\boldsymbol{\beta}}_1^{(1)\;T} {\hat{\boldsymbol{x}}^{(2,n^{(1)})}}\right)  - \hat\mu_1\right)^2 \\
        & \;\;\;\; + 2\sum_{\substack{i=1\\\text{i intervention}}}^{n_1^{(2)}} 
        \left(Y_i^{(2, n^{(1)})} - g^{-1}\left(\hat{\beta}_0^{(1)}+\hat{\boldsymbol{\beta}}_1^{(1)\;T} {\hat{\boldsymbol{x}}^{(2,n^{(1)})}}\right)  \right) \left( g^{-1}\left(\hat{\beta}_0^{(1)}+\hat{\boldsymbol{\beta}}_1^{(1)\;T} {\hat{\boldsymbol{x}}^{(2,n^{(1)})}}\right)  - \hat\mu_1\right) \\
        & \approx
        \hat{\sigma}_1^{2(1)} \left(n_1^{(2)}-1 \right) + n_1^{(2)} \left( g^{-1}\left(\hat{\beta}_0^{(1)}+\hat{\boldsymbol{\beta}}_1^{(1)\;T} {\hat{\boldsymbol{x}}^{(2,n^{(1)})}}\right)  -  \frac{\left({n_1^{(1)} \hat{\mu}_1^{(1)}+ n_1^{(2)} \; g^{-1}\left(\hat{\beta}_0^{(1)}+\hat{\boldsymbol{\beta}}_1^{(1)\;T} {\hat{\boldsymbol{x}}^{(2,n^{(1)})}}\right) }\right)}{N_1} \right) ^2 \\
        & \;\;\;\; + 2\left( g^{-1}\left(\hat{\beta}_0^{(1)}+\hat{\boldsymbol{\beta}}_1^{(1)\;T} {\hat{\boldsymbol{x}}^{(2,n^{(1)})}}\right)  - \hat\mu_1\right) \sum_{\substack{i=1\\\text{i intervention}}}^{n_1^{(2)}} 
        \left(Y_i^{(2, n^{(1)})} - g^{-1}\left(\hat{\beta}_0^{(1)}+\hat{\boldsymbol{\beta}}_1^{(1)\;T} {\hat{\boldsymbol{x}}^{(2,n^{(1)})}}\right)  \right)
        \\
        & =
        \hat{\sigma}_1^{2(1)} \left(n_1^{(2)}-1 \right) + n_1^{(2)} \left( \frac{ N_1 g^{-1}\left(\hat{\beta}_0^{(1)}+\hat{\boldsymbol{\beta}}_1^{(1)\;T} {\hat{\boldsymbol{x}}^{(2,n^{(1)})}}\right)}{N_1}  - \frac{n_1^{(2)} \; g^{-1}\left(\hat{\beta}_0^{(1)}+\hat{\boldsymbol{\beta}}_1^{(1)\;T} {\hat{\boldsymbol{x}}^{(2,n^{(1)})}}\right) }{N_1} - \frac{n_1^{(1)} \hat{\mu}_1^{(1)}}{N_1} \right) ^2
        \\
        & =
        \hat{\sigma}_1^{2(1)} \left(n_1^{(2)}-1 \right) + n_1^{(2)} \left( \frac{ n_1^{(1)} g^{-1}\left(\hat{\beta}_0^{(1)}+\hat{\boldsymbol{\beta}}_1^{(1)\;T} {\hat{\boldsymbol{x}}^{(2,n^{(1)})}}\right) }{N_1} - \frac{n_1^{(1)} \hat{\mu}_1^{(1)}}{N_1} \right) ^2
        \\
        & = 
        \hat{\sigma}_1^{2(1)} \left(n_1^{(2)}-1 \right) + \frac{n_1^{(2)}(n_1^{(1)})^2}{N_1^2} \left( g^{-1}\left(\hat{\beta}_0^{(1)}+\hat{\boldsymbol{\beta}}_1^{(1)\;T} {\hat{\boldsymbol{x}}^{(2,n^{(1)})}}\right)  - \hat\mu_1^{(1)}\right)^2.
    \end{aligned}
    \end{equation}

    Combining equations (\ref{plus minus original}), (\ref{first sum}) and (\ref{second term}), it follows that based on the stage 1 data, equation (\ref{sigma tilde}) provides the best estimate of $\hat{\sigma}_1^2$ at the end of the study.

    Under the alternative hypothesis, $T_{unpool}^2$ approximately follows a non-central $\chi^2$ distribution with 1 degree of freedom. 
    The non-centrality parameter $\lambda$ can be estimated under $\hat{\boldsymbol{x}}^{(2, n^{(1)})}$ and $\hat{\boldsymbol{\beta}}^{(1)}$ by ${\lambda}\bigl(\hat{\boldsymbol{x}}^{(2,n^{(1)})}; \hat{\boldsymbol{\beta}}^{(1)} \bigl)$ from equation (\ref{ncp t unpooled}).
    To satisfy the unconditional power goal, the condition $P(T_{unpool}^2>\chi^2_{\alpha, 1})\geq \Pi$ under the alternative hypothesis is ensured by achieving ${\lambda}\bigl(\hat{\boldsymbol{x}}^{(2,n^{(1)})}; \hat{\boldsymbol{\beta}}^{(1)} \bigl) \geq {\lambda}_{min}$. Combining equation (\ref{ncp t unpooled}) with the constraint imposed by the outcome goal implies Theorem~\ref{unconditional power theorem t test unpooled}.
\end{proof}

Next, we discuss the conditional power approach.

\begin{thm}{Two-sample t-test with unpooled variance, difference between two means, conditional power goal.}\label{two-sample t test theorem unpooled}
    \textcolor{white}{xxx}\\
    Let
    \begin{equation}\label{big delta cts unpooled}
    \begin{aligned}
        \Delta^{(g)}\left(\hat{\boldsymbol{x}}^{(2,n^{(1)})}; \hat{\boldsymbol{\beta}}^{(1)}\right) = \frac{n_1^{(2)}}{N_1} \; g^{-1}\left(\hat{\beta}_0^{(1)}+\hat{\boldsymbol{\beta}}_1^{(1)\;T} {\hat{\boldsymbol{x}}^{(2,n^{(1)})}}\right) - \frac{n_0^{(2)}}{N_0}  \; g^{-1} \left(\hat{\beta}_0^{(1)}\right),
    \end{aligned}
    \end{equation}
    \begin{equation}\label{big sigma cts unpooled}
    \begin{aligned}
        \text{and}\;\; 
        \hat{\sigma}^{(g)}  = 
        \sqrt{
        \frac{n_1^{(2)}\hat{\sigma}_1^{2{(1)}}}{N_1^2} + 
        \frac{n_0^{(2)}\hat{\sigma}_0^{2{(1)}}}{N_0^2} }.
    \end{aligned}
    \end{equation}
    Under Assumptions 2 --- 8 from the main text, the stage 2 recommended intervention ${\boldsymbol{x}}^{(2,n^{(1)})}$, subject to both an outcome goal and the conditional power goal, solves the following optimization problem:
    \begin{equation*}
    \begin{aligned}
        &\text{Min}_{\boldsymbol{x}} C\left(\boldsymbol{x}\right)
        \; \text {subject to} \; \; 
        g^{-1}\left(\hat{\beta}_0^{(1)}+\left(\hat{\boldsymbol{\beta}}_1^{(1)}\right)^T \hat{\boldsymbol{x}}^{(2,n^{(1)})}\right) \geq \tilde{\theta},
        \;\text{and}\; 
    \end{aligned}
    \end{equation*}
    \begin{equation}
    \begin{aligned}
         &z_{\alpha / 2} \sqrt{\frac{\tilde{\sigma}^2_1\bigl(\hat{\boldsymbol{x}}^{(2,n^{(1)})};\hat{\boldsymbol{\beta}}^{(1)} \bigl)}{N_1}+\frac{\hat{\sigma}_0^{2{(1)}}}{N_0}} -\frac{S_1^{(1)}}{N_1} +\frac{S_0^{(1)}}{N_0} 
         - 
         \Delta^{(g)}\left(\hat{\boldsymbol{x}}^{(2,n^{(1)})}; \hat{\boldsymbol{\beta}}^{(1)}\right) - z_{\Pi}\hat{\sigma}^{(g)}
         \leq 0,
    \end{aligned}
    \end{equation}
    where $\tilde{\sigma}^2_1\bigl(\hat{\boldsymbol{x}}^{(2,n^{(1)})};\hat{\boldsymbol{\beta}}^{(1)} \bigl)$ is defined in Theorem \ref{unconditional power theorem t test unpooled}.
\end{thm}
The superscript $^{(g)}$ in $\Delta^{(g)}\left(\hat{\boldsymbol{x}}^{(2,n^{(1)})}; \hat{\boldsymbol{\beta}}^{(1)}\right)$ and $\hat{\sigma}^{2,(g)}$ reminds us that these quantiles are in the context of LAGO trials with continuous outcomes, modeled with a GLM and general link function.

\begin{proof}
\textcolor{white}{xxx}\\
Consider the test statistic $T_{unpool}$ from equation (\ref{t unpool test statistic}). Asymptotically, this t-test approximates a z-test, and the null hypothesis $H_0: \mu_1=\mu_0$, is rejected at significance level $\alpha$ when $|T_{unpool}|>z_{\alpha / 2}$. 
Assuming the intervention aims to increase the mean of the outcome, a positive effect of the intervention on the outcome would be expected. Thus, under the alternative hypothesis for which the power is calculated, the probability of observing $T_{unpool}<-z_{\alpha / 2}$ is usually negligible. Recall that $\overline{\boldsymbol{a}}^{(1)}$ and $\overline{\boldsymbol{Y}}^{(1)}$ are the actual interventions and outcomes for all stage 1 centers, respectively.
To satisfy the conditional power goal, the condition $P\left(T_{unpool}>z_{\alpha / 2} \;\middle\vert\; \overline{\boldsymbol{a}}^{(1)}, \overline{\boldsymbol{Y}}^{(1)}\right) \geq \Pi$ needs to be satisfied under the alternative hypothesis.

Similar to the proof of Theorem 2 from the main text,
$P\left(T_{unpool}>z_{\alpha / 2}\;\middle\vert\; \overline{\boldsymbol{a}}^{(1)}, \overline{\boldsymbol{Y}}^{(1)}\right) \geq \Pi$ is equivalent to
\begin{equation}\label{two-sample t test conditional condition unpooled}
    P\left(\frac{S_1^{(2)}}{N_1}-\frac{S_0^{(2)}}{N_0}>z_{\alpha / 2} \sqrt{\frac{\hat{\sigma}_1^2}{N_1}+\frac{\hat{\sigma}_0^2}{N_0}}+\frac{S_0^{(1)}}{N_0}-\frac{S_1^{(1)}}{N_1}\;\middle\vert\; \overline{\boldsymbol{a}}^{(1)}, \overline{\boldsymbol{Y}}^{(1)}\right) \geq \Pi.
\end{equation}
By the Central Limit Theorem, both ${{S}_1^{(2)}}/{N_1}$ and ${{S}_0^{(2)}}/{N_0}$ from equation (\ref{two-sample t test conditional condition unpooled}) are asymptotically normally distributed. 
The conditional power approach estimates the mean and variance of the distribution of $ \bigl({{S}_1^{(2)}}/{N_1} - {{S}_0^{(2)}}/{N_0} \bigl) $ based on $\hat{\boldsymbol{x}}^{(2,n^{(1)})}$ and $\hat{\boldsymbol{\beta}}^{(1)}$, 
denoted by 
$\Delta^{(g)}\left(\hat{\boldsymbol{x}}^{(2,n^{(1)})}; \hat{\boldsymbol{\beta}}^{(1)}\right)$ and 
$\hat{\sigma}^{2(g)}$, respectively. 

The estimated asymptotic distributions of ${{S}_1^{(2)}}/{N_1}$ and ${{S}_0^{(2)}}/{N_0}$ based on $\hat{\boldsymbol{x}}^{(2,n^{(1)})}$ and $\hat{\boldsymbol{\beta}}^{(1)}$ are
\begin{equation}
\begin{aligned}
    &\frac{{S}_1^{(2)}}{N_1} \approx \mathcal{N} \left(
    \frac{n_1^{(2)}}{N_1} \; g^{-1}\left(\hat{\beta}_0^{(1)}+\hat{\boldsymbol{\beta}}_1^{(1)\;T} {\hat{\boldsymbol{x}}^{(2,n^{(1)})}}
    \right),
    \frac{n_1^{(2)}\hat{\sigma}_1^{2{(1)}}}{N_1^2}
    \right),\\
    &\text{and}\;\; 
    \frac{{S}_0^{(2)}}{N_0} \approx \mathcal{N}\left( 
    \frac{n_0^{(2)}}{N_0}  \; g^{-1} \left(\hat{\beta}_0^{(1)}\right),
    \frac{n_0^{(2)}\hat{\sigma}_0^{2{(1)}}}{N_0^2}
    \right).
\end{aligned}
\end{equation}
Thus, $\Delta^{(g)}\left(\hat{\boldsymbol{x}}^{(2,n^{(1)})}; \hat{\boldsymbol{\beta}}^{(1)}\right)$ and 
$\hat{\sigma}^{2,(g)}$ from equations (\ref{big delta cts unpooled}) and (\ref{big sigma cts unpooled}) are the mean and variance of the distribution of $(S_1^{(2)}/N_1 - S_0^{(2)}/N_0)$ estimated based on the stage 1 data.

Let 
$$
G_3 = 
\left(
\left(\frac{{S}_1^{(2)}}{N_1}-\frac{{S}_0^{(2)}}{N_0}\right) 
- 
\Delta^{(g)}\left(\hat{\boldsymbol{x}}^{(2,n^{(1)})}; \hat{\boldsymbol{\beta}}^{(1)}\right)
\right)
\frac{1}{\hat{\sigma}^{(g)} }.
$$
Following the same argument as in the proof for Theorem \ref{unconditional power theorem t test unpooled}, based on the stage 1 data, the best guess of $\hat{\sigma}_0^2$ at the end of the study is $\hat{\sigma}_0^{2{(1)}}$, and the best guess of $\hat{\sigma}_1^2$ at the end of the study is $\tilde{\sigma}^2_1\bigl(\hat{\boldsymbol{x}}^{(2,n^{(1)})};\hat{\boldsymbol{\beta}}^{(1)} \bigl)$ (equation (\ref{sigma tilde})).
Equation (\ref{two-sample t test conditional condition unpooled}) can be reformulated as
\footnotesize{
\begin{equation}\label{long new g2 continuous unpooled}
P\left(G_3>
\left(z_{\alpha / 2} \sqrt{\frac{\tilde{\sigma}^2_1\bigl(\hat{\boldsymbol{x}}^{(2,n^{(1)})};\hat{\boldsymbol{\beta}}^{(1)} \bigl)}{N_1}+\frac{\hat{\sigma}_0^{2{(1)}}}{N_0}} + \frac{S_0^{(1)}}{N_0}-\frac{S_1^{(1)}}{N_1}
-\Delta^{(g)}\left(\hat{\boldsymbol{x}}^{(2,n^{(1)})}; \hat{\boldsymbol{\beta}}^{(1)}\right)
\right) \frac{1}{\hat{\sigma}^{2,(g)}}\;\middle\vert\; \overline{\boldsymbol{a}}^{(1)}, \overline{\boldsymbol{Y}}^{(1)}
\right) \geq \Pi.
\end{equation}}
\normalsize
\noindent Similar to Section 2.4 of the main text, 
by the Central Limit Theorem, $G_3$ is approximately normally distributed with mean 0 and variance 1. 
With $\Phi(\cdot)$, the cumulative distribution function of the standard normal distribution, equation (\ref{long new g2 continuous unpooled}) is equivalent to
\footnotesize
\begin{equation}\label{two-sample t-test cdf part}
\begin{aligned}
&1-\Phi\left(
\left(z_{\alpha / 2} \sqrt{\frac{\tilde{\sigma}^2_1\bigl(\hat{\boldsymbol{x}}^{(2,n^{(1)})};\hat{\boldsymbol{\beta}}^{(1)} \bigl)}{N_1}+\frac{\hat{\sigma}_0^{2{(1)}}}{N_0}} + \frac{S_0^{(1)}}{N_0}-\frac{S_1^{(1)}}{N_1}
-\Delta^{(g)}\left(\hat{\boldsymbol{x}}^{(2,n^{(1)})}; \hat{\boldsymbol{\beta}}^{(1)}\right)
\right) \frac{1}{\hat{\sigma}^{2,(g)}}\;\middle\vert\; \overline{\boldsymbol{a}}^{(1)}, \overline{\boldsymbol{Y}}^{(1)}
\right)\\
&\hspace{12cm}\geq 1-\Phi\left(z_{\Pi}\right),
\end{aligned}
\end{equation}
\normalsize
because $\Pi=\Phi\left(-z_{\Pi}\right)=1-\Phi\left(z_{\Pi}\right)$.
Combining equation (\ref{two-sample t-test cdf part}) with the constraint imposed by the outcome goal implies Theorem \ref{two-sample t test theorem unpooled}.
\end{proof}

\subsection{Power Goal of the Two-Sample t-test for the Difference between Two Means with Pooled Variance}\label{two-sample t test Appendix}

Let
\begin{equation*}
    \hat{s}_{pool} =
    \sqrt{\frac{(N_1-1)\hat{\sigma}_1^2 + (N_0-1)\hat{\sigma}_0^2}{N_1+N_0-2}}
    \label{s_pool}
\end{equation*}
The test statistic of the two-sample t-test for the difference between two means with pooled variance is
\begin{equation*}\label{t pool definition}
    T_{pool}=\frac{\hat{\mu}_1-\hat{\mu}_0} {\hat{s}_{pool}  \sqrt{
    \frac{1}{N_1}+\frac{1}{N_0}
    }}.
\end{equation*}

The unconditional power and conditional power approaches are discussed separately. First, consider the unconditional power approach.

Let
    \begin{equation}\label{s hat pool estimated}
        \hat{s}_{pool}(\hat{\boldsymbol{x}}^{\left(2, n^{(1)}\right)}, \hat{\boldsymbol{\beta}}^{(1)}) =     \sqrt{\frac{(N_1-1) \tilde{\sigma}^2_1\bigl(\hat{\boldsymbol{x}}^{(2,n^{(1)})};\hat{\boldsymbol{\beta}}^{(1)} \bigl) + (N_0-1)\hat{\sigma}_0^{2(1)}}{N_1+N_0-2}},
    \end{equation}
    where $\tilde{\sigma}^2_1\bigl(\hat{\boldsymbol{x}}^{(2,n^{(1)})};\hat{\boldsymbol{\beta}}^{(1)} \bigl) $ is defined in Theorem \ref{unconditional power theorem t test unpooled}.
\begin{thm}{Two-sample t-test with unpooled variance, difference between two means, unconditional power goal.}\label{pooled t modified unconditional power theorem}
    \textcolor{white}{xxx}\\
    Let $F_{\chi^2}(k; \nu, \lambda) = P(Z^2 \leq k)$ for $Z\sim \mathcal{N}(0,1)$, be the CDF of the non-central $\chi^2$ distribution with $\nu$ degrees of freedom and non-centrality parameter $\lambda$.
    Under Assumptions 2 -- 8 from the main text, the stage 2 recommended intervention $\hat{\boldsymbol{x}}^{(2,n^{(1)})}$, subject to both an outcome goal and an unconditional power goal, is obtained by solving the following optimization problem: 
    \begin{equation*}
        \text{Min}_{\boldsymbol{x}} C\left(\boldsymbol{x}\right)
        \; \text {subject to} \; 
        \operatorname{expit}\left(\hat{\beta}_{0}^{(1)}+(\hat{\boldsymbol{\beta}}^{(1)}_1)^T {\hat{\boldsymbol{x}}^{(2,n^{(1)})}} \right)
        \geq \tilde{p}, \;\text{and} 
    \end{equation*}
    \begin{equation*}\label{pooled t test power condition long}
    \begin{aligned}
        &1 - F_{\chi^2}\left( \chi^2_{\alpha, 1} \frac{ \hat{s}_{pool}^2(\hat{\boldsymbol{x}}^{\left(2, n^{(1)}\right)}, \hat{\boldsymbol{\beta}}^{(1)}) {(\frac{1}{N_1} + \frac{1}{N_0})} } { \tilde{\sigma}^2_1\bigl(\hat{\boldsymbol{x}}^{(2,n^{(1)})};\hat{\boldsymbol{\beta}}^{(1)} \bigl)/N_1
        + \hat{\sigma}_0^{2(1)} / N_0} ; 1, {\lambda}\bigl(\hat{\boldsymbol{x}}^{(2,n^{(1)})}; \hat{\boldsymbol{\beta}}^{(1)}\bigl) \right) 
        \geq \Pi,
    \end{aligned}
    \end{equation*}
    where $\hat{S}_1^{(2)}\left(\hat{\boldsymbol{x}}^{(2,n^{(1)})}, \hat{\boldsymbol{\beta}}^{(1)}\right)$ and ${\lambda}\bigl(\hat{\boldsymbol{x}}^{(2,n^{(1)})}; \hat{\boldsymbol{\beta}}^{(1)}\bigl)$ are defined in Theorem \ref{unconditional power theorem t test unpooled}.
\end{thm}
The proof of Theorem \ref{pooled t modified unconditional power theorem} follows directly from the proof of Theorem \ref{unconditional power theorem t test unpooled}, the same way Theorem \ref{modified unconditional power theorem} follows from Theorem 1 of the main text. Therefore, the proof of Theorem \ref{pooled t modified unconditional power theorem} is omitted.

Next, we discuss the conditional power approach.
\begin{thm}{Two-sample t-test with unpooled variance, difference between two means, conditional power goal.}\label{two-sample t test theorem pooled}
    \textcolor{white}{xxx}\\
    Under Assumptions 2 --- 8 from the main text, the stage 2 recommended intervention ${\boldsymbol{x}}^{(2,n^{(1)})}$, subject to both an outcome goal and the conditional power goal, solves the following optimization problem:
    \begin{equation*}
    \begin{aligned}
        &\text{Min}_{\boldsymbol{x}} C\left(\boldsymbol{x}\right)
        \; \text {subject to} \; \; 
        g^{-1}\left(\hat{\beta}_0^{(1)}+\left(\hat{\boldsymbol{\beta}}_1^{(1)}\right)^T \hat{\boldsymbol{x}}^{(2,n^{(1)})}\right) \geq \tilde{\theta},
        \;\text{and}\; 
    \end{aligned}
    \end{equation*}
    \begin{equation*}
    \begin{aligned}
         &z_{\alpha / 2}  \;\hat{s}_{pool}(\hat{\boldsymbol{x}}^{\left(2, n^{(1)}\right)}, \hat{\boldsymbol{\beta}}^{(1)}) \sqrt{\frac{1}{N_1} + \frac{1}{N_0}} -\frac{S_1^{(1)}}{N_1} +\frac{S_0^{(1)}}{N_0} 
         - 
         \Delta^{(g)}\left(\hat{\boldsymbol{x}}^{(2,n^{(1)})}; \hat{\boldsymbol{\beta}}^{(1)}\right) - z_{\Pi}\hat{\sigma}^{(g)}
         \leq 0,
    \end{aligned}
    \end{equation*}
    where $\Delta^{(g)}\left(\hat{\boldsymbol{x}}^{(2,n^{(1)})}; \hat{\boldsymbol{\beta}}^{(1)}\right)$ and $\hat{\sigma}^{(g)}$ are defined in Theorem \ref{two-sample t test theorem unpooled}.
\end{thm}

The proof of Theorem \ref{two-sample t test theorem pooled} follows the same steps as the proof of Theorem \ref{two-sample t test theorem unpooled}, with the only difference being that the pooled variance of equation (\ref{s hat pool estimated}) is used instead of the unpooled variance. Therefore, we omit the detailed proof of Theorem \ref{two-sample t test theorem pooled}.

\subsection{Power Goal of a \texorpdfstring{$P$}{p} Degree-of-Freedom Test for Continuous Outcomes}\label{wald cts outcomes section}
Section \ref{wald cts outcomes section} considers the power goal of the Wald-type $P$-df test for continuous outcomes, using the unconditional power approach. The theorem presented below modifies Theorem \ref{binary pdf thm} from Section \ref{wald binary outcomes section}, which discusses the power goal of the Wald-type $P$-df test for binary outcomes.

\begin{thm}{P-df test, continuous outcomes, unconditional power goal.}\label{cts pdf thm} 
\textcolor{white}{xxx}\\
Let $\chi^2_{\alpha, P}$ be the upper $\alpha$ quantile of the central $\chi^2$ distribution with $P$ degrees of freedom. For $\alpha=0.05$ and $P=1$, $\chi^2_{\alpha, P}=3.84$; for $\alpha=0.05$ and $P=2$, $\chi^2_{\alpha, P}=5.99$. 
Let ${\lambda}_{P,min}$ be the minimum value of the non-centrality parameter for the non-central $\chi^2$ distribution with $P$ degrees of freedom, so that for a variable $T$ from a non-central $\chi^2$ distribution with non-centrality parameter ${\lambda}_{P,min}$, the probability of $T$ exceeding $\chi^2_{\alpha, P}$ equals $\Pi$. 

Let 
\begin{equation*}\label{cts var estimate stage 1}
\begin{aligned}
&\hat{\boldsymbol{\Sigma}}(\boldsymbol{x}_1^{(2,n^{(1)})},...,\boldsymbol{x}_J^{(2,n^{(1)})}; \hat{\boldsymbol{\beta}}^{(1)})
    =\\
    &\left(\sum_{j=1}^{J^{(1)}} \frac{n_j^{(1)}}{n}\left(\left.\frac{\partial}{\partial \boldsymbol{\beta}}\right|_{\hat{\boldsymbol{\beta}}^{(1)}} g^{-1}\left(\boldsymbol{a}_j^{(1)} ; \boldsymbol{\beta}\right)\right)^{\otimes 2}+\sum_{j=1}^{J^{(2)}} \frac{n_j^{(2)}}{n}\left(\left.\frac{\partial}{\partial \boldsymbol{\beta}}\right|_{\hat{\boldsymbol{\beta}}^{(1)}} g^{-1}\left(\boldsymbol{x}_j^{(2,n^{(1)})}; \boldsymbol{\beta}\right)\right)^{\otimes 2}\right)^{-1}\\
    &\left(\frac{1}{n} \sum_{j=1}^{J^{(1)}}\left\{\left(\left.\frac{\partial}{\partial \boldsymbol{\beta}}\right|_{\hat{\boldsymbol{\beta}}^{(1)}} g^{-1}\left(\boldsymbol{a}_j^{(1)} ; \boldsymbol{\beta}\right)\right)^{\otimes 2} \sum_{i=1}^{n_j^{(1)}}\left(Y_{i j}^{(1)}-g^{-1}\left(\boldsymbol{a}_j^{(1)} ; \hat{\boldsymbol{\beta}}^{(1)}\right)\right)^2\right\} \right.\\
    & \hspace{2cm} +\frac{1}{n} \sum_{\substack{j=1\\\text{j intervention}}}^{J^{(2)}} \left\{\left(\left.\frac{\partial}{\partial \boldsymbol{\beta}}\right|_{\hat{\boldsymbol{\beta}}^{(1)}} g^{-1}\left(\boldsymbol{x}_j^{(2,n^{(1)})}; \boldsymbol{\beta}\right)\right)^{\otimes 2} 
    \hat{\sigma}^{2(1)}_1
    n_j^{(2)}
    \right\} \\
    & \hspace{2cm} \left. +\frac{1}{n} \sum_{\substack{j=1\\\text{j control}}}^{J^{(2)}} \left\{\left(\left.\frac{\partial}{\partial \boldsymbol{\beta}}\right|_{\hat{\boldsymbol{\beta}}^{(1)}} g^{-1}\left(\boldsymbol{0}; \boldsymbol{\beta}\right)\right)^{\otimes 2} 
    \hat{\sigma}^{2(1)}_0 n_j^{(2)}
    \right\} \right)\\
    &\left(\sum_{j=1}^{J^{(1)}} \frac{n_j^{(1)}}{n}\left(\left.\frac{\partial}{\partial \boldsymbol{\beta}}\right|_{\hat{\boldsymbol{\beta}}^{(1)}} g^{-1}\left(\boldsymbol{a}_j^{(1)}; \boldsymbol{\beta}\right)\right)^{\otimes 2}+\sum_{j=1}^{J^{(2)}} \frac{n_j^{(2)}}{n}\left(\left.\frac{\partial}{\partial \boldsymbol{\beta}}\right|_{\hat{\boldsymbol{\beta}}^{(1)}} g^{-1}\left(\boldsymbol{x}_j^{(2,n^{(1)})}; \boldsymbol{\beta}\right)\right)^{\otimes 2}\right)^{-1}
    ,
\end{aligned}
\end{equation*} 
where $Z^{\otimes2} = ZZ^T$ is the Kronecker product, and
$\hat{\sigma}_0^{2{(1)}}$ and $\hat{\sigma}_1^{2{(1)}}$ are the stage 1-based estimated variances of the outcome in the control and intervention groups, respectively.

Let
$\left({\hat{\boldsymbol{\Sigma}}(\boldsymbol{x}_1^{(2,n^{(1)})},...,\boldsymbol{x}_J^{(2,n^{(1)})}; \hat{\boldsymbol{\beta}}^{(1)})}\right)^{-1}_{\boldsymbol{\beta}_1}$ be the sub-matrix of $\left({\hat{\boldsymbol{\Sigma}}(\boldsymbol{x}_1^{(2,n^{(1)})},...,\boldsymbol{x}_J^{(2,n^{(1)})}; \hat{\boldsymbol{\beta}}^{(1)})}\right)^{-1}$ relevant only to $\boldsymbol{\beta}_1$. 
Let 
\begin{equation*}
\begin{aligned}
{\lambda}\bigl(\boldsymbol{x}_1^{(2,n^{(1)})},...,\boldsymbol{x}_J^{(2,n^{(1)})}; \hat{\boldsymbol{\beta}}^{(1)}\bigl) = n\; \left(\hat{\boldsymbol{\beta}}_1^{(1)}\right)^T \left({\hat{\boldsymbol{\Sigma}}(\boldsymbol{x}_1^{(2,n^{(1)})},...,\boldsymbol{x}_J^{(2,n^{(1)})}; \hat{\boldsymbol{\beta}}^{(1)})}\right)^{-1}_{\boldsymbol{\beta}_1} \hat{\boldsymbol{\beta}}_1^{(1)}.
\end{aligned}
\end{equation*}

Under Assumptions 2 --- 8 from the main text, the stage 2 recommended interventions \\$\boldsymbol{x}_1^{(2,n^{(1)})},...,\boldsymbol{x}_J^{(2,n^{(1)})}$, subject to both an outcome goal and an unconditional power goal for the $P$-df test, are the solution to the following optimization problem: 
    \begin{equation*}
    \begin{aligned}
    &\text{Min}_{\boldsymbol{x}_1,..., \boldsymbol{x}_{j}} C\left(\boldsymbol{x}_1, ..., \boldsymbol{x}_{j}\right)
    \; \text {subject to} \; 
    \frac{1}{J}\sum_{j=1}^{J}
    g^{-1}(\hat{\beta}_{0}^{(1)}+(\hat{\boldsymbol{\beta}}^{(1)}_1)^T {\boldsymbol{x}_j^{(2,n^{(1)})}}) \geq \tilde{\theta}, \\
    &\;\;\text{and}\; 
     {\lambda}\bigl(\boldsymbol{x}_1^{(2,n^{(1)})},...,\boldsymbol{x}_J^{(2,n^{(1)})}; \hat{\boldsymbol{\beta}}^{(1)}\bigl) \geq {\lambda}_{P,min}.
    \end{aligned}
    \end{equation*}
\end{thm}

The proof of Theorem \ref{cts pdf thm} is similar to the proof of Theorem \ref{binary pdf thm} and is therefore omitted.

\section{Computing Stage 1 Recommended Interventions Based on Pre-trial Information}\label{both lago stages}

Section \ref{both lago stages} describes the calculation of the stage 1 recommended intervention, ${\boldsymbol{x}}^{(1)}$ based on pre-trial information, considering both an outcome goal and a power goal.  

The calculation of ${\boldsymbol{x}}^{(1)}$ begins with initial guesses of the coefficient estimates $\boldsymbol{\beta}$, denoted as $\bigl({\boldsymbol{\beta}}^{(0)}\bigr)^T = \bigl( {\beta}_0^{(0)}, \bigl({\boldsymbol{\beta}}_1^{(0)}\bigr)^T \bigr)$. ${\boldsymbol{\beta}}^{(0)}$ may come from prior research or through iterative discussions with subject matter experts. Experts typically propose initial values of ${\boldsymbol{\beta}}^{(0)}$ based on their domain knowledge and relevant studies. Statisticians then conduct preliminary analyses, such as calculating projected means or probabilities, and review these results with experts to ensure both practicality and alignment with trial objectives.  
After establishing the initial guess ${\boldsymbol{\beta}}^{(0)}$, the stage 1 recommended intervention ${\boldsymbol{x}}^{(1)}$ can be obtained by solving:
\begin{equation}\label{stage 1 rec int optimization}
\text{Min}_{\boldsymbol{x}} C\left(\boldsymbol{x}\right)
\; \text {such that} 
\begin{cases}
&{p}_{\boldsymbol{x}}\left( {\boldsymbol{\beta}}^{(0)} \right) \geq \tilde{p} \;\; 
\text{or} \;\; {E}\left(Y_{ij} | \boldsymbol{x}; {\boldsymbol{\beta}}^{(0)} \right) \geq \tilde{\mu} \;\;
, \\
&{Power}\left(\boldsymbol{x}; {\boldsymbol{\beta}}^{(0)} \right) \geq \Pi.
\end{cases}
\end{equation}

In planning LAGO stage 1, we assume that recommended interventions will be used for both stage 1 and stage 2.
Thus, power goal calculations are based on $\boldsymbol{\beta}^{(0)}$, assuming all stage 1 and stage 2 centers in the intervention group were to implement these recommended interventions.  
While methods for calculating power during LAGO stage 1 planning are similar to conventional power calculations, LAGO trials uniquely calculate the intervention component compositions by considering the power goal, the outcome goal, and the cost simultaneously.
The LAGO design ensures that if observed stage 1 data indicate unmet pre-trial expectations, stage 2 recommended interventions will be adjusted to better align with these expectations.

We present methods for calculating stage 1 recommended interventions based on an outcome goal and a power goal of a two-sample z-test for proportions in the context of LAGO trials with a binary outcome.

First, consider the outcome goal of equation (\ref{stage 1 rec int optimization}), ${p}_{\boldsymbol{x}}\bigl( {\boldsymbol{\beta}}^{(0)} \bigl) \geq \tilde{p}$, which implies that $\boldsymbol{x}^{(1)}$ needs to satisfy 
\begin{equation}\label{binary stage 1 rec outcome goal}
\operatorname{expit}\left({\beta}_0^{(0)}+\left({\boldsymbol{\beta}}_1^{(0)}\right)^T \boldsymbol{x}^{(1)}\right)
\geq \tilde{p}.    
\end{equation}

Next, consider the same power goal as from equation (\ref{stage 1 rec int optimization}). Unlike Section 3 of the main text, there is no conditional approach for incorporating such power goal as there is no data to be conditioned on before the start of the trial. 
Theorem \ref{unconditional power theorem stage 1} presents the unconditional power approach.

\begin{thm}\label{unconditional power theorem stage 1}
    \textcolor{white}{xxx}\\
    Let $\chi^2_{\alpha, 1}$ be the upper $\alpha$ quantile of the central $\chi^2$ distribution with 1 degree of freedom. 
    For $\alpha=0.05$, $\chi^2_{\alpha, 1}=3.84$. 
    Let ${\lambda}_{min}$ be the minimum value of the non-centrality parameter for the non-central $\chi^2$ distribution with 1 degree of freedom, so that for a variable $T$ from a non-central $\chi^2$ distribution with non-centrality parameter ${\lambda}_{min}$, the probability of $T$ exceeding $\chi^2_{\alpha, 1}$ equals $\Pi$. 
    Let
    \begin{equation*}
        p_{\boldsymbol{x}^{(1)}}\left({\boldsymbol{\beta}}^{(0)}\right) = \operatorname{expit}\left({\beta}_0^{(0)}+\left({\boldsymbol{\beta}}_1^{(0)}\right)^T \boldsymbol{x}^{(1)}\right).
    \end{equation*}
    Under Assumptions 1, 3, 4, 5, 7 and 8 from the main text, the stage 1 recommended intervention ${\boldsymbol{x}}^{(1)}$, subject to both an outcome goal and a power goal, solves the following optimization problem: 
    \begin{equation*}
        \text{Min}_{\boldsymbol{x}} C\left(\boldsymbol{x}\right)
        \; \text {subject to} \; 
        \operatorname{expit}\left({\beta}_{0}^{(0)}+({\boldsymbol{\beta}}^{(0)}_1)^T {\boldsymbol{x}^{(1)}} \right)
        \geq \tilde{p}, \;\text{and}
    \end{equation*}
    \begin{equation*}
    \begin{aligned}
        \frac{ p_{\boldsymbol{x}^{(1)}}\left({\boldsymbol{\beta}}^{(0)}\right)  - \operatorname{expit}\left( {\beta}_0^{(0)}\right) }{\sqrt{ p_{\boldsymbol{x}^{(1)}}\left({\boldsymbol{\beta}}^{(0)}\right) \left.\left( 1- p_{\boldsymbol{x}^{(1)}}\left({\boldsymbol{\beta}}^{(0)}\right) \right) \right/{N_1} + \operatorname{expit}\left( {\beta}_0^{(0)}\right) \left.\left( 1- \operatorname{expit}\left( {\beta}_0^{(0)}\right) \right) \right/{N_0}}} \geq {\lambda}_{min}.
    \end{aligned}
    \end{equation*}
\end{thm}
The proof of Theorem \ref{unconditional power theorem stage 1} is omitted here since it follows similar arguments to the proof of Theorem 1 from the main text.

\section{Simulations: Additional Details and Results} \label{additional simulation tables}
Section \ref{cubic cost motivation section} describes the motivation behind using a cubic cost function in health care settings, and provides details of the cubic cost function used in Section 5 of the main text. 
Section \ref{simulation scenario 1a appendix} provides the full simulation results for Scenario 1a.
Section \ref{simulation scenario 1b appendix} provides simulation results for Scenario 1b. 
Section \ref{simulation scenario 2a appendix} provides the full simulation results for Scenario 2a.
Section \ref{simulation scenario 2b appendix} provides simulation results for Scenario 2b.

\subsection{The Cubic Cost Function}\label{cubic cost motivation section}
As described in Section 5 of the main text, we considered a cubic cost function. The idea behind using a cubic cost function is to mimic the cost function to be closer to the real world setting where cost initially has a economy of scale, followed by increasing marginal cost once the intervention component levels are above certain thresholds. This type of cost function is particularly useful in the context of implementation science in health-related fields (see Appendix Section E.5 of \citet{bing2023learnasyougo}). 

The cubic cost function considered in the simulations was $C(\boldsymbol{x})=2x_1^3-1.19x_1^2+10x_1+10 + 0.1x_2^3-0.2x_2^2+2x_2$. Figure \ref{fig:cubic-cost} shows the pattern of the initial economy of scale followed by increasing marginal cost after certain thresholds for both components $x_1$ and $x_2$. 

\begin{figure}[hbtp!]
    \centering
    \includegraphics[width=1\linewidth]{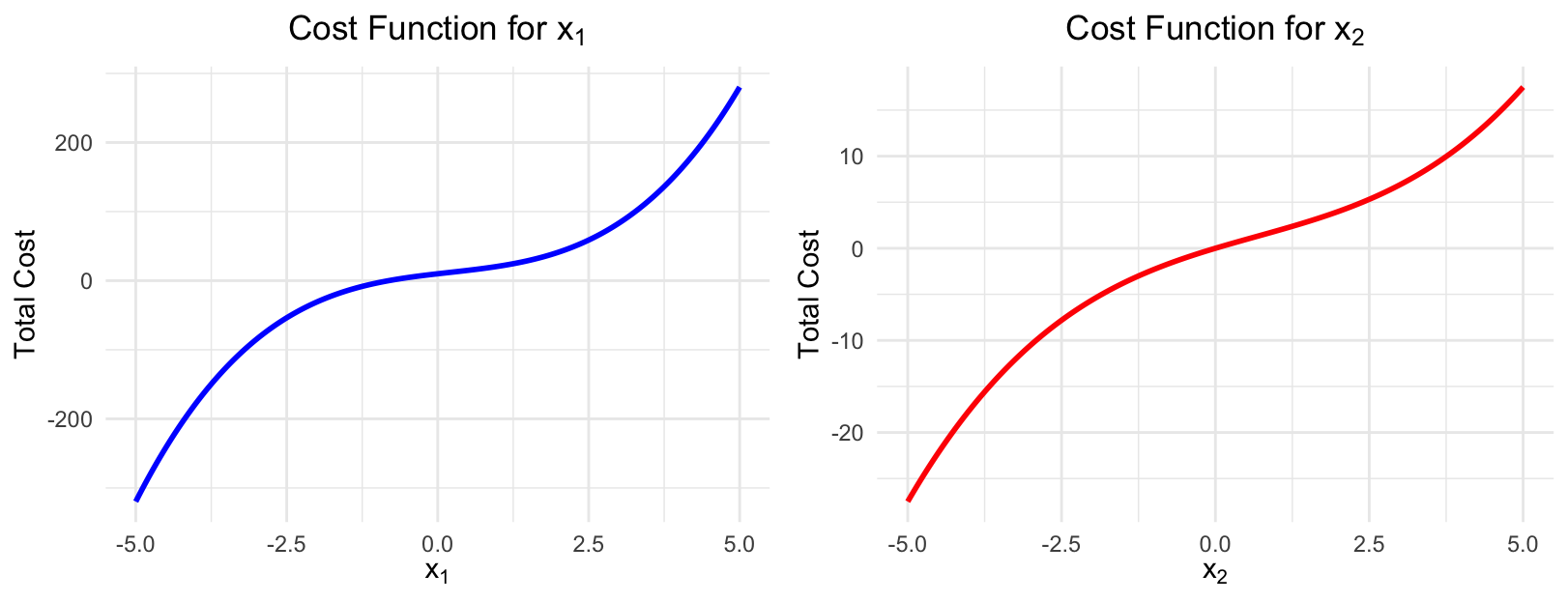}
    \caption{Cubic cost functions for $x_1$ and $x_2$.}
    \label{fig:cubic-cost}
\end{figure}

\newpage 
\subsection{Scenario 1a Additional Results} \label{simulation scenario 1a appendix}
Appendix \ref{simulation scenario 1a appendix} provides the full results for Simulation Scenario 1a.

Scenario 1a compared two optimization criteria for calculating the stage 2 recommended interventions: based on only an outcome goal, and based on both an outcome goal and a power goal of the two-sample z-test for the difference between two proportions. 
The total number of centers in both the intervention group and the control group was $J = 4$ for each stage.
The per-center sample sizes considered were $n_j^{(1)}=n_j^{(2)}=40, 50, 60, \text{and}\; 100$.
The minimum and maximum values of $x_1$ and $x_2$ were $[L_1, U_1] = [0,2]$ and $[L_2, U_2] = [0,8]$, respectively. 
Stage 1 of the two-stage LAGO trial used a fractional factorial design, setting the two-component intervention package ($x_1$, $x_2$) at $(0,0)$ in the control group, and $(1,0)$, $(0,4)$, $(1,4)$ in the intervention group.
The true coefficients were $(\beta^*_0, \beta^*_{11}, \beta^*_{12}) = (0.1, 0.3, 0.15)$.
The model for the binary outcome was $\operatorname{logit}\left(\operatorname{pr}\left(Y_{i j}=1 \mid \boldsymbol{A}=\right. \left.\boldsymbol{a}, \boldsymbol{X}=\boldsymbol{x}; \boldsymbol{\beta}\right)\right)
= \beta_0+\boldsymbol{\beta}_1^T \boldsymbol{a}.$
The power goals considered were $\Pi = 0.80, 0.85, 0.90 \;\text{and}\; 0.95$.
The cubic cost function of Section \ref{cubic cost motivation section} for the intervention package was considered:
$C(\boldsymbol{x})=2x_1^3-1.19x_1^2+10x_1+10 + 0.1x_2^3-0.2x_2^2+2x_2$.
The outcome goal was $\tilde{p} = 0.7$, so that the true optimal intervention based on the outcome goal only was $(0.5, 4.0)$ within two decimal places. 

\newpage
\begin{table}[h!]
\renewcommand{\arraystretch}{0.5}
    \caption{
    Scenario 1a simulation results for individual package component effects and power in LAGO trials with binary outcomes and a cubic cost function.}
    \vspace{0.4cm}
    \centering
    \scriptsize{
        \begin{tabular}{p{0.1in}p{0.1in}lp{0.25in}lllllllllllll}
          &  &  &  &  & \multicolumn{3}{c}{$\hat{\beta}_{11}$} &  & \multicolumn{3}{c}{$\hat{\beta}_{12}$} &  &  \\ \cline{6-8} \cline{10-12}
            $n_j^{(1)}$ & $n_j^{(2)}$ & \begin{tabular}[c]{@{}l@{}}\%\\ Power\\ Goal\end{tabular} & \begin{tabular}[c]{@{}l@{}}Power\\ Goal\\ Approach\\ (U/C)\end{tabular} &  & \begin{tabular}[c]{@{}l@{}}\%Rel\\ Bias\end{tabular} & \begin{tabular}[c]{@{}l@{}}SE/\\ EMP.SD\\ ($\times 100$)\end{tabular} & CP95 &  & \begin{tabular}[c]{@{}l@{}}\%Rel\\ Bias\end{tabular} & \begin{tabular}[c]{@{}l@{}}SE/\\ EMP.SD\\ ($\times 100$)\end{tabular} & CP95 &  & \begin{tabular}[c]{@{}l@{}}\%\\ Power\end{tabular} \\ \hline
            40 & 40 & — & --- &  & 3.23 & 101.3 & 95.4 &  & 1.20 & 96.4 & 95.0 &  & 68.8* \\
             &  & 80 & C &  & 4.08 & 101.8 & 95.4 &  & 2.63 & 94.8 & 94.9 &  & 81.2 \\
             &  &  & U &  & 4.23 & 102.3 & 95.5 &  & 2.89 & 95.2 & 95.0 &  & 83.2 \\
             &  & 85 & C &  & 4.02 & 101.9 & 95.5 &  & 2.77 & 95.0 & 94.9 &  & 82.6 \\
             &  &  & U &  & 4.31 & 102.7 & 95.7 &  & 3.27 & 95.7 & 95.0 &  & 84.4 \\
             &  & 90 & C &  & 4.19 & 102.6 & 95.6 &  & 3.06 & 95.5 & 95.0 &  & 83.9 \\
             &  &  & U &  & 3.68 & 103.2 & 95.6 &  & 3.59 & 96.6 & 95.0 &  & 85.6 \\
             &  & 95 & C &  & 3.80 & 103.1 & 95.7 &  & 3.60 & 96.5 & 95.1 &  & 85.5 \\
             &  &  & U &  & 3.07 & 102.9 & 95.7 &  & 4.07 & 98.2 & 95.1 &  & 87.8 \\ \hline
            50 & 50 & — & --- &  & 1.09 & 99.9 & 94.5 &  & 0.94 & 97.8 & 94.9 &  & 79.3* \\
             &  & 80 & C &  & 1.14 & 100.0 & 94.5 &  & 1.54 & 97.0 & 95.1 &  & 87.4 \\
             &  &  & U &  & 1.59 & 100.0 & 94.5 &  & 1.75 & 97.1 & 94.9 &  & 88.7 \\
             &  & 85 & C &  & 1.47 & 100.1 & 94.5 &  & 1.71 & 96.9 & 95.0 &  & 88.2 \\
             &  &  & U &  & 2.06 & 100.1 & 94.6 &  & 1.99 & 97.2 & 95.0 &  & 89.9 \\
             &  & 90 & C &  & 1.74 & 100.1 & 94.5 &  & 1.80 & 97.3 & 94.9 &  & 89.5 \\
             &  &  & U &  & 2.36 & 100.3 & 94.6 &  & 2.28 & 98.1 & 95.1 &  & 91.3 \\
             &  & 95 & C &  & 2.37 & 100.2 & 94.7 &  & 2.18 & 97.7 & 95.0 &  & 90.8 \\
             &  &  & U &  & 2.01 & 100.7 & 94.9 &  & 2.55 & 99.0 & 95.0 &  & 92.8 \\ \hline
            60 & 60 & — & --- &  & 4.42 & 103.3 & 95.6 &  & 2.47 & 103.6 & 96.1 &  & 87.5* \\
             &  & 80 & C &  & 5.02 & 102.7 & 95.5 &  & 2.44 & 103.6 & 96.0 &  & 92.5 \\
             &  &  & U &  & 5.15 & 103.0 & 95.6 &  & 2.43 & 103.8 & 96.0 &  & 93.2 \\
             &  & 85 & C &  & 5.15 & 103.0 & 95.5 &  & 2.45 & 103.7 & 96.1 &  & 93.0 \\
             &  &  & U &  & 5.26 & 102.9 & 95.5 &  & 2.45 & 104.3 & 96.2 &  & 94.0 \\
             &  & 90 & C &  & 5.25 & 102.9 & 95.5 &  & 2.47 & 104.1 & 96.0 &  & 93.7 \\
             &  &  & U &  & 5.19 & 102.7 & 95.2 &  & 2.56 & 104.7 & 96.2 &  & 95.1 \\
             &  & 95 & C &  & 5.09 & 102.7 & 95.3 &  & 2.57 & 104.5 & 96.1 &  & 94.9 \\
             &  &  & U &  & 5.73 & 103.2 & 95.2 &  & 2.68 & 105.7 & 96.1 &  & 96.4 \\ \hline
            100 & 100 & --- & --- &  & 5.18 & 99.6 & 94.3 &  & 3.42 & 98.7 & 95.7 &  & 98.2*
            \end{tabular}
    }
    \label{simulation 1a table 1}
    \vspace{0.4cm}
    \scriptsize \\
    {
    \setlength{\baselineskip}{0.5\baselineskip}
        \raggedright{
            2000 simulated datasets.  
            Number of centers for each stage: $J$ = 4. 
            True coefficients: $(\beta^*_0, \beta^*_{11}, \beta^*_{12}) = (0.1, 0.3, 0.15)$. \\
            A fractional factorial design was used in stage 1 with intervention package components set to: (0,0), (1,0), (0,4), and (1,4), and true success probabilities = 0.525, 0.599, 0.668, and 0.731.\\ 
            Cost function (cubic): $C(\boldsymbol{x})=2x_1^3-1.19x_1^2+10x_1+10 + 0.1x_2^3-0.2x_2^2+2x_2$. 
            Ranges: for $x_1$, $[0,2]$, for $x_2$, $[0,8]$. \\ 
            Outcome goal: $\Tilde{p}$ = 0.7. 
            No unplanned variation in the interventions.\\
            \medskip
            $^*$: percent power calculated using the LAGO design with only an outcome goal.\\
            $n_j^{(1)}$: number of participants in each center $j$ at stage 1. $n_j^{(2)}$: number of participants in each center $j$ at stage 2.\\
            \%Power Goal: percent power goal of the two-sample z-test for the difference between two proportions. For scenarios with only an outcome goal, the \%Power Goal is set to ---.  \\
            Power Goal Approach (U/C): U -- the unconditional power approach, C -- the conditional power approach. \\ 
            \%RelBias: percent relative bias $|100(\hat{\beta}-\beta^\star)/\beta^\star|$.\\
            SE: mean estimated standard error, 
            EMP.SD: empirical standard deviation.\\
            CP95: empirical coverage rate of the 95\% confidence intervals.\\
            \%Power: percent power of the two-sample z-test for the difference between two proportions at the end of the LAGO trial.\\
        }}
\end{table}

\newpage
\begin{table}[h]
\renewcommand{\arraystretch}{0.5}
    \caption{
    Scenario 1a simulation results for the estimated optimal intervention in LAGO trials with binary outcomes and a cubic cost function.}
    \vspace{0.4cm}
    \centering
    \scriptsize{
        \begin{tabular}{llllllllllll}
         &  &  &  &  & \multicolumn{3}{c}{Stage 1} &  & \multicolumn{3}{c}{Stage 2 / LAGO Optimized} \\ \cline{6-8} \cline{10-12} 
        $n_j^{(1)}$ & $n_j^{(2)}$ & \begin{tabular}[c]{@{}l@{}}\%\\ Power\\ Goal\end{tabular} & \begin{tabular}[c]{@{}l@{}}Power\\ Goal\\ Approach\\ (U/C)\end{tabular} &  & \begin{tabular}[c]{@{}l@{}}\%Rel\\ Bias\\ $\hat{x}_1^{opt}$\end{tabular} & \begin{tabular}[c]{@{}l@{}}\%Rel\\ Bias\\ $\hat{x}_2^{opt}$\end{tabular} & \begin{tabular}[c]{@{}l@{}}PrOpt\\ (Q2.5,Q97.5)\end{tabular} &  & \begin{tabular}[c]{@{}l@{}}\%Rel\\ Bias\\ $\hat{x}_1^{opt}$\end{tabular} & \begin{tabular}[c]{@{}l@{}}\%Rel\\ Bias\\ $\hat{x}_2^{opt}$\end{tabular} & \begin{tabular}[c]{@{}l@{}}PrOpt\\ (Q2.5,Q97.5)\end{tabular} \\ \hline
        40 & 40 & --- & --- &  & 3.5 & 23.8 & (0.570, 0.779) &  & 12.9 & 13.3 & (0.606, 0.788) \\
         &  & 80 & C &  &  &  &  &  & 11.5 & 14.1 & (0.605, 0.786) \\
         &  &  & U &  &  &  &  &  & 11.1 & 14.2 & (0.605, 0.786) \\
         &  & 85 & C &  &  &  &  &  & 11.4 & 14.2 & (0.605, 0.786) \\
         &  &  & U &  &  &  &  &  & 10.6 & 14.3 & (0.605, 0.786) \\
         &  & 90 & C &  &  &  &  &  & 11.0 & 14.3 & (0.605, 0.786) \\
         &  &  & U &  &  &  &  &  & 9.6 & 14.2 & (0.604, 0.786) \\
         &  & 95 & C &  &  &  &  &  & 9.9 & 14.2 & (0.604, 0.786) \\
         &  &  & U &  &  &  &  &  & 7.3 & 14.4 & (0.602, 0.784) \\ \hline
        50 & 50 & --- & --- &  & 9.5 & 19.7 & (0.581, 0.767) &  & 8.8 & 11.0 & (0.612, 0.786) \\
         &  & 80 & C &  &  &  &  &  & 7.3 & 11.1 & (0.616, 0.775) \\
         &  &  & U &  &  &  &  &  & 7.1 & 11.3 & (0.615, 0.777) \\
         &  & 85 & C &  &  &  &  &  & 7.3 & 11.4 & (0.615, 0.774) \\
         &  &  & U &  &  &  &  &  & 6.8 & 11.5 & (0.614, 0.774) \\
         &  & 90 & C &  &  &  &  &  & 6.9 & 11.4 & (0.615, 0.775) \\
         &  &  & U &  &  &  &  &  & 5.9 & 11.5 & (0.612, 0.775) \\
         &  & 95 & C &  &  &  &  &  & 5.8 & 11.5 & (0.612, 0.773) \\
         &  &  & U &  &  &  &  &  & 4.7 & 11.5 & (0.612, 0.770) \\ \hline
        60 & 60 & --- & --- &  & 4.1 & 15.9 & (0.593, 0.772) &  & 7.6 & 10.4 & (0.629, 0.759) \\
         &  & 80 & C &  &  &  &  &  & 6.8 & 11.0 & (0.628, 0.757) \\
         &  &  & U &  &  &  &  &  & 6.4 & 10.8 & (0.628, 0.757) \\
         &  & 85 & C &  &  &  &  &  & 6.2 & 10.8 & (0.628, 0.757) \\
         &  &  & U &  &  &  &  &  & 5.8 & 10.7 & (0.629, 0.757) \\
         &  & 90 & C &  &  &  &  &  & 5.9 & 10.7 & (0.629, 0.757) \\
         &  &  & U &  &  &  &  &  & 5.3 & 10.6 & (0.629, 0.757) \\
         &  & 95 & C &  &  &  &  &  & 5.3 & 10.6 & (0.628, 0.757) \\
         &  &  & U &  &  &  &  &  & 4.5 & 10.8 & (0.626, 0.756) \\ \hline
        100 & 100 & --- & --- &  & 5.7 & 10.3 & (0.614, 0.761) &  & 6.0 & 7.7 & (0.645, 0.744)
        \end{tabular}
        \label{simulation 1a table 2}
    }
    \vspace{0.4cm}
    \scriptsize \\
    {
    \setlength{\baselineskip}{0.5\baselineskip}
        \raggedright{
            2000 simulated datasets.  
            Number of centers for each stage: $J$ = 4. 
            True coefficients: $(\beta^*_0, \beta^*_{11}, \beta^*_{12}) = (0.1, 0.3, 0.15)$. \\
            A fractional factorial design was used in stage 1 with intervention package components set to: (0,0), (1,0), (0,4), and (1,4), and true success probabilities = 0.525, 0.599, 0.668, and 0.731.\\ 
            Cost function (cubic): $C(\boldsymbol{x})=2x_1^3-1.19x_1^2+10x_1+10 + 0.1x_2^3-0.2x_2^2+2x_2$. 
            Ranges: for $x_1$, $[0,2]$, for $x_2$, $[0,8]$. \\ 
            Outcome goal: $\Tilde{p}$ = 0.7. 
            No unplanned variation in the interventions.\\
            \medskip
            $n_j^{(1)}$: number of participants in each center $j$ at stage 1. 
            $n_j^{(2)}$: number of participants in each center $j$ at stage 2.\\
            \%Power Goal: percent power goal of the two-sample z-test for the difference between two proportions. For scenarios with only an outcome goal, the \%Power Goal is set to ---.  \\
            Power Goal Approach (U/C): U -- the unconditional power approach, C -- the conditional power approach. \\ 
            \%Rel Bias $\hat{x}_1^{opt}$: relative bias of the first component of the estimated optimal intervention $100(\hat{x}_1^{opt} - {x}_1^{opt})/{x}_1^{opt}$. \\
            \%Rel Bias $\hat{x}_2^{opt}$: relative bias of the second component of the estimated optimal intervention $100(\hat{x}_2^{opt} - {x}_2^{opt})/{x}_2^{opt}$.\\
            PrOpt: true success probability under the recommended intervention, calculated using true coefficients. \\
            Q2.5, Q97.5: 2.5\% and 97.5\% quantiles.\\
        }}
\end{table}

\newpage 
\subsection{Scenario 1b Results}\label{simulation scenario 1b appendix}
Compared to Scenario 1a, Scenario 1b used the same true parameter values, number of centers, per-center sample sizes, minimum and maximum values of the two-component intervention package, and power goals. 
Notably, Scenario 1b considered a linear cost function for the intervention package: $C(\boldsymbol{x}) = x_1 + 4x_2$. The outcome goal was $\tilde{p}=0.7455$, so that the true optimal intervention based on only the outcome goal was $(3.25, 0)$ within three decimal places. 

\begin{table}[h!]
\renewcommand{\arraystretch}{0.5}
    \caption{
    Scenario 1b simulation results for individual package component effects and power in LAGO trials with binary outcomes and a linear cost function.}
    \vspace{0.4cm}
    \centering
    \scriptsize{
        \begin{tabular}{llllllllllllll}
         &  &  &  &  & \multicolumn{3}{c}{$\hat{\beta}_{11}$} &  & \multicolumn{3}{c}{$\hat{\beta}_{12}$} &  &  \\ \cline{6-8} \cline{10-12}
        $n_j^{(1)}$ & $n_j^{(2)}$ & \begin{tabular}[c]{@{}l@{}}\%\\ Power\\ Goal\end{tabular} & \begin{tabular}[c]{@{}l@{}}Power\\ Goal\\ Approach\\ (U/C)\end{tabular} &  & \begin{tabular}[c]{@{}l@{}}\%Rel\\ Bias\end{tabular} & \begin{tabular}[c]{@{}l@{}}SE/\\ EMP.SD\\ ($\times 100$)\end{tabular} & CP95 &  & \begin{tabular}[c]{@{}l@{}}\%Rel\\ Bias\end{tabular} & \begin{tabular}[c]{@{}l@{}}SE/\\ EMP.SD\\ ($\times 100$)\end{tabular} & CP95 &  & \begin{tabular}[c]{@{}l@{}}\%\\ Power\end{tabular} \\ \hline
        40 & 40 & — & --- &  & 13.52 & 89.9 & 94.0 &  & 4.91 & 97.6 & 94.9 &  & 75.9* \\
         &  & 80 & C &  & 12.03 & 91.5 & 94.2 &  & 4.04 & 95.7 & 95.0 &  & 81.1 \\
         &  &  & U &  & 11.74 & 92.2 & 94.4 &  & 3.88 & 95.1 & 94.8 &  & 82.5 \\
         &  & 85 & C &  & 11.90 & 91.9 & 94.3 &  & 3.98 & 95.3 & 94.9 &  & 82.1 \\
         &  &  & U &  & 11.20 & 93.0 & 94.5 &  & 3.71 & 94.8 & 94.8 &  & 83.2 \\
         &  & 90 & C &  & 11.46 & 92.7 & 94.4 &  & 3.80 & 95.0 & 94.8 &  & 82.9 \\
         &  &  & U &  & 11.03 & 94.0 & 94.6 &  & 3.44 & 94.8 & 94.9 &  & 84.2 \\
         &  & 95 & C &  & 11.07 & 93.7 & 94.5 &  & 3.46 & 94.8 & 94.9 &  & 84.0 \\
         &  &  & U &  & 10.28 & 95.0 & 94.6 &  & 2.47 & 94.4 & 94.9 &  & 85.6 \\ \hline
        50 & 50 & — & --- &  & 13.50 & 87.7 & 93.7 &  & 5.52 & 98.0 & 94.5 &  & 85.4* \\
         &  & 80 & C &  & 13.08 & 88.5 & 93.6 &  & 5.11 & 97.0 & 94.7 &  & 88.2 \\
         &  &  & U &  & 12.73 & 89.0 & 93.8 &  & 4.99 & 96.5 & 94.5 &  & 88.8 \\
         &  & 85 & C &  & 12.82 & 88.9 & 93.7 &  & 5.02 & 96.7 & 94.6 &  & 88.6 \\
         &  &  & U &  & 12.28 & 89.4 & 93.9 &  & 4.91 & 96.2 & 94.5 &  & 89.6 \\
         &  & 90 & C &  & 12.45 & 89.3 & 93.9 &  & 4.99 & 96.5 & 94.5 &  & 89.3 \\
         &  &  & U &  & 11.58 & 89.8 & 94.0 &  & 4.71 & 95.9 & 94.5 &  & 90.1 \\
         &  & 95 & C &  & 11.78 & 89.7 & 93.9 &  & 4.75 & 95.9 & 94.4 &  & 89.9 \\
         &  &  & U &  & 11.19 & 91.0 & 94.4 &  & 4.47 & 95.3 & 94.3 &  & 91.3 \\ \hline
        60 & 60 & — & --- &  & 11.02 & 88.1 & 94.4 &  & 4.89 & 103.5 & 96.8 &  & 92.0* \\
         &  & 80 & C &  & 10.66 & 88.3 & 94.4 &  & 4.90 & 103.1 & 96.7 &  & 93.1 \\
         &  &  & U &  & 10.51 & 88.8 & 94.6 &  & 4.90 & 102.8 & 96.7 &  & 93.6 \\
         &  & 85 & C &  & 10.51 & 88.6 & 94.4 &  & 4.90 & 102.9 & 96.7 &  & 93.4 \\
         &  &  & U &  & 10.24 & 89.1 & 94.6 &  & 4.92 & 102.8 & 96.7 &  & 93.8 \\
         &  & 90 & C &  & 10.35 & 88.9 & 94.5 &  & 4.94 & 102.9 & 96.7 &  & 93.7 \\
         &  &  & U &  & 9.88 & 89.3 & 94.4 &  & 4.90 & 102.6 & 96.7 &  & 94.4 \\
         &  & 95 & C &  & 9.93 & 89.2 & 94.5 &  & 4.89 & 102.7 & 96.7 &  & 94.4 \\
         &  &  & U &  & 9.01 & 90.4 & 94.4 &  & 4.85 & 102.1 & 96.6 &  & 95.0 \\ \hline
        100 & 100 & --- & --- &  & 6.03 & 86.2 & 93.9 &  & 3.49 & 102.4 & 95.6 &  & 98.8*
        \end{tabular}
    }
    \label{simulation 1b table 1}
    \vspace{0.4cm}
    \scriptsize \\
    {
    \setlength{\baselineskip}{0.5\baselineskip}
        \raggedright{
            2000 simulated datasets.  
            Number of centers for each stage: $J$ = 4. 
            True coefficients: $(\beta^*_0, \beta^*_{11}, \beta^*_{12}) = (0.1, 0.3, 0.15)$. \\
            A fractional factorial design was used in stage 1 with intervention package components set to: (0,0), (1,0), (0,4), and (1,4), and true success probabilities = 0.525, 0.599, 0.668, and 0.731. 
            Cost function (linear): $C(\boldsymbol{x})=x_1 + 4x_2$. \\
            Ranges: for $x_1$, $[0,4]$, for $x_2$, $[0,8]$. 
            Outcome goal: $\Tilde{p}$ = 0.7455. 
            No unplanned variation in the interventions.\\
            \medskip
            $^*$: percent power calculated using the LAGO design with only an outcome goal.\\
            $n_j^{(1)}$: number of participants in each center $j$ at stage 1. $n_j^{(2)}$: number of participants in each center $j$ at stage 2.\\
            \%Power Goal: percent power goal of the two-sample z-test for the difference between two proportions. For scenarios with only an outcome goal, the \%Power Goal is set to ---.  \\
            Power Goal Approach (U/C): U -- the unconditional power approach, C -- the conditional power approach. \\ 
            \%RelBias: percent relative bias $|100(\hat{\beta}-\beta^\star)/\beta^\star|$.\\
            SE: mean estimated standard error, 
            EMP.SD: empirical standard deviation.\\
            CP95: empirical coverage rate of the 95\% confidence intervals.\\
            \%Power: percent power of the two-sample z-test for the difference between two proportions at the end of the LAGO trial.\\
        }}
\end{table}

\newpage
\begin{table}[h!]
\renewcommand{\arraystretch}{0.5}
    \caption{
    Scenario 1b simulation results for the estimated optimal intervention in LAGO trials with binary outcomes and a linear cost function.}
    \vspace{0.4cm}
    \centering
    \scriptsize{
        \begin{tabular}{llllllllllll}
         &  &  &  &  & \multicolumn{3}{c}{Stage 1} &  & \multicolumn{3}{c}{Stage 2 / LAGO Optimized} \\ \cline{6-8} \cline{10-12} 
        $n_j^{(1)}$ & $n_j^{(2)}$ & \begin{tabular}[c]{@{}l@{}}\%\\ Power\\ Goal\end{tabular} & \begin{tabular}[c]{@{}l@{}}Power\\ Goal\\ Approach\\ (U/C)\end{tabular} &  & \begin{tabular}[c]{@{}l@{}}\%Rel\\ Bias\\ $\hat{x}_1^{opt}$\end{tabular} & \begin{tabular}[c]{@{}l@{}}\%Rel\\ Bias\\ $\hat{x}_2^{opt}$\end{tabular} & \begin{tabular}[c]{@{}l@{}}PrOpt\\ (Q2.5,Q97.5)\end{tabular} &  & \begin{tabular}[c]{@{}l@{}}\%Rel\\ Bias\\ $\hat{x}_1^{opt}$\end{tabular} & \begin{tabular}[c]{@{}l@{}}\%Rel\\ Bias\\ $\hat{x}_2^{opt}$\end{tabular} & \begin{tabular}[c]{@{}l@{}}PrOpt\\ (Q2.5,Q97.5)\end{tabular} \\ \hline
        40 & 40 & --- & --- &  & 31.3 & 6.8 & (0.550, 0.837) &  & 18.9 & 12.0 & (0.611, 0.838) \\
         &  & 80 & C &  &  &  &  &  & 18.9 & 9.5 & (0.610, 0.838) \\
         &  &  & U &  &  &  &  &  & 18.9 & 8.7 & (0.612, 0.838) \\
         &  & 85 & C &  &  &  &  &  & 18.9 & 9.3 & (0.612, 0.838) \\
         &  &  & U &  &  &  &  &  & 19.0 & 7.8 & (0.609, 0.838) \\
         &  & 90 & C &  &  &  &  &  & 18.9 & 8.1 & (0.610, 0.838) \\
         &  &  & U &  &  &  &  &  & 19.3 & 7.2 & (0.610, 0.836) \\
         &  & 95 & C &  &  &  &  &  & 19.3 & 7.3 & (0.609, 0.835) \\
         &  &  & U &  &  &  &  &  & 19.7 & 6.2 & (0.608, 0.832) \\ \hline
        50 & 50 & --- & --- &  & 28.2 & 2.1 & (0.577, 0.827) &  & 15.6 & 12.3 & (0.630, 0.830) \\
         &  & 80 & C &  &  &  &  &  & 15.9 & 11.1 & (0.630, 0.829) \\
         &  &  & U &  &  &  &  &  & 16.1 & 10.3 & (0.630, 0.827) \\
         &  & 85 & C &  &  &  &  &  & 16.1 & 10.5 & (0.630, 0.826) \\
         &  &  & U &  &  &  &  &  & 16.1 & 10.1 & (0.631, 0.827) \\
         &  & 90 & C &  &  &  &  &  & 16.1 & 10.2 & (0.631, 0.827) \\
         &  &  & U &  &  &  &  &  & 16.1 & 9.1 & (0.628, 0.824) \\
         &  & 95 & C &  &  &  &  &  & 16.0 & 9.3 & (0.629, 0.823) \\
         &  &  & U &  &  &  &  &  & 16.3 & 8.5 & (0.626, 0.826) \\ \hline
        60 & 60 & --- & --- &  & 21.3 & 0.6 & (0.612, 0.835) &  & 11.1 & 11.2 & (0.652, 0.830) \\
         &  & 80 & C &  &  &  &  &  & 11.2 & 10.6 & (0.649, 0.829) \\
         &  &  & U &  &  &  &  &  & 11.1 & 10.5 & (0.648, 0.829) \\
         &  & 85 & C &  &  &  &  &  & 11.2 & 10.5 & (0.649, 0.829) \\
         &  &  & U &  &  &  &  &  & 11.1 & 10.0 & (0.648, 0.826) \\
         &  & 90 & C &  &  &  &  &  & 11.3 & 10.2 & (0.648, 0.826) \\
         &  &  & U &  &  &  &  &  & 11.4 & 9.3 & (0.645, 0.824) \\
         &  & 95 & C &  &  &  &  &  & 11.3 & 9.4 & (0.646, 0.824) \\
         &  &  & U &  &  &  &  &  & 11.2 & 7.7 & (0.643, 0.825) \\ \hline
        100 & 100 & --- & --- &  & 14.5 & 3.6 & (0.627, 0.834) &  & 5.8 & 7.6 & (0.665, 0.814)
        \end{tabular}
        \label{simulation 1b table 2}
    }
    \vspace{0.4cm}
    \scriptsize \\
    {
    \setlength{\baselineskip}{0.5\baselineskip}
        \raggedright{
            2000 simulated datasets.  
            Number of centers for each stage: $J$ = 4. 
            True coefficients: $(\beta^*_0, \beta^*_{11}, \beta^*_{12}) = (0.1, 0.3, 0.15)$. \\
            A fractional factorial design was used in stage 1 with intervention package components set to: (0,0), (1,0), (0,4), and (1,4), and true success probabilities = 0.525, 0.599, 0.668, and 0.731. 
            Cost function (linear): $C(\boldsymbol{x})=x_1 + 4x_2$. \\
            Ranges: for $x_1$, $[0,4]$, for $x_2$, $[0,8]$. 
            Outcome goal: $\Tilde{p}$ = 0.7455. 
            No unplanned variation in the interventions.\\
            \medskip
            $n_j^{(1)}$: number of participants in each center $j$ at stage 1. 
            $n_j^{(2)}$: number of participants in each center $j$ at stage 2.\\
            \%Power Goal: percent power goal of the two-sample z-test for the difference between two proportions. For scenarios with only an outcome goal, the \%Power Goal is set to ---.  \\
            Power Goal Approach (U/C): U -- the unconditional power approach, C -- the conditional power approach. \\ 
            \%Rel Bias $\hat{x}_1^{opt}$: relative bias of the first component of the estimated optimal intervention $100(\hat{x}_1^{opt} - {x}_1^{opt})/{x}_1^{opt}$. \\
            \%Rel Bias $\hat{x}_2^{opt}$: relative bias of the second component of the estimated optimal intervention $100(\hat{x}_2^{opt} - {x}_2^{opt})/{x}_2^{opt}$.\\
            PrOpt: true success probability under the recommended intervention, calculated using true coefficients. \\
            Q2.5, Q97.5: 2.5\% and 97.5\% quantiles.\\
        }}
\end{table}

\newpage 
\textcolor{white}{xxx}
\newpage 
\subsection{Scenario 2a Results}\label{simulation scenario 2a appendix}

Scenario 2 includes 2 cases: 2a and 2b, both comparing the performance of the LAGO design with a power goal to the non-LAGO, fractional factorial design, with a focus on power. The fractional factorial design used the same interventions as stage 1 of the LAGO trial throughout.  
The setups for Scenarios 2a and 2b were similar to Scenarios 1a and 1b, with the difference being that the LAGO designs used in Scenarios 2a and 2b only included a power goal. 
Scenario 2a considered the same cubic cost function as described in Scenario 1a, while Scenario 2b considered the same linear cost function as described in Scenario 1b. 
Other parameters for Scenarios 2a and 2b, including the true coefficients, minimum and maximum values of the intervention components, model for the outcome, outcome goal, power goal, and sample sizes were the same as described in Scenarios 1a and 1b, respectively. 

Table \ref{simulation 2a table} reports bias, the ratios between the average estimated standard error and the empirical standard error, coverage rates of the 95\% confidence intervals of the coefficients, and the power of the two-sample z-test for the difference between two proportions at the end of the LAGO trials. The first entry for each set of per-center sample sizes shows results under the fractional factorial design.

Under the LAGO design with a power goal, the relative bias of the coefficients, the ratio of the average estimated standard error to the empirical standard error of the coefficients, and the empirical coverage rate of the 95\% confidence intervals were all comparable to those calculated under the fractional factorial design. 
Table \ref{simulation 2a table} shows that the LAGO design with a power goal had noticeably higher power compared to the fractional factorial design. As expected, higher power goals corresponded to higher resulting power at the end of the LAGO trial.  
With $n_j^{(1)}=n_j^{(2)}=50$ and higher power goals, e.g., when the power goal was 0.90, neither the unconditional power or the conditional power approach reached the power goal. The unconditional power approach led to higher final power compared to the conditional power approach. With $n_j^{(1)}=n_j^{(2)}=60$ and lower power goals, e.g., when the power goal was 0.8, both unconditional and conditional power approaches reached the power goal. As expected, the conditional power approach led to power that exceeded the desired power goal less than the unconditional power approach. 

\begin{table}[h!]
\renewcommand{\arraystretch}{0.5}
    \caption{
    Scenario 2a simulation results for individual package component effects and power in LAGO trials with binary outcomes and a cubic cost function.}
    \vspace{0.4cm}
    \centering
    \scriptsize{
        \begin{tabular}{p{0.1in}p{0.1in}lp{0.25in}lllllllllllll}
         &  &  &  &  & \multicolumn{3}{c}{$\hat{\beta}_{11}$} &  & \multicolumn{3}{c}{$\hat{\beta}_{12}$} &  &  \\ \cline{6-8} \cline{10-12}
            $n_j^{(1)}$ & $n_j^{(2)}$ & \begin{tabular}[c]{@{}l@{}}\%\\ Power\\ Goal\end{tabular} & \begin{tabular}[c]{@{}l@{}}Power\\ Goal\\ Approach\\ (U/C)\end{tabular} &  & \begin{tabular}[c]{@{}l@{}}\%Rel\\ Bias\end{tabular} & \begin{tabular}[c]{@{}l@{}}SE/\\ EMP.SD\\ ($\times 100$)\end{tabular} & CP95 &  & \begin{tabular}[c]{@{}l@{}}\%Rel\\ Bias\end{tabular} & \begin{tabular}[c]{@{}l@{}}SE/\\ EMP.SD\\ ($\times 100$)\end{tabular} & CP95 &  & \begin{tabular}[c]{@{}l@{}}\%\\ Power\end{tabular} \\ \hline
            40 & 40 & — & --- &  & 1.94 & 99.4 & 94.6 &  & 0.07 & 99.4 & 94.9 &  & 59.2* \\
             &  & 80 & C &  & 2.30 & 97.8 & 95.0 &  & 0.37 & 91.9 & 94.4 &  & 72.4 \\
             &  &  & U &  & 2.90 & 98.2 & 95.1 &  & 0.07 & 91.8 & 94.4 &  & 77.4 \\
             &  & 85 & C &  & 2.61 & 98.0 & 95.1 &  & 0.19 & 91.9 & 94.4 &  & 76.2 \\
             &  &  & U &  & 3.38 & 98.7 & 95.3 &  & 0.48 & 92.1 & 94.5 &  & 80.2 \\
             &  & 90 & C &  & 3.15 & 98.6 & 95.2 &  & 0.25 & 92.2 & 94.4 &  & 79.1 \\
             &  &  & U &  & 3.38 & 99.3 & 95.3 &  & 1.11 & 92.9 & 94.7 &  & 83.5 \\
             &  & 95 & C &  & 3.43 & 99.3 & 95.4 &  & 1.13 & 92.9 & 94.6 &  & 83.3 \\
             &  &  & U &  & 3.52 & 99.7 & 95.4 &  & 2.15 & 94.9 & 95.1 &  & 87.4 \\ \hline
            50 & 50 & — & --- &  & 0.90 & 96.0 & 94.4 &  & 0.53 & 102.4 & 95.6 &  & 68.3* \\
             &  & 80 & C &  & 1.45 & 94.9 & 94.0 &  & 1.88 & 93.0 & 94.6 &  & 75.7 \\
             &  &  & U &  & 0.42 & 95.2 & 94.0 &  & 1.67 & 92.4 & 94.4 &  & 80.2 \\
             &  & 85 & C &  & 0.80 & 95.2 & 94.0 &  & 1.70 & 92.5 & 94.5 &  & 78.6 \\
             &  &  & U &  & 0.56 & 95.7 & 93.9 &  & 1.35 & 92.3 & 94.6 &  & 83.8 \\
             &  & 90 & C &  & 0.12 & 95.6 & 93.9 &  & 1.51 & 92.5 & 94.5 &  & 82.5 \\
             &  &  & U &  & 1.63 & 96.5 & 94.2 &  & 0.62 & 93.0 & 94.9 &  & 87.8 \\
             &  & 95 & C &  & 1.56 & 96.3 & 94.2 &  & 0.76 & 92.8 & 94.7 &  & 87.1 \\
             &  &  & U &  & 2.09 & 97.5 & 94.7 &  & 0.27 & 94.3 & 94.8 &  & 91.7 \\ \hline
            60 & 60 & --- & --- &  & 1.72 & 102.7 & 95.5 &  & 0.88 & 102.0 & 95.4 &  & 76.5* \\
             &  & 80 & C &  & 1.37 & 98.4 & 95.4 &  & 1.18 & 100.6 & 96.1 &  & 80.0 \\
             &  &  & U &  & 2.01 & 98.9 & 95.5 &  & 1.16 & 100.7 & 96.3 &  & 83.6 \\
             &  & 85 & C &  & 1.91 & 98.8 & 95.4 &  & 1.10 & 100.8 & 96.4 &  & 82.5 \\
             &  &  & U &  & 2.58 & 99.3 & 95.2 &  & 0.83 & 100.8 & 96.3 &  & 86.6 \\
             &  & 90 & C &  & 2.42 & 99.1 & 95.4 &  & 0.84 & 101.0 & 96.3 &  & 85.7 \\
             &  &  & U &  & 3.21 & 99.7 & 95.0 &  & 0.38 & 101.0 & 96.2 &  & 91.0 \\
             &  & 95 & C &  & 2.99 & 99.7 & 95.1 &  & 0.36 & 101.0 & 96.2 &  & 90.4 \\
             &  &  & U &  & 4.57 & 100.8 & 95.0 &  & 0.31 & 101.9 & 96.0 &  & 94.8
            \end{tabular}
        \label{simulation 2a table}
    }
    \vspace{0.4cm}
    \scriptsize \\
    {
    \setlength{\baselineskip}{0.5\baselineskip}
        \raggedright{
            2000 simulated datasets.  
            Number of centers for each stage: $J$ = 4. 
            True coefficients: $(\beta^*_0, \beta^*_{11}, \beta^*_{12}) = (0.1, 0.3, 0.15)$. \\
            A fractional factorial design was used in stage 1 with intervention package components set to: (0,0), (1,0), (0,4), and (1,4), and true success probabilities = 0.525, 0.599, 0.668, and 0.731. \\
            Cost function (cubic): $C(\boldsymbol{x})=2x_1^3-1.19x_1^2+10x_1+10 + 0.1x_2^3-0.2x_2^2+2x_2$. 
            Ranges: for $x_1$, $[0,2]$, for $x_2$, $[0,8]$. \\ 
            No unplanned variation in the interventions.\\
            \medskip
            $^*$: percent power calculated using the fractional factorial design.\\
            $n_j^{(1)}$: number of participants in each center $j$ at stage 1. $n_j^{(2)}$: number of participants in each center $j$ at stage 2.\\
            \%Power Goal: percent power goal of the two-sample z-test for the difference between two proportions. For scenarios without the power goal, the \%Power Goal is set to ---.  \\
            Power Goal Approach (U/C): U -- the unconditional power approach, C -- the conditional power approach. \\ 
            \%RelBias: percent relative bias $|100(\hat{\beta}-\beta^\star)/\beta^\star|$.\\
            SE: mean estimated standard error, 
            EMP.SD: empirical standard deviation.\\
            CP95: empirical coverage rate of the 95\% confidence intervals.\\
            \%Power: percent power of the two-sample z-test for the difference between two proportions at the end of the LAGO trial.\\
        }}
\end{table}

\newpage 
\textcolor{white}{xxx}
\newpage
\subsection{Simulation 2b Results}\label{simulation scenario 2b appendix}

\begin{table}[h!]
\renewcommand{\arraystretch}{0.5}
    \caption{
    Simulation 2b study results for individual package component effects and power in LAGO trials with a binary outcome and a linear cost function.}
    \vspace{0.4cm}
    \centering
    \scriptsize{
        \begin{tabular}{llllllllllllll}
         &  &  &  &  & \multicolumn{3}{c}{$\hat{\beta}_{11}$} &  & \multicolumn{3}{c}{$\hat{\beta}_{12}$} &  &  \\ \cline{6-8} \cline{10-12}
        $n_j^{(1)}$ & $n_j^{(2)}$ & \begin{tabular}[c]{@{}l@{}}\%\\ Power\\ Goal\end{tabular} & \begin{tabular}[c]{@{}l@{}}Power\\ Goal\\ Approach\\ (U/C)\end{tabular} &  & \begin{tabular}[c]{@{}l@{}}\%Rel\\ Bias\end{tabular} & \begin{tabular}[c]{@{}l@{}}SE/\\ EMP.SD\\ ($\times 100$)\end{tabular} & CP95 &  & \begin{tabular}[c]{@{}l@{}}\%Rel\\ Bias\end{tabular} & \begin{tabular}[c]{@{}l@{}}SE/\\ EMP.SD\\ ($\times 100$)\end{tabular} & CP95 &  & \begin{tabular}[c]{@{}l@{}}\%\\ Power\end{tabular} \\ \hline
        40 & 40 & — & --- &  & 1.94 & 99.4 & 94.6 &  & 0.07 & 99.4 & 94.9 &  & 59.2* \\
         &  & 80 & C &  & 2.45 & 88.6 & 94.6 &  & 1.61 & 95.0 & 94.5 &  & 72.7 \\
         &  &  & U &  & 1.90 & 88.1 & 94.5 &  & 1.78 & 95.2 & 94.4 &  & 77.0 \\
         &  & 85 & C &  & 1.95 & 88.1 & 94.6 &  & 1.82 & 95.1 & 94.5 &  & 76.1 \\
         &  &  & U &  & 1.44 & 88.2 & 94.5 &  & 1.94 & 95.4 & 94.4 &  & 79.6 \\
         &  & 90 & C &  & 1.65 & 88.2 & 94.6 &  & 1.86 & 95.5 & 94.4 &  & 78.7 \\
         &  &  & U &  & 0.46 & 88.2 & 94.5 &  & 2.10 & 95.9 & 94.7 &  & 83.1 \\
         &  & 95 & C &  & 0.53 & 88.0 & 94.5 &  & 2.07 & 95.8 & 94.6 &  & 82.5 \\
         &  &  & U &  & 1.35 & 89.2 & 94.4 &  & 1.84 & 95.9 & 94.8 &  & 86.7 \\ \hline
        50 & 50 & — & --- &  & 0.90 & 96.0 & 94.4 &  & 0.53 & 102.4 & 95.6 &  & 68.3* \\
         &  & 80 & C &  & 1.53 & 86.5 & 93.2 &  & 2.12 & 96.7 & 95.0 &  & 76.3 \\
         &  &  & U &  & 1.64 & 86.0 & 93.5 &  & 2.37 & 97.0 & 94.6 &  & 80.3 \\
         &  & 85 & C &  & 1.57 & 86.2 & 93.5 &  & 2.31 & 96.9 & 94.8 &  & 79.3 \\
         &  &  & U &  & 1.68 & 85.3 & 93.3 &  & 2.64 & 97.3 & 94.9 &  & 84.3 \\
         &  & 90 & C &  & 1.71 & 85.4 & 93.2 &  & 2.61 & 97.3 & 94.6 &  & 83.1 \\
         &  &  & U &  & 1.57 & 84.9 & 93.2 &  & 2.80 & 98.1 & 94.9 &  & 87.2 \\
         &  & 95 & C &  & 1.46 & 84.9 & 93.2 &  & 2.81 & 97.9 & 94.8 &  & 86.8 \\
         &  &  & U &  & 0.10 & 85.8 & 93.6 &  & 3.16 & 98.2 & 95.0 &  & 91.0 \\ \hline
        60 & 60 & --- & --- &  & 1.72 & 102.7 & 95.5 &  & 0.88 & 102.0 & 95.4 &  & 76.5* \\
         &  & 80 & C &  & 3.98 & 90.5 & 94.9 &  & 1.19 & 104.3 & 96.3 &  & 80.5 \\
         &  &  & U &  & 4.18 & 90.0 & 94.6 &  & 1.49 & 104.8 & 96.4 &  & 84.6 \\
         &  & 85 & C &  & 4.14 & 90.2 & 94.9 &  & 1.40 & 104.7 & 96.5 &  & 83.7 \\
         &  &  & U &  & 4.05 & 89.2 & 94.6 &  & 1.84 & 105.6 & 96.5 &  & 87.5 \\
         &  & 90 & C &  & 4.22 & 89.4 & 94.7 &  & 1.71 & 105.4 & 96.6 &  & 86.8 \\
         &  &  & U &  & 3.09 & 88.3 & 94.7 &  & 2.36 & 106.0 & 96.5 &  & 90.7 \\
         &  & 95 & C &  & 3.25 & 88.4 & 94.6 &  & 2.27 & 106.0 & 96.5 &  & 90.4 \\
         &  &  & U &  & 1.90 & 88.2 & 94.6 &  & 3.01 & 105.7 & 96.6 &  & 94.0
        \end{tabular}
        \label{simulation 2b table}
    }
    \vspace{0.4cm}
    \scriptsize \\
    {
    \setlength{\baselineskip}{0.5\baselineskip}
        \raggedright{
            2000 simulated datasets.  
            Number of centers for each stage: $J$ = 4. 
            True coefficients: $(\beta^*_0, \beta^*_{11}, \beta^*_{12}) = (0.1, 0.3, 0.15)$. \\
            A fractional factorial design was used in stage 1 with intervention package components set to: (0,0), (1,0), (0,4), and (1,4), and true success probabilities = 0.525, 0.599, 0.668, and 0.731. 
            Cost function (linear): $C(\boldsymbol{x})=x_1 + 4x_2$. \\
            Ranges: for $x_1$, $[0,2]$, for $x_2$, $[0,8]$. 
            No unplanned variation in the interventions.\\
            \medskip
            $^*$: percent power calculated using the fractional factorial design.\\
            $n_j^{(1)}$: number of participants in each center $j$ at stage 1. $n_j^{(2)}$: number of participants in each center $j$ at stage 2.\\
            \%Power Goal: percent power goal of the two-sample z-test for proportions. For scenarios without the power goal, the \%Power Goal is set to ---.  \\
            power approach (U/C): U -- the unconditional power approach, C -- the conditional power approach. \\ 
            \%RelBias: percent relative bias $|100(\hat{\beta}-\beta^\star)/\beta^\star|$.\\
            SE: mean estimated standard error, 
            EMP.SD: empirical standard deviation.\\
            CP95: empirical coverage rate of the 95\% confidence intervals.\\
            \%Power: percent power of the two-sample z-test for proportions at the end of the LAGO trial.\\
        }}
\end{table}

\newpage 
\section{LAGO Design with More Than Two Stages}\label{K>2}
Section \ref{K>2} discusses the LAGO design with both an outcome goal and a power goal for the case of LAGO trials with more than 2 stages, $K>2$.

We first define the notation for stage $k$, where $k=2,\ldots,K$ and $K>2$.
Let $\hat{\boldsymbol{x}}_{j}^{\left(k, \bar{n}_{k-}\right)}$ be the recommended intervention package for center $j$ in stage $k$. 
The superscript in $\hat{\boldsymbol{x}}_{j}^{\left(k, \bar{n}_{k-}\right)}$ reminds us that $\hat{\boldsymbol{x}}_{j}^{\left(k, \bar{n}_{k-}\right)}$ is calculated using data from all participants from stages 1 to $k-1$.
Let $\boldsymbol{A}_{j}^{\left(k, \bar{n}_{k-}\right)}$ be the actual intervention package for center $j$ in stage $k$, where 
$\boldsymbol{A}_{j}^{\left(k, \bar{n}_{k-}\right)} = \hat{\boldsymbol{x}}_{j}^{\left(k, \bar{n}_{k-}\right)}$ in settings without unplanned variation in the intervention.
In settings with unplanned variation in the intervention, $\boldsymbol{A}_{j}^{\left(k, \bar{n}_{k-}\right)} =  h_j^{(k)} \bigl(\hat{\boldsymbol{x}}_{j}^{\left(k, \bar{n}_{k-}\right)}\bigr)$, and $h_j^{(k)}$, as for LAGO trials with 2 stages, is a continuous deterministic function for center $j$ in stage $k$. 

Let $n_j^{(k)}$ be the number of participants in each center $j$ of stage $k$, where $j=1,\ldots,J^{(k)}$.
Let $\boldsymbol{Y}_{j}^{\left(k, \bar{n}_{k-}\right)}= \bigl(Y_{1 j}^{\left(k, \bar{n}_{k-}\right)}, \ldots, Y_{n^{(k)}_j j}^{\left(k, \bar{n}_{k-}\right)}\bigr) $ be the outcomes for center $j$ in stage $k$. 
Let $\overline{\boldsymbol{A}}^{\left(k, \bar{n}_{k-}\right)}=\bigl(\boldsymbol{A}_{1}^{\left(k, \bar{n}_{k-}\right)}, \ldots, \boldsymbol{A}_{J^{(k)}}^{\left(k, \bar{n}_{k-}\right)}\bigr)$, 
$\overline{{\boldsymbol{x}}}^{\left(k, \bar{n}_{k-}\right)}=\bigl(\hat{{\boldsymbol{x}}}_{1}^{\left(k, \bar{n}_{k-}\right)}, \ldots, \hat{\boldsymbol{x}}_{J^{(k)}}^{\left(k, \bar{n}_{k-}\right)}\bigr)$,
and $\overline{\boldsymbol{Y}}^{\left(k, \bar{n}_{k-}\right)}=\bigl(\boldsymbol{Y}_{1}^{\left(k, \bar{n}_{k-}\right)}, \ldots, \boldsymbol{Y}_{J^{(k)}}^{\left(k, \bar{n}_{k-}\right)}\bigr)$ be the actual interventions, the recommended interventions and outcomes for all $J^{(k)}$ centers in stage $k$, respectively. 
Let $\widetilde{\boldsymbol{x}}^{\left(k, \bar{n}_{k-}\right)}=\bigl(\overline{\boldsymbol{x}}^{(1)}, \overline{\boldsymbol{x}}^{\left(2, n^{(1)}\right)}, \ldots, \overline{\boldsymbol{x}}^{\left(k, \bar{n}_{k-}\right)}\bigr)$, 
$\widetilde{\boldsymbol{A}}^{\left(k, \bar{n}_{k-}\right)}=\bigl(\overline{\boldsymbol{a}}^{(1)}, \overline{\boldsymbol{A}}^{\left(2, n^{(1)}\right)}, \ldots, \overline{\boldsymbol{A}}^{\left(k, \bar{n}_{k-}\right)}\bigr)$, and
$\widetilde{\boldsymbol{Y}}^{\left(k, \bar{n}_{k-}\right)}=\bigl(\overline{\boldsymbol{Y}}^{(1)}, \ldots, \overline{\boldsymbol{Y}}^{\left(k, \bar{n}_{k-}\right)}\bigr)$ be the recommended intervention packages, actual intervention packages, and outcomes from stages $1$ to $k$, respectively. 

Next, we show the modified assumptions for settings with $K>2$ stages.

\begin{assu}\label{conditional indep assumption}
    Given $\overline{{\boldsymbol{x}}}^{\left(k, \bar{n}_{k-}\right)},\bigl(\overline{\boldsymbol{A}}^{\left(k, \bar{n}_{k-}\right)}, \overline{\boldsymbol{Y}}^{\left(k, \bar{n}_{k-}\right)}\bigr)$ are independent of the data from the previous stages, $\bigl(\widetilde{\boldsymbol{A}}^{\left(k-1, \bar{n}_{(k-1)-}\right)}, \widetilde{\boldsymbol{Y}}^{\left(k-1, \bar{n}_{(k-1)-}\right)}\bigr)$.
\end{assu}
\noindent Assumption \ref{conditional indep assumption} states that learning from data collected so-far takes place only through the determination of the recommended interventions.

\begin{assu}\label{K>2 intervention_cvg_assumption}
    For non center-specific stage $k$ recommended interventions, $\hat{\boldsymbol{x}}^{\left(k, \bar{n}_{k-}\right)} \xrightarrow{P} {\boldsymbol{x}}^{\left(k\right)}$, where ${\boldsymbol{x}}^{\left(k\right)}$ is the limit of the recommended intervention.
    For center-specific stage $k$ interventions, for each $j=\{1, \ldots, J^{(k)}\}$, $\hat{\boldsymbol{x}}_j^{\left(k, \bar{n}_{k-}\right)} \xrightarrow{P} {\boldsymbol{x}}_j^{\left(k\right)}$.
\end{assu}
\noindent Following Section 4 of the main text, Assumption \ref{K>2 intervention_cvg_assumption} holds when the stage $k$ recommended intervention is obtained by solving an optimization problem with both an outcome goal and a power goal. 

As in \citet{nevo2021analysis, bing2023learnasyougo}, under Assumption \ref{K>2 intervention_cvg_assumption}, the definition of the $h_j^{(k)}$ and the continuous mapping theorem imply that $\boldsymbol{A}_j^{\left(k, \bar{n}_{k-}\right)}$ converge in probability to their center-specific limits $\boldsymbol{a}_j^{(k)}$.

We show that Theorems 1 and 2 from the main text hold for settings with $K>2$ stages. 
First, notice that given the cost function $C(\boldsymbol{x})$, to satisfy the outcome goal, the stage $k$ recommended intervention $\hat{\boldsymbol{x}}^{\left(k, \bar{n}_{k-}\right)}$ is obtained by solving
\begin{equation}\label{binary outcome goal simplified supp}
    \operatorname{expit}\left(\hat{\beta}_0^{(k-1)}+\left(\hat{\boldsymbol{\beta}}_1^{(k-1)}\right)^T \hat{\boldsymbol{x}}^{\left(k, \bar{n}_{k-}\right)} \right)
    \geq \tilde{p}.
\end{equation}
To discuss the calculation of the recommended intervention $\hat{\boldsymbol{x}}^{\left(k, \bar{n}_{k-}\right)}$ that satisfies the power goal, it is assumed that $\hat{\boldsymbol{x}}^{\left(k, \bar{n}_{k-}\right)}$ will be implemented in stage $k$ and onwards, both when including an unconditional power goal (Theorem \ref{unconditional power theorem stage k}) and when including a conditional power goal (Theorem \ref{conditional power goal theorem K>2}).

\begin{thm}{Two-sample z-test with unpooled variance, difference between two proportions, unconditional power goal, number of stages $K>2$.}\label{unconditional power theorem stage k}
    \textcolor{white}{xxx}\\
    Let $\chi^2_{\alpha, 1}$ be the upper $\alpha$ quantile of the central $\chi^2$ distribution with 1 degree of freedom. 
    For $\alpha=0.05$, $\chi^2_{\alpha, 1}=3.84$. 
    Let ${\lambda}_{min}$ be the minimum value of the non-centrality parameter for the non-central $\chi^2$ distribution with 1 degree of freedom, so that for a variable $T$ from a non-central $\chi^2$ distribution with non-centrality parameter ${\lambda}_{min}$, the probability of $T$ exceeding $\chi^2_{\alpha, 1}$ equals $\Pi$. 
    Let
    \begin{equation}
        \hat{S}_1^{(k+)}\left(\hat{\boldsymbol{x}}^{(k,\bar{n}_{k-})}, \hat{\boldsymbol{\beta}}^{(k-1)}\right) = \left(\sum_{l=k}^{K} n_1^{(l)} \right) \operatorname{expit}\left(\hat{\beta}_0^{(k-1)}+\left(\hat{\boldsymbol{\beta}}_1^{(k-1)}\right)^T\hat{\boldsymbol{x}}^{\left(k, \bar{n}_{k-}\right)}\right),  
        \label{mu_1_2_stage_k}
    \end{equation}
    \begin{equation}
        \hat{S}_0^{(k+)}\left( \hat{\boldsymbol{\beta}}^{(k-1)} \right) = \left(\sum_{l=k}^{K} n_0^{(l)} \right) \operatorname{expit}\left( \hat{\beta}_0^{(k-1)}\right),
        \label{mu_0_2_stage_k}
    \end{equation}
    \begin{equation*}
        VAR_{1} = \frac{\sum_{l=1}^{k-1} S_1^{(l)}+\hat{S}_1^{(k+)}(\hat{\boldsymbol{x}}^{(k,\bar{n}_{k-})}, \hat{\boldsymbol{\beta}}^{(k-1)})}{N_1} \left(1-\frac{\sum_{l=1}^{k-1} S_1^{(l)}+\hat{S}_1^{(k+)}(\hat{\boldsymbol{x}}^{(k,\bar{n}_{k-})}, \hat{\boldsymbol{\beta}}^{(k-1)})}{N_1} \right) \frac{1}{N_1}, 
    \end{equation*}
    \begin{equation*}
        VAR_{2} = { \frac{{\sum_{l=1}^{k-1} S_0^{(l)}+\hat{S}_0^{(k+)}(\hat{\boldsymbol{\beta}}^{(k-1)})}}{N_0} \left(1-\frac{{\sum_{l=1}^{k-1} S_0^{(l)}+\hat{S}_0^{(k+)}(\hat{\boldsymbol{\beta}}^{(k-1)})}}{N_0}\right)} \frac{1}{N_0},
    \end{equation*}
    and
    \begin{equation*}
    \begin{aligned}
        &{\lambda}\left(\hat{\boldsymbol{x}}^{(k,\bar{n}_{k-})}; \hat{\boldsymbol{\beta}}^{(k-1)}\right) = \\
        &
        \frac{\left(\left.\left( \sum_{l=1}^{k-1} S_1^{(l)}+\hat{S}_1^{(k+)}\bigl(\hat{\boldsymbol{x}}^{(k,\bar{n}_{k-})}, \hat{\boldsymbol{\beta}}^{(k-1)}\bigl) \right) \right/{N_1} - \left.\left({\sum_{l=1}^{k-1} S_0^{(l)}+\hat{S}_0^{(k+)}}\bigl( \hat{\boldsymbol{\beta}}^{(k-1)} \bigl)\right)\right/{N_0} \right)^2}
        {VAR_{1} + VAR_{2}}. 
    \end{aligned}
    \end{equation*}
    Under Assumptions \ref{conditional indep assumption} and \ref{K>2 intervention_cvg_assumption}, and Assumptions 1, 5, 7 and 8 from the main text, the stage $k$ recommended intervention $\hat{\boldsymbol{x}}^{(k,\bar{n}_{k-})}$, subject to both an outcome goal and an unconditional power goal, solves the following optimization problem: 
    \begin{equation*}
        \text{Min}_{\boldsymbol{x}} C\left(\boldsymbol{x}\right)
        \; \text {subject to} \; 
        \operatorname{expit}\left(\hat{\beta}_0^{(k-1)}+\left(\hat{\boldsymbol{\beta}}_1^{(k-1)}\right)^T\hat{\boldsymbol{x}}^{\left(k, \bar{n}_{k-}\right)}\right)
        \geq \tilde{p}, \;\;\text{and} \;\;{\lambda}\left(\hat{\boldsymbol{x}}^{(k,\bar{n}_{k-})}; \hat{\boldsymbol{\beta}}^{(k-1)}\right)  \geq \lambda_{min}. 
    \end{equation*}
\end{thm}

\begin{proof}
\textcolor{white}{xxx}\\
The proof of Theorem \ref{unconditional power theorem stage k} can be obtained by adapting the proof of Theorem \ref{unconditional power theorem appendix} as follows. In the proof, the data collected in stage 1, as referenced in Theorem \ref{unconditional power theorem appendix}, are replaced by data collected from stages 1 through $k-1$. Similarly, the recommended intervention calculated for stage 2 in Theorem \ref{unconditional power theorem appendix} is replaced by the recommended intervention calculated for stages $k$ through $K$. 
\end{proof}

\begin{thm}{Two-sample z-test with unpooled variance, difference between two proportions, conditional power goal, number of stages $K>2$.}\label{conditional power goal theorem K>2}
    \textcolor{white}{xxx}\\
    Let 
    \begin{equation*}
    \Delta\left(\hat{\boldsymbol{x}}^{(k,\bar{n}_{k-})}, \hat{\boldsymbol{\beta}}^{(k-1)}\right) =
    \frac{ \hat{S}_1^{(k+)}\left(\hat{\boldsymbol{x}}^{(k,\bar{n}_{k-})}, \hat{\boldsymbol{\beta}}^{(k-1)}\right)}{N_1} 
    -\frac{\hat{S}_0^{(k+)}\left( \hat{\boldsymbol{\beta}}^{(k-1)} \right) }{N_0}, \\
    \end{equation*}
    where $\hat{S}_1^{(k+)}\left(\hat{\boldsymbol{x}}^{(k,\bar{n}_{k-})}, \hat{\boldsymbol{\beta}}^{(k-1)}\right)$ and $\hat{S}_0^{(k+)}\left( \hat{\boldsymbol{\beta}}^{(k-1)} \right)$ are defined in Theorem \ref{unconditional power theorem stage k}, equations (\ref{mu_1_2_stage_k}) and (\ref{mu_0_2_stage_k}), respectively.
    Let
    \begin{equation*}
    \begin{aligned}
    &\hat{\sigma}^2\left(\hat{\boldsymbol{x}}^{(k,\bar{n}_{k-})}, \hat{\boldsymbol{\beta}}^{(k-1)}\right) = 
    \frac{\sum_{l=k}^{K} n_0^{(l)} }{N_0^2} \operatorname{expit}\left( \hat{\beta}_0^{(k-1)}\right)\left(1-\operatorname{expit}\left( \hat{\beta}_0^{(k-1)}\right)\right) \\
    & +\frac{ \sum_{l=k}^{K} n_1^{(l)} }{N_1^2} \operatorname{expit}\left(\hat{\beta}_0^{(k-1)}+\left(\hat{\boldsymbol{\beta}}_1^{(k-1)}\right)^T\hat{\boldsymbol{x}}^{\left(k, \bar{n}_{k-}\right)}\right) \left(1-\operatorname{expit}\left(\hat{\beta}_0^{(k-1)}+\left(\hat{\boldsymbol{\beta}}_1^{(k-1)}\right)^T\hat{\boldsymbol{x}}^{\left(k, \bar{n}_{k-}\right)}\right)\right),
    \end{aligned}
    \end{equation*}
    and 
    \begin{equation*}
    \begin{aligned}
        &\hat{\sigma}^2_{uncond}\left(\hat{\boldsymbol{x}}^{(k,\bar{n}_{k-})}, \hat{\boldsymbol{\beta}}^{(k-1)}\right) = \\
        &\hspace{2cm} { \frac{\sum_{l=1}^{k-1} S_1^{(l)}+\hat{S}_1^{(k+)}(\hat{\boldsymbol{x}}^{(k,\bar{n}_{k-})}, \hat{\boldsymbol{\beta}}^{(k-1)})}{N_1} \left(1-\frac{\sum_{l=1}^{k-1} S_1^{(l)}+\hat{S}_1^{(k+)}(\hat{\boldsymbol{x}}^{(k,\bar{n}_{k-})}, \hat{\boldsymbol{\beta}}^{(k-1)})}{N_1} \right)} \frac{1}{N_1} \\
        &\hspace{2cm} +{ \frac{{\sum_{l=1}^{k-1} S_0^{(l)}+\hat{S}_0^{(k+)}(\hat{\boldsymbol{\beta}}^{(k-1)})}}{N_0} \left(1-\frac{{\sum_{l=1}^{k-1} S_0^{(l)}+\hat{S}_0^{(k+)}(\hat{\boldsymbol{\beta}}^{(k-1)})}}{N_0}\right)} \frac{1}{N_0}.
    \end{aligned}
    \end{equation*}
    
    Under Assumptions \ref{conditional indep assumption} and \ref{K>2 intervention_cvg_assumption}, and Assumptions 1, 5, 7 and 8 from the main text, the stage $k$ recommended intervention $\hat{\boldsymbol{x}}^{(k,\bar{n}_{k-})}$, subject to both an outcome goal and the conditional power goal, solves the following optimization problem: 
    \begin{equation*}
        \text{Min}_{\boldsymbol{x}} C\left(\boldsymbol{x}\right)
        \; \text {subject to} \; 
        \operatorname{expit}\left(\hat{\beta}_0^{(k-1)}+\left(\hat{\boldsymbol{\beta}}_1^{(k-1)}\right)^T\hat{\boldsymbol{x}}^{\left(k, \bar{n}_{k-}\right)}\right)
        \geq \tilde{p}, \;\text{and}
    \end{equation*}
    \small 
    \begin{equation*}
    \begin{aligned}
        & z_{\alpha/2}
        \sqrt{ \hat{\sigma}^2_{uncond}\left(\hat{\boldsymbol{x}}^{(k,\bar{n}_{k-})}, \hat{\boldsymbol{\beta}}^{(k-1)}\right)}
        - \frac{\sum_{l=1}^{k-1} S_1^{(l)}}{N_1} + \frac{\sum_{l=1}^{k-1} S_0^{(l)}}{N_0} \\ 
        &\hspace{2cm}- \Delta\left(\hat{\boldsymbol{x}}^{(k,\bar{n}_{k-})},\hat{\boldsymbol{\beta}}^{(k-1)}\right) 
         - z_{\Pi} \; \hat{\sigma}\left(\hat{\boldsymbol{x}}^{(k,\bar{n}_{k-})},\hat{\boldsymbol{\beta}}^{(k-1)}\right)
         \leq 0. \\ 
    \end{aligned}
    \end{equation*}
\end{thm}

Following the same idea as for the Theorem \ref{unconditional power theorem stage k}, the proof of Theorem \ref{conditional power goal theorem K>2} follows immediately from the proof of Theorem  \ref{conditional power goal theorem appendix} and is therefore omitted.

\section{Relationship between the Power Goal and the Outcome Goal}\label{power goal and outcome goal appendix}

Section \ref{power goal and outcome goal appendix} provides a detailed discussion of the relationship between the outcome goal and power goal. 
Results from Section 6 of the main text show that in settings with large per-center sample sizes and small number of centers, e.g., $n_j^{(1)}=n_j^{(2)}=100$ and $J=4$, the LAGO design with only an outcome goal already achieves high power at the end of the LAGO trial. In these settings, adding a power goal has minimal effect and does not alter the recommended intervention package compositions; the outcome goal dominates the power goal in these settings.

To identify which settings benefit from adding a power goal in addition to an outcome goal, we considered the power goal of the two-sample z-test for the difference between two proportions as a concrete example.

We computed the threshold of the outcome goal needed to dominate the power goal across various per-center sample sizes. Consider a setting similar to Scenario 1 of the Simulations.
The total number of centers in both the intervention and control groups was $J = 4$ for each stage.
The per-center sample sizes were $n_j^{(1)}=n_j^{(2)}=20, 30, 40, ..., 100$.
Stage 1 of the two-stage LAGO trial used a fractional factorial design. The two intervention components $x_1$ and $x_2$ were set to $(0,0)$ in the control group, and to $(1,0)$, $(0,4)$, $(1,4)$ in the intervention group.
The minimum and maximum values for $x_1$ and $x_2$ were $[\mathcal{L}_1, \mathcal{U}_1] = [0,2]$ and $[\mathcal{L}_2, \mathcal{U}_2] = [0,8]$, respectively. 
The model for the binary outcome was $\operatorname{logit}\left(\operatorname{pr}\left(Y_{i j}=1 \mid \boldsymbol{A}=\right. \left.\boldsymbol{a}, \boldsymbol{X}=\boldsymbol{x}; \boldsymbol{\beta}\right)\right)
= \beta_0+\boldsymbol{\beta}_1^T \boldsymbol{a}.$
The true coefficients were $(\beta^*_0, \beta^*_{11}, \beta^*_{12}) = (0.1, 0.3, 0.15)$.
The power goals considered were $\Pi = 0.80 \;\text{and}\; 0.90$.

Figure \ref{relationship between goals plot} shows the threshold of the outcome goal needed to dominate the power goal. The percentages above the lines represent the outcome goal as a percentage increase from the control value $\operatorname{expit}(\beta^*_0)=0.525$. For example, a value of 61.9\% indicates that the outcome goal needed is 1.619 times the control value (e.g., 0.525*1.619=0.8499). The first row displays results using the unconditional power approach, while the second row shows results using the conditional power approach.
The first row shows that with smaller sample sizes, e.g., $n_j^{(1)}=n_j^{(2)}=40$, adding a power goal of 0.8 (unconditional power approach) is beneficial when the outcome goal is below 35.4\% above the control value (i.e., below 1.354 times control). Similarly, adding a power goal of 0.9 (unconditional power approach) is beneficial when the outcome goal is below 46.5\% above the control value. For larger sample sizes, e.g., $n_j^{(1)}=n_j^{(2)}=100$, adding a power goal is beneficial when the relative outcome goal is below 10\% above the control value. 
The second row exhibits a pattern similar to the first row. 

In summary, the power goal is particularly useful with small sample sizes and when the pre-specified outcome goal is close to the control outcome value. The definition of ``close" varies substantially across different implementation science research areas, and setting an appropriate outcome goal before the trial can be challenging. In settings with large sample sizes, the outcome goal typically provides sufficient power at the end of the LAGO trial. 
In conclusion, it can be expected that for the LAGO design with an outcome goal, the power goal functions as an additional safety net that helps prevent trial failure, especially when sample sizes are small, or when the outcome goal turns out to be close to the control probability/mean outcome.

\begin{figure}[hpt]
    \centering
    \includegraphics[width=1\linewidth]{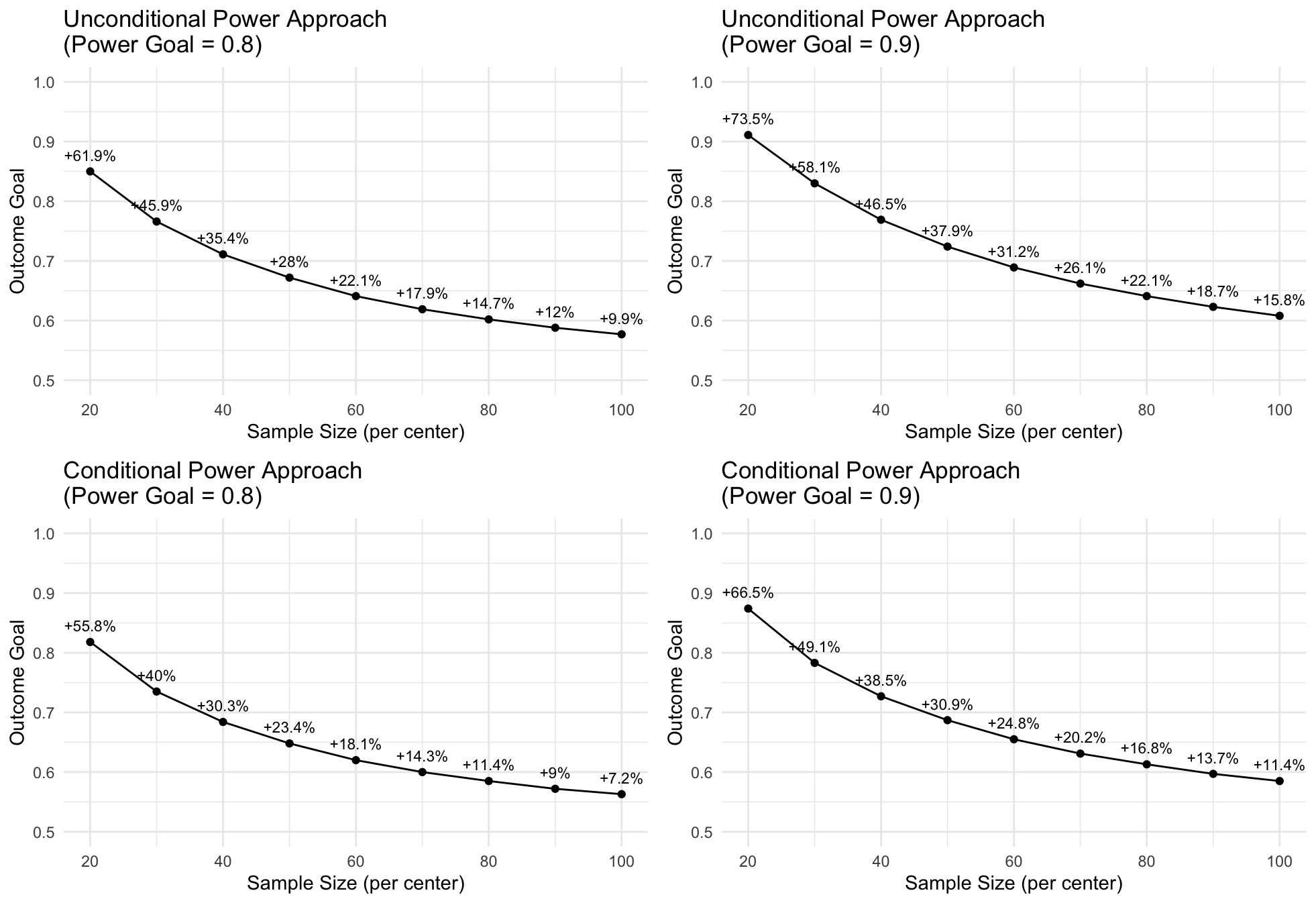}
    \caption{Threshold of the outcome goal needed to dominate the power goal at various per-center sample sizes, in setting of Section \ref{power goal and outcome goal appendix}.}
    \label{relationship between goals plot}
\end{figure}

\newpage 
\section{Method for Empirically Verifying Assumption 7 of Convergence of the Recommended Intervention with Non-linear Cost Functions}

To empirically verify Assumption~7 from the main text, which states that equation~(1) from the main text has a unique solution $\boldsymbol{x}$ in an open neighborhood of $\boldsymbol{\beta}^*$ with continuous dependence on $\boldsymbol{\beta}$, the following method can be used. 
This method is particularly useful for non-linear cost functions where theoretical guarantees may be difficult to achieve. We first introduce the method then describe an extension. 

\bigskip 

First, the outcome model is fit to the observed data to obtain an estimate $\hat{\boldsymbol{\beta}}$. Then, equation~(1) from the main text is solved using $\boldsymbol{\beta}$ in place of $\hat{\boldsymbol{\beta}}$, yielding a solution $\hat{\boldsymbol{x}}$. 
Next, an open neighborhood is constructed as $\mathcal{N} = \{\boldsymbol{\beta} : \|\boldsymbol{\beta} - \hat{\boldsymbol{\beta}}\|_2 < \epsilon\}$, where $\epsilon$ is predetermined by the trial statisticians. Within this neighborhood, $L$ parameter values $\{\boldsymbol{\beta}_{1}, \ldots, \boldsymbol{\beta}_{L}\}$ are randomly selected.
For each $\boldsymbol{\beta}_{l}$, the corresponding solution to equation (1) from the main text, denoted  $\boldsymbol{x}_{l}$ is computed.
Lastly, the maximum difference between $\boldsymbol{x}_{l}$ and $\hat{\boldsymbol{x}}$ is calculated as
\[
\delta^{\text{max}} = \max_{1 \leq l \leq L} \|\boldsymbol{x}_{l} - \hat{\boldsymbol{x}}\|.
\]
If $\delta^{\text{max}}$ is below a predefined tolerance threshold $\eta$, it suggests that Assumption~7 from the main text holds.

\bigskip 
The method described above can be extended to check the assumption across a wider range of likely parameter values, not just those immediately around $\hat{\boldsymbol{\beta}}$.

\bigskip 
First, using the 95\% confidence intervals of $\hat{\boldsymbol{\beta}}$, $M$ candidate vectors $\{\boldsymbol{\beta}_m\}_{m=1}^M$ are selected based on quantiles within the confidence intervals. 
For each $\boldsymbol{\beta}_m$, equation~(1) from the main text is solved using $\boldsymbol{\beta}$ in place of $\boldsymbol{\beta}_m$, yielding a solution ${\boldsymbol{x}}_m$.
Next, for each $\boldsymbol{\beta}_m$, an open neighborhood is constructed as $\mathcal{N}_m = \{\boldsymbol{\beta} : \|\boldsymbol{\beta} - \boldsymbol{\beta}_m\|_2 < \epsilon\}$. Within each neighborhood $\mathcal{N}_m$, $L$ parameter vectors $\{\boldsymbol{\beta}_{m,1}, \ldots, \boldsymbol{\beta}_{m,L}\}$ are randomly selected, and their corresponding solutions to equation (1) from the main text, $\boldsymbol{x}_{m,l}$, are computed.
Lastly, the maximum difference between $\boldsymbol{x}_{m,l}$ and ${\boldsymbol{x}}_m$ is calculated as
\[
\delta_m^{\text{max}} = \max_{1 \leq l \leq L} \|\boldsymbol{x}_{m,l} - \boldsymbol{x}_m\|.
\]
If $\delta_m^{\text{max}}$ remains below the tolerance threshold $\eta$ for all $m$, this provides stronger evidence that Assumption~7 from the main text holds.

\end{document}